\documentclass[journal, draftcls, onecolumn, 12pt]{IEEEtranTCOM}
\normalsize

\usepackage{cite}
\usepackage{bm}
\usepackage{amsmath}
\usepackage{amsthm}
\newtheorem{defn}{Definition} % definition numbers are dependent on theorem numbers
\usepackage{amssymb}
\newtheorem{theorem}{Theorem}
\usepackage{caption}
\usepackage{subcaption}
\usepackage{graphicx}
\usepackage{dcolumn}
\usepackage{algorithmic}
\newtheorem{proposition}{Proposition}

\theoremstyle{remark}
\newtheorem*{remark}{Remark}

\theoremstyle{example}
\newtheorem{example}{Example}

\theoremstyle{example_nc}
\newtheorem*{example_nc}{Example}

\begin{document}
%

% paper title
% can use linebreaks \\ within to get better formatting as desired
\title{Minimax Optimum Clock Skew and Offset Estimators for IEEE 1588}
%
%
% author names and IEEE memberships
% note positions of commas and nonbreaking spaces ( ~ ) LaTeX will not break
% a structure at a ~ so this keeps an author's name from being broken across
% two lines.
% use \thanks{} to gain access to the first footnote area
% a separate \thanks must be used for each paragraph as LaTeX2e's \thanks
% was not built to handle multiple paragraphs
%

\author{ Anantha K. Karthik, \IEEEmembership{Student Member, IEEE} and Rick S. Blum, \IEEEmembership{IEEE Fellow} % <-this % stops a space
	
\thanks{This work was supported by the Department of Energy under Award DE-OE0000779, and by the National Science Foundation under Grant ECCS 1744129.}% <-this % stops a space
	
\thanks{Anantha K. Karthik and Rick S. Blum are with the Department of Electrical and Computer Engineering, Lehigh University, Bethlehem, PA 18015 USA (e-mail: akk314@lehigh.edu; rblum@lehigh.edu).}}% <-this % stops a space

% note the % following the last \IEEEmembership and also \thanks - 
% these prevent an unwanted space from occurring between the last author name
% and the end of the author line. i.e., if you had this:
% 
% \author{....lastname \thanks{...} \thanks{...} }
%                     ^------------^------------^----Do not want these spaces!
%
% a space would be appended to the last name and could cause every name on that
% line to be shifted left slightly. This is one of those "LaTeX things". For
% instance, "\textbf{A} \textbf{B}" will typeset as "A B" not "AB". To get
% "AB" then you have to do: "\textbf{A}\textbf{B}"
% \thanks is no different in this regard, so shield the last } of each \thanks
% that ends a line with a % and do not let a space in before the next \thanks.
% Spaces after \IEEEmembership other than the last one are OK (and needed) as
% you are supposed to have spaces between the names. For what it is worth,
% this is a minor point as most people would not even notice if the said evil
% space somehow managed to creep in.

% The paper headers
\markboth{IEEE Transactions on Communications}%
{Submitted paper}
% The only time the second header will appear is for the odd numbered pages
% after the title page when using the twoside option.
% 
% *** Note that you probably will NOT want to include the author's ***
% *** name in the headers of peer review papers.                   ***
% You can use \ifCLASSOPTIONpeerreview for conditional compilation here if
% you desire.

% If you want to put a publisher's ID mark on the page you can do it like
% this:
%\IEEEpubid{0000--0000/00\$00.00~\copyright~2007 IEEE}
% Remember, if you use this you must call \IEEEpubidadjcol in the second
% column for its text to clear the IEEEpubid mark.

% make the title area
\maketitle

\begin{abstract}
%\boldmath
This paper addresses the problem of clock skew and offset estimation for the IEEE 1588 precision time protocol. Built on the classical two-way message exchange scheme, IEEE 1588 is a prominent synchronization protocol for packet switched networks. It is employed in various applications including cellular base station synchronization in 4G long-term evaluation backhaul networks, substation synchronization in electrical grid networks and industrial control. Due to the presence of random queuing delays in a packet switched network, the recovery of clock skew and offset from the received packet timestamps can be viewed as a statistical estimation problem. Recently, assuming perfect clock skew information, minimax optimum clock offset estimators were developed for IEEE 1588. Building on this work, we develop minimax optimum clock skew and offset estimators for IEEE 1588 in this paper. Simulation results indicate the proposed minimax estimators exhibit a lower mean square estimation error than the estimators available in the literature for various network scenarios.
\end{abstract}
% IEEEtran.cls defaults to using nonbold math in the Abstract.

% This preserves the distinction between vectors and scalars. However,
% if the journal you are submitting to favors bold math in the abstract,
% then you can use LaTeX's standard command \boldmath at the very start
% of the abstract to achieve this. Many IEEE journals frown on math
% in the abstract anyway.

% Note that keywords are not normally used for peerreview papers.

% For peer review papers, you can put extra information on the cover
% page as needed:
% \ifCLASSOPTIONpeerreview
% \begin{center} \bfseries EDICS Category: 3-BBND \end{center}
% \fi
%
% For peerreview papers, this IEEEtran command inserts a page break and
% creates the second title. It will be ignored for other modes.
\IEEEpeerreviewmaketitle

\section{Introduction}\label{Sec1}
Precise synchronization of events is essential to ensure the proper functioning of a distributed network. The IEEE 1588 Precision Time Protocol (PTP) \cite{IEEE1588} is a popular time synchronization protocol for synchronizing the slave clocks to a master clock. It is cost effective and offers accuracy comparable to  Global Positioning System (GPS)-based timing. PTP is utilized in various  applications including electrical grid networks \cite{Gaderer_2005}, cellular base station synchronization in 4G Long Term Evaluation (LTE) \cite{Hadzic_2011}, substation communication networks \cite{IEC61850} and industrial control \cite{IEEE_Industrial_Control}. In this paper, we will develop clock synchronization algorithms for PTP in a packet switched network.

%PTP is used in conjunction with Synchronous Ethernet (SyncE) for synchronization in LTE backhaul networks. Although the SyncE standards are now mature, much of the deployed base of Ethernet equipment does not support it \cite{IEEE1588v2_whitepaper}. 
%%If a single Ethernet switch in the chain does not support SyncE, all nodes lower in the hierarchy do not receive the timing service \cite{IEEE1588v2_whitepaper}. 
%PTP is the primary option for synchronization to operators with packet backhaul networks that do not support SyncE \cite{IEEE1588v2_whitepaper, IEEE1588v2_Huawei}. In this paper, we will develop synchronization algorithms for PTP when used in LTE backhaul networks. Although we evaluate the performance of the proposed PTP synchronization algorithms in LTE backhaul networks, they apply to other non-LTE network scenarios where PTP is used. An example of such a scenario is when PTP-based timing is used as an alternative to GPS-based timing in the electrical grid for scenarios when GPS signals are not available at the substation \cite{GPSattacks, GPSattacks1, NASPI}.

The clock at the slave node can be modeled mathematically, as a function $c(t)$ of the time of the master node's clock $t$. When the clocks of the slave and master node are synchronized, then $c(t) = t$. However, in practice these clocks are not synchronized, implying a synchronization error $e(t) = |c(t) - t|$, that tends to grow over large time scales unless synchronization approaches are implemented. In general, the clock of the slave node is modeled as  $c(t) = \phi t + \delta$ \cite{Noh_2007, Chaudhari_2008, Leng_2010, Wu_2011, Anand_bounds, Anand_2015}, where $\phi$ and $\delta$ denote the relative clock skew and offset of the slave's clock with respect to the master's clock respectively.

%% Paragraph-3
A number of time synchronization protocols including PTP, Timing Protocol for Sensor Networks (TPSN) \cite{TPSN}, tiny-sync \cite{MiniSync}, and Lightweight Time Synchronization (LTS) \cite{LTS} are built on the classical two-way message exchange scheme. In these protocols, the slave node exchanges a series of synchronization packets with the master node and uses the packet timestamps to estimate $\phi$ and $\delta$. The messages traveling between the master and slave node can encounter several intermediate switches and routers accumulating delays at each node. The main factors contributing to the overall delay are the fixed propagation and processing delays at the intermediate nodes along the network path between the master and slave node and the random queuing delays at each such node. This randomness in the overall network traversal time is referred to as Packet Delay Variation (PDV) \cite{Anand_2015}, and the problem of estimating $\phi$ and $\delta$, while combating the noisiness in the observations that occur due to PDV is called the ``Clock Skew and Offset Estimation" (CSOE) problem.

%% Paragraph-4
Popular Probability Density Function (pdf) models available in the literature to model the PDV include Gaussian, exponential, gamma, Weibull, and log-normal \cite{Wu_2011}. The Cramer-Rao lower Bound (CRB) and the Maximum Likelihood (ML) estimate of the clock skew and offset for some of these PDV delay models were derived in \cite{Noh_2007, Leng_2010, Chaudhari_2008}. The popular PDV delay models, however, seem unsuitable for general packet switched networks. For example, consider the scenario where PTP is used to synchronize the cellular base stations in 4G LTE networks using mobile backhaul networks. The backhaul networks are leased from commercial Internet Service Providers (ISPs), and the network is shared with other commercial and non-commercial users. The background traffic generated by these users often results in random delays for the synchronization packets. In the context of the backhaul networks, ITU-T G.8261 specification \cite{ITU} provides models for modeling the background traffic. The empirical pdf of the PDV in the backhaul networks were obtained in \cite{Anand_bounds}, and are shown in Figure \ref{Empirical_pdf}. Similar random delays can occur in any case where a shared network is utilized. The popular available pdf models do not closely match most of the cases in Figure \ref{Empirical_pdf}.

%% Paragraph-6
Furthermore, the PDV pdf typically has limited support in a packet switched network. Hence, the CRB (the most popular lower bound in estimation theory) is not suitable for evaluating the performance of a CSOE scheme in these networks as the regularity conditions are violated \cite{Anand_bounds}. Guruswamy \emph{et al.} \cite{Anand_bounds} addressed this issue and developed performance lower bounds for an invariant clock offset estimation scheme for PTP assuming knowledge of the clock skew. Building on their previous work of \cite{Anand_bounds}, Guruswamy \emph{et al.} \cite{Anand_2015} developed minimax optimum clock offset estimation schemes for PTP under the squared error loss function.

%% Paragraph-7
Following the work of \cite{Anand_2015}, we will for the first time, develop minimax optimum CSOE schemes for PTP in this paper. The problem of estimating the clock skew and offset in the presence of PDV falls under a variant of the location-scale parameter problems \cite{Berger}, with the unknown clock skew as the scale parameter and the unknown clock offset as the location parameter. Fixing the loss function to skew-normalized squared error loss, we use invariant decision theory (see chapter 6 of \cite{Berger}) to design the optimum invariant CSOE scheme. Then, using results from \cite{Berger, Lehmann, Bondar1981, marchand2012 }, we show the developed optimum invariant CSOE schemes are minimax optimum under the skew-normalized squared error loss function. Simulation results indicate the minimax optimum estimators exhibit a lower mean square estimation error than the estimators available in the literature in a variety of network scenarios, and theoretical results which prove the minimax optimum schemes must be as good or better are also provided.

\textbf{Notations:} We use bold upper case, bold lower case, and italic lettering to denote matrices, column vectors and scalars respectively. The notations $(.)^T$ and $\otimes$ denote the transpose and Kronecker product respectively. $\bm{I}_N$ stands for a $N$-dimensional identity matrix and $\bm{1}_{N}$ denotes a column vector of length $N$ with all the elements equal to $1$. Further, $\mathbb{R}$ denotes the set of real numbers, $\mathbb{R}^+$ denotes the set of positive real numbers, $\mathbb{R}_0^+$ denotes the set of non-negative real numbers and $\mathcal{I}_A(x)$ denotes the indicator function having the value $1$ when $x \in A$ and $0$ when $x \notin A$. 

\section{Signal Model and Problem Statement}\label{Sec2}
In this section, we briefly describe the two-way message exchange scheme used in IEEE 1588 and present the considered problem statement. Recall that the relative clock skew and offset of the slave node with respect to the master node are denoted by $\phi \in \mathbb{R}^+$ and $\delta \in \mathbb{R}$, respectively. Assuming $P$ rounds of two-way message exchanges, the following sequence of messages are exchanged between the master and slave node during the $i^{th}$ round of message exchanges ($i = 1, 2, \cdots, P$). For each $i$, the master node initiates a two-way message exchange by sending a \emph{sync} packet to the slave at time $t_{1i}$.  The value of $t_{1i}$ is later communicated to the slave via a \emph{follow\_up} message. The slave node records the time of reception of the \emph{sync} message as $t_{2i}$. The slave node sends a \emph{delay\_req} message to the master node while recording the time of transmission as $t_{3i}$. The master records the time of arrival of the \emph{delay\_req} packet at time $t_{4i}$ and this value is later communicated to the slave using a \emph{delay\_resp} packet. This procedure can be mathematically modeled as \cite{Noh_2007, Chaudhari_2008, Leng_2010, Wu_2011}
\begin{eqnarray}
t_{2i} & = & (t_{1i} + d_{ms} + w_{1i})\phi + \delta,  \label{FwdData} \\
t_{3i} & = & (t_{4i} - d_{sm} - w_{2i})\phi + \delta   \label{RevData}
\end{eqnarray}
for $i = 1, 2, \cdots, P$. In (\ref{FwdData}) and (\ref{RevData}), $d_{ms}$ and $d_{sm}$ denote the fixed propagation delays in the master-to-slave forward path and slave-to-master reverse path respectively. The variables $w_{1i}$ and $w_{2i}$ denote the random queuing delays in the forward and reverse path respectively. Define $\bm{w}_k = [w_{k1}, w_{k2}, \cdots, w_{kP}]$ for $k = 1, 2$ and $\bm{t}_k = [t_{k1}, t_{k2}, \cdots, t_{kP}]$ for $k = 1, 2, 3, 4$. The joint pdf of $\bm{w}_k$ is defined as $f_{\bm{w}_k}(\bm{w}_k) = f_{k}(w_{k1}, w_{k2}, \cdots, w_{kP})$, for $k = 1, 2$. In our work, we assume the queuing delays in the forward and reverse path are independent. Following \cite{Chaudhari_2008, Anand_2015}, we consider two observation models based on the amount of information available regarding the fixed path delays $d_{ms}$ and $d_{sm}$:

\subsubsection{Known fixed delay model (K-Model)}
In this model, we assume complete knowledge of the fixed-path delays $d_{ms}$ and $d_{sm}$. The received timestamps can be arranged in vector form as follows
\begin{eqnarray}\label{Kmodel}
\bm{y} & = & \bm{u}\phi + \delta\bm{1}_{2P},
\end{eqnarray}
where, from (\ref{FwdData}) and (\ref{RevData}), we have $\bm{y} = [\bm{t}_2, \bm{t}_3]^T$,  and $\bm{u}  = [\bm{u}_1, \bm{u}_2]^T$
%\begin{eqnarray}
%\bm{y} = [\bm{t}_2, \bm{t}_3]^T,  & \mbox{ and } & \bm{u}  = [\bm{u}_1, \bm{u}_2]^T, \label{Eq1} 
%\end{eqnarray}
with $\bm{u}_1 = (\bm{t}_1 + \bm{w}_1 + d_{ms}\bm{1}_P^T)$ and $\bm{u}_2 = (\bm{t}_4 - \bm{w}_2 - d_{sm}\bm{1}_P^T)$.  The unknown parameters in this model are $\phi$ and $\delta$.
%, and the joint pdf of $\bm{u}$ is given by (assuming the forward path queuing delay samples are independent of the reverse path queuing delay samples)
%\begin{eqnarray}\label{K_model_pdf}
%f_{\bm{u}}(\bm{u}) & = & \prod_{k=1}^{2} f_{\bm{u}_k}(\bm{u}_k) =  f_{\bm{w}_1}(\bm{u}_{1} - \bm{t}_{1} - d_{ms}\bm{1}_P^T) f_{\bm{w}_2}(\bm{t}_{4} - \bm{u}_{2} - d_{sm}\bm{1}_P^T).
%\end{eqnarray}

\subsubsection{Standard model (S-Model)}
Freris \emph{et al.} \cite{Freris} provided some necessary conditions for obtaining a unique solution for the system of equations given in (\ref{FwdData}) and (\ref{RevData}). We need to know either one of the fixed path delays (either $d_{ms}$ or $d_{sm}$), or have a {prior} known affine relationship between the fixed delays (see Theorem 4 in \cite{Freris}). In this model, we assume a prior known affine relationship between the fixed path delays. For simplicity, we assume the fixed path delays are equal, i.e., $d_{ms} = d_{sm} = d$, where $d$ represents the unknown fixed path delay in the master-slave communication path\footnote{We should mention here that the proposed estimators are also applicable when there is a prior known affine relationship between $d_{ms}$ and $d_{sm}$, i.e., $d_{ms} = ad_{sm} + c$, where the constants $a$ and $c$ are known.}. The received time stamps can be arranged in vector form as follows
\begin{eqnarray}\label{Smodel}
\bm{y} & = & (\bm{h}d + \bm{v})\phi + \delta \bm{1}_{2P},
\end{eqnarray}
where $\bm{v} = [\bm{v}_1, \bm{v}_2]^T$ with $\bm{v}_1 = (\bm{t}_1 + \bm{w}_1)$ and $\bm{v}_2 = (\bm{t}_4 - \bm{w}_2)$, $\bm{h} = [\bm{1}_P^T, -\bm{1}_P^T]^T$, and $\bm{y} = [\bm{t}_2, \bm{t}_3]^T$. The unknown parameters in this model are $\phi$, $d$ and $\delta$.

%, and the joint pdf of $\bm{v}$ is given by
%\begin{eqnarray}\label{S_model_pdf}
%f_{\bm{v}}(\bm{v}) & = & \prod_{k=1}^{2} f_{\bm{v}_k}(\bm{v}_k) = f_{\bm{w}_1}(\bm{v}_{1} - \bm{t}_{1}) f_{\bm{w}_2}(\bm{t}_{4} - \bm{v}_{2}).
%\end{eqnarray}
{\bf Problem Statement:} In this paper, we look to develop CSOE schemes for estimating $\phi$ and $\delta$ from the received timestamps for the considered observation models.

\section{Statistical Preliminaries}\label{Sec3}
The purpose of this section is to formalize the concept of invariance by defining groups of transformations over parameter and observation spaces. To this end, we repeat several essential definitions from \cite{Berger} to establish some concepts of invariant estimation theory. It is assumed throughout this section that the observed data $\bm{x} \in \mathbb{R}^N$ is characterized by the pdf $f(\bm{x}| \bm{\theta})$, which depends upon the vector of unknown parameters $\bm{\theta}$ with the corresponding parameter space $\bm{\Theta}$.

Suppose we are interested in estimating an unknown scalar parameter $\theta \in \bm{\theta}$. Let ${\psi}$ denote an estimator of ${\theta}$, ${\psi}({\bm{x}})$ denote the estimate of $\theta$ obtained using the estimator ${\psi}$ on $\bm{x}$, and $L({\psi}({\bm{x}}), \bm{\theta})$ denote the considered loss function. The performance of the estimator ${\psi}$ can be characterized by the following \cite{Lehmann}:
\begin{enumerate}
	\item The conditional risk of an estimator 
	\begin{eqnarray}
	\mathcal{R}({\psi}, \bm{\theta}) & = & \int_{\mathbb{R}^N} L({\psi}({\bm{x}}), \bm{\theta}) f(\bm{x}| \bm{\theta}) d\bm{x},
	\end{eqnarray}
	
	\item The maximum risk of an estimator
	\begin{eqnarray}
	\mathcal{M}({\psi}) & = & \sup_{\bm{\theta} \in \bm{\Theta}} \mathcal{R}({\psi}, \bm{\theta}),
	\end{eqnarray}
	
	\item The average risk of an estimator
	\begin{eqnarray}
	\mathcal{B}({\psi}, p) & = & \int_{\bm{\theta} \in \bm{\Theta}} \mathcal{R}({\psi}, \bm{\theta}) p(\bm{\theta}) d\bm{\Theta},
	\end{eqnarray}
	where $p(\bm{\theta})$ is a prior distribution defined over $\bm{\theta} \in \bm{\Theta}$.
\end{enumerate}

\noindent
In our work, we are primarily interested in developing minimax estimators, that is estimators that minimize the maximum risk over all possible estimators of the parameter of interest. We first present the definition of a minimax estimator from \cite{Lehmann}.
\begin{defn}[Minimax estimators]
%	The minimax risk is defined as 
%	\begin{eqnarray}
%	\mathcal{R}_{MinMax}(\bm{\Theta}) & = & \inf_{{\psi}}\mathcal{M}({\psi}) = \inf_{{\psi}} \sup_{\bm{\theta} \in \bm{\Theta}} \mathcal{R}({\psi}, \bm{\theta}).
%	\end{eqnarray}
	An estimator $\bm{\psi}_{MinMax}$ of ${\theta} \in \bm{\theta}$ is said to be a minimax estimator of ${\theta}$ for the considered loss function, if
	\begin{eqnarray}
	\mathcal{M}({\psi}_{MinMax}) & = & \inf_{{\psi}}\mathcal{M}({\psi})  = \inf_{{\psi}} \sup_{\bm{\theta} \in \bm{\Theta}} \mathcal{R}({\psi}, \bm{\theta}).%\mathcal{R}_{MinMax}(\bm{\Theta}).
	\end{eqnarray}
\end{defn}
\noindent
We use the approach given in \cite{Berger} (see Chapter 5) to design a minimax estimator of $\theta$. We first construct the optimum invariant estimator of $\theta$ for a considered (invariant) loss function and then show the optimum invariant estimator is a minimax estimator of $\theta$ for the considered loss function. 

We now present some important definitions from \cite{Berger} with regards to invariant estimation theory. A measurable function $f : \mathbb{R}^N \rightarrow \mathbb{R}^N$ is called a transformation on $\mathbb{R}^N$. If $g_1$ and $g_2$ are two transformations on $\mathbb{R}^N$, the composition of $g_1$ and $g_2$, denoted by $g_2g_1$, is defined as $g_2g_1(\bm{m}) = g_2(g_1(\bm{m}))$ for $\bm{m} \in \mathbb{R}^N$. We are now ready to define a group of transformations.

\begin{defn}[Section 6.2.1, \cite{Berger}]
	A group of transformations on $\mathbb{R}^N$, denoted by $\mathcal{G}$, is a set of one-to-one and onto transformations which satisfy the following conditions:
	\begin{itemize}
		\item If $g_1 \in \mathcal{G}$ and $g_2 \in \mathcal{G}$, then $g_2g_1 \in \mathcal{G}$.
		\item If $g \in \mathcal{G}$, then $g^{-1}$, the inverse transformation defined by the relation $g^{-1}(g(\bm{x})) = \bm{x}$, is in $\mathcal{G}$.
		\item The identity transformation $e$, defined by $e(\bm{x}) = \bm{x}$, is in $\mathcal{G}$.
	\end{itemize}
\end{defn}

Let $\mathcal{F}$ denote the class of all densities $f(\bm{x}|\bm{\theta})$ for $\bm{\theta} \in \bm{\Theta}$ and $\mathcal{G}$ denote a group of transformations on $\mathbb{R}^N$.

\begin{defn}[Section 6.2.2, \cite{Berger}]
	The family of densities $\mathcal{F}$ is said to be invariant under $\mathcal{G}$, if for every $g \in \mathcal{G}$ and $\bm{\theta} \in \bm{\Theta}$, there exists a unique $\bm{\theta}^* \in \bm{\Theta}$ such that $\bm{x}_g = g(\bm{x})$ has density $f(\bm{x}_g|\bm{\theta}^*)$. We denote $\bm{\theta}^*$ as $\bar{g}(\bm{\theta})$.
\end{defn}

\begin{remark}
	If $\mathcal{F}$ is invariant under $\mathcal{G}$, then
	\begin{eqnarray}
	\bar{\mathcal{G}} & = & \{ \bar{g} : g \in \mathcal{G} \}
	\end{eqnarray}
	is a group of transformations on $\bm{\Theta}$ \cite{Berger}. We now present a simple example to illustrate these ideas.
\end{remark}

\begin{example}
Let $\bm{x} = [x_1, x_2, \cdots, x_N]$ and $h(.)$ be a known density. Consider the class of densities of the form
\begin{eqnarray}
f(\bm{x}|\bm{\theta}) & = & \frac{1}{\sigma^N}h\left( \frac{\bm{x} - \mu\bm{1}_N^T}{\sigma}  \right),
\end{eqnarray}
where $\mu \in \mathbb{R}$ (location parameter) and $\sigma \in \mathbb{R}^+$ (scale parameter) are both unknown. From Definition $3$, the class of such densities is invariant under the group of location-scale transformations (see Example 5, Section 6.2.1, \cite{Berger}) $\mathcal{G}_{affine}$, on $\mathbb{R}^N$, defined as
\begin{eqnarray}\label{Example_group}
\mathcal{G}_{affine}  & = & \{ g_{a,b}(\bm{m}) : g_{a, b}(\bm{m})  = a\bm{m} + b\bm{1}_N \} \mbox{ where } a \in \mathbb{R}^+, b \in \mathbb{R} \mbox{ and }\bm{m} \in \mathbb{R}^N,
\end{eqnarray} 
since $\bm{x}_g = g_{a, b}(\bm{x}) = [x_{g1}, x_{g2}, \cdots, x_{gN}]$ has the density $(a\sigma)^{-N} h\left( \frac{\bm{x}_g - (a\mu + b)\bm{1}_N^T}{a\sigma}  \right)$. The group, $\bar{\mathcal{G}}_{affine}$, of induced transformations on $\bm{\Theta} = \{ (\mu, \sigma): \mu \in \mathbb{R} \mbox{ and }\sigma \in \mathbb{R}^+ \}$, is given by
\begin{eqnarray}
\bar{\mathcal{G}}_{affine} & = & \{\bar{g}_{a, b}(\mu, \sigma) : \bar{g}_{a, b}(\mu, \sigma) = (a\mu + b, a\sigma)  \}.
\end{eqnarray}
\end{example}
\noindent 
To be invariant, an estimation problem must have a loss function which is unchanged by the relevant transformations. We now present the definition of an invariant loss function.

\begin{defn}[Section 6.2.2, \cite{Berger}]
Let $\mathcal{F}$ be invariant under the group $\mathcal{G}$ and $\hat{{\theta}}$ be an estimate of ${\theta}$. A loss function $L(\hat{{\theta}}, {\bm{\theta}} )$ is said to be invariant under $\mathcal{G}$, if for every $g \in \mathcal{G}$, there exists an $\hat{{\theta}}^*$ such that $L(\hat{{\theta}}, {\bm{\theta}} ) = L(\hat{{\theta}}^*, \bar{g}({\bm{\theta}}) )$ for all $\bm{\theta} \in \bm{\Theta}$. We denote $\hat{{\theta}}^*$ by $\tilde{g}({\theta})$.
\end{defn}

In an invariant estimation problem, the formal structures of the statistical distributions of $\bm{x}$ and $g(\bm{x})$ are identical. Hence the invariance principle states that the estimates obtained from $\bm{x}$ and $g(\bm{x})$, using an estimator must be related \cite{Berger}. We now present the definition of an invariant estimator.
\begin{defn}[Section 6.2.3, \cite{Berger}]
	Let ${\psi}$ denote an estimator of ${\theta} \in \bm{\theta}$, and ${\psi}(\bm{x})$ denote the estimate of ${\theta}$ obtained from the received observations $\bm{x}$ characterized by the pdf $f(\bm{x}|\bm{\theta})$. We say ${\psi}$ is invariant under group $\mathcal{G}$ if for all $\bm{x} \in \mathbb{R}^N$ and $g \in \mathcal{G}$,
	\begin{eqnarray}
	{\psi}(g(\bm{x})) & = & \tilde{g}({\psi}(\bm{x})).
	\end{eqnarray}
\end{defn}

\begin{example_nc}[Example 1 continued]
	Let $\hat{\mu}$ and $\hat{\sigma}$ denote estimators of $\mu$ and $\sigma$, respectively and let $\hat{\mu}(\bm{x})$ and $\hat{\sigma}(\bm{x})$ denote the estimates obtained from $\bm{x}$. From Definition 5, the estimators $\hat{\mu}$ and $\hat{\sigma}$ are invariant under $\mathcal{G}_{affine}$ defined in (\ref{Example_group}), if
	\begin{eqnarray}
	\hat{\mu}(g_{a, b}(\bm{x}))  = a\hat{\mu}(\bm{x}) + b,&  \mbox { and } &\hat{\sigma}(g_{a, b}(\bm{x})) = a\hat{\sigma}(\bm{x}). \label{invariant_location}
	\end{eqnarray}	
	for all $g_{a, b} \in \mathcal{G}_{affine}$. From Definition $4$, the loss functions for $\mu$ and $\sigma$, defined by
	\begin{eqnarray}\label{Lossfn_example1_location}
	L_{\mu}(\hat{\mu}(\bm{x}), [\mu, \sigma]) = \frac{(\mu - \hat{\mu}(\bm{x}))^2}{\sigma^2}, \mbox{ and }
	L_{\sigma}(\hat{\sigma}(\bm{x}), [\mu, \sigma]) = \frac{(\sigma - \hat{\sigma}(\bm{x}))^2}{\sigma^2},
	\end{eqnarray}
	respectively, are invariant under $\mathcal{G}_{affine}$ from (\ref{Example_group}), since
	\begin{eqnarray}
	\frac{(\hat{\mu}(\bm{x}) - \mu)^2}{\sigma^2} = \frac{\left(\hat{\mu}(g_{a, b}(\bm{x})) - (a\mu + b)\right)^2}{a^2\sigma^2}, \mbox{ and } 
	\frac{(\hat{\sigma}(\bm{x}) - \sigma)^2}{\sigma^2}  =  \frac{\left(\hat{\sigma}(g_{a, b}(\bm{x})) - a\sigma\right)^2}{a^2\sigma^2},
	\end{eqnarray}
	for all $g_{a, b} \in \mathcal{G}_{affine}$ from (\ref{Example_group}). The loss functions given in (\ref{Lossfn_example1_location}) are called the scale-normalized squared error loss.
\end{example_nc}

\noindent
We now present an important definition regarding the transitivity of the group of transformations on $\bm{\Theta}$ and the conditional risk of invariant estimators.

\begin{defn}[Section 6.2.3, \cite{Berger}]
	A group $\bar{\mathcal{G}}$ of transformations of $\bm{\Theta}$ is said to be transitive if for any $\bm{\theta}_1, \bm{\theta}_2 \in \bm{\Theta}$, there exists some $\bar{g} \in \bar{\mathcal{G}}$ for which $\bm{\theta}_2 = \bar{g}(\bm{\theta}_1)$.
\end{defn}

\begin{theorem}[Section 6.2.3, \cite{Berger}]
	When $\bar{\mathcal{G}}$ is transitive and the loss function is invariant, the conditional risk of an invariant estimator ${\psi}$ of $\theta \in \bm{\theta}$, is constant for all $\bm{\theta} \in \bm{\Theta}$.
\end{theorem}
\begin{remark}
If the group of transformations $\bar{\mathcal{G}}$ on $\bm{\Theta}$ is transitive, and ${\psi}$ is an invariant estimator of ${\theta}$, we have
\begin{eqnarray}\label{Invariant_risk}
\mathcal{R}({\psi}, \bm{\theta}) = \mathcal{M}({\psi}) = \mathcal{B}({\psi}, p),
\end{eqnarray}
for any $p(\bm{\theta})$ defined over $\bm{\theta} \in \bm{\Theta}$.
\end{remark}

When $\bar{\mathcal{G}}$ is transitive, we can construct the optimum (or minimum conditional risk) invariant estimator under $\mathcal{G}$, when the loss function is invariant under $\mathcal{G}$ using the theory from \cite{Berger}. In this paper, we use the concepts of invariant estimation theory to design the optimum invariant CSOE schemes under the \emph{K-model} (see (\ref{Kmodel})) and \emph{S-model} (see (\ref{Smodel})). As we are primarily interested in estimating $\delta$ and $\phi$, we consider the loss functions defined by
\begin{eqnarray}\label{loss_fn1}
L_{1}(a_{\delta}, \bm{\theta}) & = & \frac{(a_{\delta} - \delta)^2}{\phi^2},
\end{eqnarray}
and
\begin{eqnarray}\label{loss_fn2}
L_{2}(a_{\phi}, \bm{\theta}) & = & \frac{(a_{\phi} - \phi)^2}{\phi^2},
\end{eqnarray}
for $\delta$ and $\phi$, respectively\footnote{As seen in equations (\ref{Kmodel}) and (\ref{Smodel}), the unknown clock skew is similar to the unknown scale parameter in Example 1 and is multiplied with the random queuing delays.}. In (\ref{loss_fn1}) and (\ref{loss_fn2}), $a_{\delta}$ and $a_{\phi}$ denote estimates of $\delta$ and $\phi$, respectively, $\bm{\theta} = [\phi, \delta]$ in case of the \emph{K-model}, and  $\bm{\theta} = [\phi, d, \delta]$ for the \emph{S-model}\footnote{We are not interested in estimating the value of $d$, as it is a nuisance parameter.}. We then use results from \cite{Berger, marchand2012, Bondar1981, Lehmann} to show the derived optimum invariant estimators of $\delta$ and $\phi$ are minimax for the skew-normalized squared error loss functions defined in (\ref{loss_fn1}) and (\ref{loss_fn2}), respectively.

\section{Minimax optimum CSOE scheme under K-model}\label{Sec4}
We now apply invariant decision theory to derive the optimum invariant estimator of $\phi$ and $\delta$ in the \emph{K-model}. Recall from (\ref{Kmodel}), the observations under the \emph{K-model} can be represented as
\begin{eqnarray}
\bm{y} & = & \bm{u}\phi + \delta\bm{1}_{2P},
\end{eqnarray}
where $\bm{y} \in \mathbb{R}^{2P}$, $\bm{u} \in \mathbb{R}^{2P}$, $\phi \in \mathbb{R}^+$ and $\delta \in \mathbb{R}$. Let $\bm{\theta} = [\phi, \delta]$ denote the vector of unknown parameters. The parameter space of $\bm{\theta}$, denoted by $\bm{\Theta}$, is given by
\begin{eqnarray}\label{K_Param_space}
\bm{\Theta} & = & \{(\phi, \delta): \phi \in \mathbb{R}^+, \delta \in \mathbb{R} \}.
\end{eqnarray}
From (\ref{Kmodel}), we have $f_{\bm{u}_1}(\bm{u}_1) = f_{\bm{w}_1}(\bm{u}_{1} - \bm{t}_{1} - d_{ms}\bm{1}_P^T)$, $f_{\bm{u}_2}(\bm{u}_2) = f_{\bm{w}_2}(\bm{t}_{4} - \bm{u}_{2} - d_{sm}\bm{1}_P^T)$, $f_{\bm{u}}(\bm{u}) = f_{\bm{u}_1}(\bm{u}_1) f_{\bm{u}_2}(\bm{u}_2)$ and
\begin{eqnarray}\label{main_Kmodel_joint_dist}
f(\bm{y}|\bm{\theta}) & = & \frac{1}{\phi^{2P}} f_{\bm{u}}\left( \frac{\bm{y} - \delta\bm{1}_{2P}}{\phi} \right) =  \frac{1}{\phi^{2P}} f_{\bm{u}_1}\left( \frac{\bm{t}_2 - \delta\bm{1}_{P}^T}{\phi} \right) f_{\bm{u}_2}\left( \frac{\bm{t}_3 - \delta\bm{1}_{P}^T}{\phi} \right), \\
& = & \frac{1}{\phi^{2P}}f_{\bm{w}_1}\left(\frac{\bm{t}_{2} - \delta\bm{1}_P^T}{\phi} - d_{ms}\bm{1}_P^T - \bm{t}_{1}\right) f_{\bm{w}_2}\left(\frac{\delta\bm{1}_P^T - \bm{t}_{3}}{\phi} - d_{sm}\bm{1}_P^T + \bm{t}_{4}\right).
\end{eqnarray}

Let $\mathcal{F}_{KModel}$ denote the class of all densities $f(\bm{y}|\bm{\theta})$ for $\bm{\theta} \in \bm{\Theta}$. 
%As previously shown in Example 1,
 The class of such densities is invariant under the group of location-scale transformations (see Example 5, Section 6.2.1, \cite{Berger}) $\mathcal{G}_{KModel}$, on $\mathbb{R}^{2P}$, defined as
\begin{eqnarray}\label{K_model_group}
\mathcal{G}_{KModel}  = \{ g_{a,b}(\bm{m}) : g_{a, b}(\bm{m})  = a\bm{m} + b\bm{1}_{2P} \} \mbox{ where } a \in \mathbb{R}^+, b \in \mathbb{R} \mbox{ and }\bm{m} \in \mathbb{R}^{2P},
\end{eqnarray} 
since $\bm{y}_g = g_{a, b}(\bm{y})$ has the density $\frac{1}{(a\phi)^{2P}} f_{\bm{u}}\left( \frac{\bm{y}_g - (a\delta + b)\bm{1}_{2P}}{a\phi}  \right)$. The group, $\bar{\mathcal{G}}_{KModel}$, of induced transformations on $\bm{\Theta}$ defined in (\ref{K_Param_space}), is given by
\begin{eqnarray}\label{K_model_param_group}
\bar{\mathcal{G}}_{KModel} = \{ \bar{g}_{a,b}((\phi, \delta)) : \bar{g}_{a, b}((\phi, \delta))  = (a\phi, (a\delta + b)) \},
\end{eqnarray}
where $a \in \mathbb{R}^+, b \in \mathbb{R}, \phi \in \mathbb{R}^+$ and $\delta \in \mathbb{R}$.
%thus consists of the transformations defined by
%\begin{eqnarray}
%\bar{g}_{a, b}((\phi, \delta)) & = & (a\phi, (a\delta + b)) \label{Transform_K}.
%\end{eqnarray}

Let $\hat{\delta}_I$ and $\hat{\phi}_I$ denote estimators of $\delta$ and $\phi$, respectively and let $\hat{\delta}_I(\bm{y})$ and $\hat{\phi}_I(\bm{y})$ denote the estimates obtained from the received data $\bm{y}$ characterized by the pdf $f(\bm{y}|\bm{\theta}) = \frac{1}{\phi^{2P}} f_{\bm{u}}\left( \frac{\bm{y} - \delta\bm{1}_{2P}}{\phi} \right)$. The estimators $\hat{\phi}_{I}(\bm{y})$ and $\hat{\delta}_{I}(\bm{y})$ are invariant under $\mathcal{G}_{KModel}$ from (\ref{K_model_group}) if for all $(a, b) \in \mathbb{R}^+ \times \mathbb{R}$
\begin{eqnarray}
\hat{\delta}_{I}(g_{a, b}(\bm{y})) = \hat{\delta}_{I}(a\bm{y} + b\bm{1}_{2P}) & = & a\hat{\delta}_{I}(\bm{y}) + b, \label{invariant_Klocation} \\
\hat{\phi}_{I}(g_{a, b}(\bm{y})) = \hat{\phi}_{I}(a\bm{y} + b\bm{1}_{2P}) & = & a\hat{\phi}_{I}(\bm{y}). \label{invariant_Kscale}
\end{eqnarray}
Further, the skew-normalized loss functions defined in (\ref{loss_fn1}) and (\ref{loss_fn2}) for $\delta$ and $\phi$, respectively, are invariant under $\mathcal{G}_{KModel}$ from (\ref{K_model_group}), since
\begin{eqnarray}
\frac{(\hat{\delta}_{I}(\bm{y}) - \delta)^2}{\phi^2} = \frac{\left(\hat{\delta}_{I}(g_{a, b}(\bm{y})) - (a\delta + b)\right)^2}{a^2\phi^2}, \mbox{ and } 
\frac{(\hat{\phi}_{I}(\bm{y}) - \phi)^2}{\phi^2}  =  \frac{\left(\hat{\phi}_{I}(g_{a, b}(\bm{y})) - a\phi\right)^2}{a^2\phi^2}
\end{eqnarray}
for all $g_{a, b} \in \mathcal{G}_{KModel}$. We now present the minimax optimum estimators of $\delta$ and $\phi$ under the \emph{K-model}.

\begin{proposition}\label{Minimax_optimum_Kmodel_estimator_phase_freq}
	The optimum (or minimum conditional risk) invariant estimators of $\delta$ and $\phi$, denoted by $\hat{\delta}_{MinRisk}$ and $\hat{\phi}_{MinRisk}$, respectively, under $\mathcal{G}_{KModel}$ defined in (\ref{K_model_group}), for the skew-normalized squared error loss function defined in (\ref{loss_fn1}) and (\ref{loss_fn2}), respectively, are given by 
\begin{eqnarray}
\hat{\delta}_{MinRisk}(\bm{y})  & =  & \frac{\int_{\mathbb{R}^+}\int_{\mathbb{R}} \frac{\delta}{\phi^{3}}  f(\bm{y}|\bm{\theta}) d\delta d\phi }{\int_{\mathbb{R}^+}\int_{\mathbb{R}} \frac{1}{\phi^{3}} f(\bm{y}|\bm{\theta}) d\delta d\phi}, \label{Minimax_Kmodel_offset} \\
\hat{\phi}_{MinRisk}(\bm{y}) & = & \frac{\int_{\mathbb{R}^+}\int_{\mathbb{R}} \frac{1}{\phi^{2}} f(\bm{y}|\bm{\theta}) d\delta d\phi }{\int_{\mathbb{R}^+}\int_{\mathbb{R}} \frac{1}{\phi^{3}} f(\bm{y}|\bm{\theta}) d\delta d\phi}, \label{Minimax_Kmodel_skew}
\end{eqnarray}
where $f(\bm{y}|\bm{\theta}) = \frac{1}{\phi^{2P}}f_{\bm{w}_1}\left(\frac{\bm{t}_{2} - \delta\bm{1}_P^T}{\phi} - d_{ms}\bm{1}_P^T - \bm{t}_{1}\right) f_{\bm{w}_2}\left(\frac{\delta\bm{1}_P^T - \bm{t}_{3}}{\phi} - d_{sm}\bm{1}_P^T + \bm{t}_{4}\right)$.
Further, the derived optimum invariant estimators are minimax for the skew-normalized squared error loss.
\end{proposition}
\noindent
\begin{proof}
As $\bar{\mathcal{G}}_{KModel}$ defined in (\ref{K_model_param_group}) is transitive on $\bm{\Theta}$, the optimum invariant estimator of $\delta$ under $\mathcal{G}_{KModel}$ in (\ref{K_model_group}), denoted by $\hat{\delta}_{MinRisk}$, can be obtained by solving (See {\bf Result 3} in Section 6.6.2 of \cite{Berger})
\begin{eqnarray}\label{Bayes_risk_Kmodel}
\hat{\delta}_{MinRisk}(\bm{y})  =  \underbrace{\arg \min}_{\hat{\delta}} \int_{\bm{\Theta}} L_{1}(\hat{\delta}(\bm{y}), \bm{\theta})  \pi^r(\bm{\theta}|\bm{y}) d\bm{\theta}=\underbrace{\arg \min}_{\hat{\delta}} \int_{\bm{\Theta}} \frac{(\hat{\delta}(\bm{y}) - \delta)^2}{\phi^2}  \pi^r(\bm{\theta}|\bm{y}) d\bm{\theta},
\end{eqnarray}
where $\pi^r(\bm{\theta}|\bm{y}) =  \frac{f(\bm{y}|\bm{\theta})\pi^r(\bm{\theta})}{\int_{\bm{\Theta}}f(\bm{y}|\bm{\theta})\pi^r(\bm{\theta}) d\bm{\theta}}$
%\begin{eqnarray}\label{posterior}
%\pi^r(\bm{\theta}|\bm{y}) & = & \frac{f(\bm{y}|\bm{\theta})\pi^r(\bm{\theta})}{\int_{\bm{\Theta}}f(\bm{y}|\bm{\theta})\pi^r(\bm{\theta}) d\bm{\theta}},
%\end{eqnarray}
is the posterior density of $\bm{\theta}$ based on the right invariant prior $\pi^r$ on $\bm{\Theta}$ (see Section 6.6.1, \cite{Berger})\footnote{The right invariant prior density need not be an actual density \cite{Berger} (See section 6.6, page 409).}. The right invariant prior for the location-scale group was derived in  \cite{Berger} (see Section 6.6). As $\mathcal{G}_{KModel}$ from (\ref{K_model_group}) is a location-scale group, the right invariant prior density for $\mathcal{G}_{KModel}$ is given by
\begin{eqnarray}
\pi^r(\bm{\theta}) & = & \frac{1}{\phi} \mathcal{I}_{\mathbb{R}^+}(\phi) \mathcal{I}_{\mathbb{R}}(\delta).
\end{eqnarray}
\noindent
To find $\hat{\delta}_{MinRisk}$, we differentiate the objective function in (\ref{Bayes_risk_Kmodel}) with respect to $\hat{\delta}(\bm{y})$ and set the result equal to zero (Section 2.4.1, \cite{VanTrees}). We obtain
%Taking the derivative of (\ref{Bayes_risk_Kmodel}) with respect to $\hat{\delta}_{mri}$, setting the result to zero and solving for $\hat{\delta}_{mri}$ yields
\begin{eqnarray}
\hat{\delta}_{MinRisk}(\bm{y}) & = & \frac{\int_{\mathbb{R}^+}\int_{\mathbb{R}} \frac{\delta}{\phi^2} \pi^r(\bm{\theta}|\bm{y}) d\bm{\theta} }{\int_{\mathbb{R}^+}\int_{\mathbb{R}} \frac{1}{\phi^2} \pi^r(\bm{\theta}|\bm{y}) d\bm{\theta}} = \frac{\int_{\mathbb{R}^+}\int_{\mathbb{R}} \frac{\delta}{\phi^3} f(\bm{y}|\bm{\theta}) d\bm{\theta} }{\int_{\mathbb{R}^+}\int_{\mathbb{R}} \frac{1}{\phi^3} f(\bm{y}|\bm{\theta}) d\bm{\theta}}.
\end{eqnarray}
Similarly, the optimum invariant estimator of $\phi$ under $\mathcal{G}_{KModel}$ in (\ref{K_model_group}), denoted by $\hat{\phi}_{MinRisk}$, can be obtained by 
\begin{eqnarray}\label{Bayes_risk_Kmodel_phi}
\hat{\phi}_{MinRisk}(\bm{y}) & = & \underbrace{\arg \min}_{\hat{\phi}} \int_{\bm{\Theta}} \frac{(\hat{\phi}(\bm{y}) - \phi)^2}{\phi^2}  \pi^r(\bm{\theta}|\bm{y}) d\bm{\theta}.
\end{eqnarray}
Solving using the same derivative-based approach, we obtain
\begin{eqnarray}
\hat{\phi}_{MinRisk}(\bm{y}) & = & \frac{\int_{\mathbb{R}^+}\int_{\mathbb{R}} \frac{1}{\phi^{2}} f(\bm{y}|\bm{\theta}) d\delta d\phi }{\int_{\mathbb{R}^+}\int_{\mathbb{R}} \frac{1}{\phi^{3}} f(\bm{y}|\bm{\theta}) d\delta d\phi}. 
\end{eqnarray}

When the class of densities is invariant under the location-scale group, it was shown in \cite{Bondar1981} that the optimum invariant estimator of a parameter for an invariant loss function is also a minimax estimator of the parameter for the considered loss function. As the class of densities $\mathcal{F}_{KModel}$ is invariant under $\mathcal{G}_{KModel}$ in (\ref{K_model_group}) (a location-scale group), and the scale invariant loss function is invariant under $\mathcal{G}_{KModel}$, the optimum invariant estimators $\hat{\delta}_{MinRisk}$ and $\hat{\phi}_{MinRisk}$, are minimax optimum estimators of $\delta$ and $\phi$, respectively, for the skew-normalized squared error loss functions given in (\ref{loss_fn1}) and (\ref{loss_fn2}), respectively.
\end{proof}

We now present an important result with regards to the mean square estimation error performance of the minimax optimum estimators when compared to ML estimators. Let $\hat{\delta}$ and $\hat{\phi}$ denote estimators of $\delta$ and $\phi$, respectively. The Mean Square estimation Errors (MSEs) of $\hat{\delta}$ and $\hat{\phi}$, denoted by $\mbox{MSE}(\hat{\delta})$ and $\mbox{MSE}(\hat{\phi})$, respectively, are defined as
\begin{eqnarray}\label{MSE_delta}
\mbox{MSE}(\hat{\delta}) = E\left\{ (\hat{\delta} - \delta)^2|\bm{\theta} \right\}, \mbox{ and } 
\mbox{MSE}(\hat{\phi}) = E\left\{ (\hat{\phi} - \phi)^2 |\bm{\theta} \right\},
\end{eqnarray}
where $E\{.\}$ denotes the expectation operator and $\bm{\theta}$ is the vector of unknown parameters. 
\begin{proposition}\label{MLE_comp}
Let $\hat{\delta}_{MLE}$ and $\hat{\phi}_{MLE}$ denote the ML estimators of $\delta$ and $\phi$, respectively. Under the K-model, the MSE of $\hat{\delta}_{MLE}$ is always greater than or equal to the MSE of $\hat{\delta}_{MinRisk}$. Also, under the K-model, the MSE of $\hat{\phi}_{MLE}$ is always greater than or equal to the MSE of $\hat{\phi}_{MinRisk}$.
\end{proposition}
\begin{proof}
In the \emph{K-model}, we have $\bm{\theta} = [\phi, \delta]$. Let $\hat{\phi}_{MLE}(\bm{y})$ and $\hat{\delta}_{MLE}(\bm{y})$ denote the ML estimates obtained from $\bm{y}$ characterized by the pdf $f(\bm{y}|\bm{\theta}) = \frac{1}{\phi^{2P}}f_{\bm{u}}(\frac{\bm{y} - \delta\bm{1}_{2P}}{\phi})$ from (\ref{main_Kmodel_joint_dist}). We have
\begin{eqnarray}\label{ML_estimate}
\hat{\bm{\theta}}_{MLE}(\bm{y})  = [\hat{\phi}_{MLE}(\bm{y}), \hat{\delta}_{MLE}(\bm{y})]& = & \underbrace{\arg \max}_{\bm{\theta}} \log \mathcal{L}(\bm{\theta}|\bm{y}),
\end{eqnarray}
where $\mathcal{L}(\bm{\theta}|\bm{y})$ is the likelihood function and is equal to $f(\bm{y}|\bm{\theta})$. Let $g_{a, b} \in \mathcal{G}_{KModel}$ from (\ref{K_model_group}) and define $\bm{y}_g = g_{a, b}(\bm{y})$. From (\ref{K_model_param_group}), the corresponding transformation of the parameter vector $\bm{\theta}$ is given by $\bm{\theta}_{g} = \bar{g}_{a, b}(\bm{\theta}) = (a\phi, (a\delta + b))$. From the functional invariance of ML estimators \cite{Casella} (see Chapter 7, Theorem 7.2.10), we have $\hat{\bm{\theta}}_{MLE}(\bm{y}_g) = \bar{g}_{a, b}(\hat{\bm{\theta}}_{MLE}(\bm{y}))$. So, we have the following relationship 
\begin{eqnarray}
\hat{\delta}_{MLE}(\bm{y}_g) = a\hat{\delta}_{MLE}(\bm{y}) + b, & \mbox{ and } & \hat{\phi}_{MLE}(\bm{y}_g) = a\hat{\phi}_{MLE}(\bm{y}).
\end{eqnarray}

As this holds true for all $g_{a, b} \in \mathcal{G}_{KModel}$ from (\ref{K_model_group}), the ML estimators of $\delta$ and $\phi$ are invariant under $\mathcal{G}_{KModel}$ as they satisfy (\ref{invariant_Klocation}) and (\ref{invariant_Kscale}). For the skew-normalized loss function defined in (\ref{loss_fn1}), we have
\begin{eqnarray}\label{Risk_comp}
\mathcal{R}(\hat{\delta}_{MinRisk}, \bm{\theta}) & \le & \mathcal{R}(\hat{\delta}_{MLE}, \bm{\theta}), 
\end{eqnarray}
since $\hat{\delta}_{MinRisk}$ is the optimum invariant estimator under $\mathcal{G}_{KModel}$ in (\ref{K_model_group}) and achieves the minimum conditional risk among all estimators that are invariant under $\mathcal{G}_{KModel}$ (see Proposition \ref{Minimax_optimum_Kmodel_estimator_phase_freq}). From (\ref{Risk_comp}), we have
\begin{eqnarray}
\int_{\mathbb{R}^{2P}} \frac{(\hat{\delta}_{MinRisk}(\bm{y}) - \delta)^2}{\phi^2}f(\bm{y}|\bm{\theta})d\bm{y} & \le & \int_{\mathbb{R}^{2P}} \frac{(\hat{\delta}_{MLE}(\bm{y}) - \delta)^2}{\phi^2}f(\bm{y}|\bm{\theta})d\bm{y}, \\
\implies \int_{\mathbb{R}^{2P}} {(\hat{\delta}_{MinRisk}(\bm{y}) - \delta)^2}f(\bm{y}|\bm{\theta})d\bm{y} & \le & \int_{\mathbb{R}^{2P}} {(\hat{\delta}_{MLE}(\bm{y}) - \delta)^2}f(\bm{y}|\bm{\theta})d\bm{y}, \\
\implies \mbox{MSE}(\hat{\delta}_{MinRisk}) & \le & \mbox{MSE}(\hat{\delta}_{MLE}).
\end{eqnarray}
Following similar steps, we can show that $\mbox{MSE}(\hat{\phi}_{MinRisk}) \le \mbox{MSE}(\hat{\phi}_{MLE})$.
\end{proof}

\section{Minimax optimum CSOE scheme under S-model}\label{Sec5}
%%%%%%%%%%%%%%%%%%%%%%%%%%%%%%%%  S-model estimator %%%%%%%%%%%%%%%%%%%%%%%%%%%%%%%%%%%%%%%%%%%
We now apply invariant decision theory to derive the optimum invariant estimator of $\phi$ and $\delta$ under the \emph{S-model}. Recall from (\ref{Smodel}), the observations under the \emph{S-model} can be represented as
\begin{eqnarray}
\bm{y} & = & (\bm{h}d + \bm{v})\phi + \delta \bm{1}_{2P},
\end{eqnarray}
where $\bm{y} \in \mathbb{R}^{2P}$, $\bm{v} \in \mathbb{R}^{2P}$, $\phi \in \mathbb{R}^+$ and $\delta \in \mathbb{R}$. As the unknown fixed delay $d$ is always non-negative, we have $d \in \mathbb{R}_0^+$. However, it is not possible to design invariant estimators under this constraint\footnote{When $d \in \mathbb{R}_0^+$, it is not possible to construct a group of transformations for which the class of densities in the \emph{S-model} is invariant under the group of transformations.}, so we assume $d \in \mathbb{R}$, but later we see this is not a problem as we derive the minimax optimum estimator in Proposition \ref{Minimax_optimum_Smodel_estimator_phase_freq}. Let $\bm{\theta} = [\phi, d, \delta]$ denote the vector of unknown parameters. The unrestricted parameter space of $\bm{\theta}$, denoted by $\bm{\Theta}$, is given by
\begin{eqnarray}\label{S_unres_param}
\bm{\Theta} & = & \{(\phi, d, \delta): \phi \in \mathbb{R}^+, d \in \mathbb{R}, \delta \in \mathbb{R} \},
\end{eqnarray}
and the restricted parameter space of $\bm{\theta}$, denoted by $\bm{\Theta}^*$, is given by
\begin{eqnarray}\label{S_res_param}
\bm{\Theta}^* & = & \{(\phi, d, \delta): \phi \in \mathbb{R}^+, d \in \mathbb{R}_0^+, \delta \in \mathbb{R} \}.
\end{eqnarray}
From (\ref{Smodel}), we have $f_{\bm{v}_1}(\bm{v}_1) = f_{\bm{w}_1}(\bm{v}_1 - \bm{t}_1)$, $f_{\bm{v}_2}(\bm{v}_2) = f_{\bm{w}_2}(\bm{t}_4 - \bm{v}_2)$, $f_{\bm{v}}(\bm{v}) = f_{\bm{v}_1}(\bm{v}_1)f_{\bm{v}_2}(\bm{v}_2)$ and
\begin{eqnarray}
f(\bm{y}|\bm{\theta}) & = & \frac{1}{\phi^{2P}} f_{\bm{v}}\left( \frac{\bm{y} - \delta\bm{1}_{2P}}{\phi} - \bm{h}d\right) = \frac{1}{\phi^{2P}} f_{\bm{v}_1}\left( \frac{\bm{t}_2 - \delta\bm{1}_{P}^T}{\phi} - d\bm{1}_P^T\right)f_{\bm{v}_2}\left( \frac{\bm{t}_3 - \delta\bm{1}_{P}^T}{\phi} + d\bm{1}_P^T\right), \nonumber \\
& = & \frac{1}{\phi^{2P}} f_{\bm{w}_1}\left(\frac{\bm{t}_{2} - \delta\bm{1}_P^T}{\phi} - d\bm{1}_P^T - \bm{t}_{1}\right) f_{\bm{w}_2}\left(\frac{\delta\bm{1}_P^T - \bm{t}_{3}}{\phi} - d\bm{1}_P^T + \bm{t}_{4}\right). \label{S_model_pdf}
\end{eqnarray}

Let $\mathcal{F}_{SModel}$ denote the class of all densities $f(\bm{y}|\bm{\theta})$ for $\bm{\theta} \in \bm{\Theta}$. 
The class of such densities is invariant under the group of transformations $\mathcal{G}_{SModel}$, on $\mathbb{R}^{2P}$, defined as
\begin{eqnarray}\label{S_model_group}
\mathcal{G}_{SModel}  & = & \{ g_{a,b,c}(\bm{m}) : g_{a, b, c}(\bm{m})  = a(\bm{m} + \bm{h}b) + c\bm{1}_{2P} \},
\end{eqnarray} 
where $a \in \mathbb{R}^+, b \in \mathbb{R}, c \in \mathbb{R} \mbox{ and }\bm{m} \in \mathbb{R}^{2P}$, since $\bm{y}_g = g_{a, b, c}(\bm{y})$ has the density $\frac{1}{(a\phi)^{2P}} \allowbreak f_{\bm{v}}\left( \frac{\bm{y}_g - (a\delta + c)\bm{1}_{2P}}{a\phi}  - \bm{h}\left(d + \frac{b}{\phi} \right)\right)$. The group, $\bar{\mathcal{G}}_{SModel}$, of induced transformations on $\bm{\Theta}$ defined in (\ref{S_unres_param}),
%\begin{eqnarray}
%\bm{\Theta} & = & \{(\phi, d, \delta): \phi \in \mathbb{R}^+, d \in \mathbb{R}, \delta \in \mathbb{R} \}
%\end{eqnarray}
is given by
\begin{eqnarray}\label{S_model_param_group}
\bar{\mathcal{G}}_{SModel} = \{ \bar{g}_{a,b,c}((\phi, d, \delta)) : \bar{g}_{a, b, c}((\phi, d, \delta))  = (a\phi, (d + b/\phi), (a\delta + c)) \},
\end{eqnarray}
where $a \in \mathbb{R}^+, b \in \mathbb{R}, c \in \mathbb{R}, \phi \in \mathbb{R}^+, d \in \mathbb{R}$ and $\delta \in \mathbb{R}$.

Let $\hat{\delta}_I$ and $\hat{\phi}_I$ denote estimators of $\delta$ and $\phi$, respectively and let $\hat{\delta}_I(\bm{y})$ and $\hat{\phi}_I(\bm{y})$ denote the estimates obtained from the received data $\bm{y}$ characterized by the pdf $f(\bm{y}|\bm{\theta}) = \frac{1}{\phi^{2P}} f_{\bm{v}}\left( \frac{\bm{y} - \delta\bm{1}_{2P}}{\phi} - \bm{h}d\right)$.  The estimators $\hat{\phi}_{I}(\bm{y})$ and $\hat{\delta}_{I}(\bm{y})$ are invariant under $\mathcal{G}_{SModel}$ from (\ref{S_model_group}), if for all $(a, b, c) \in \mathbb{R}^+ \times \mathbb{R} \times \mathbb{R}$,
\begin{eqnarray}
\hat{\delta}_{I}(g_{a, b, c}(\bm{y})) = \hat{\delta}_{I}(a(\bm{y} + \bm{h}b) + c\bm{1}_{2P}) & = & a\hat{\delta}_{I}(\bm{y}) + c, \\
\hat{\phi}_{I}(g_{a, b, c}(\bm{y})) = \hat{\phi}_{I}(a(\bm{y} + \bm{h}b) + c\bm{1}_{2P}) & = & a\hat{\phi}_{I}(\bm{y}).
\end{eqnarray}
Further, the skew-normalized loss functions defined in (\ref{loss_fn1}) and (\ref{loss_fn2}) for $\delta$ and $\phi$, respectively, are invariant under $\mathcal{G}_{SModel}$ from (\ref{S_model_group}), since
\begin{eqnarray}
\frac{(\hat{\delta}_{I}(\bm{y}) - \delta)^2}{\phi^2} = \frac{\left(\hat{\delta}_{I}(g_{a, b, c}(\bm{y})) - (a\delta + c)\right)^2}{a^2\phi^2}, \mbox{ and } \frac{(\hat{\phi}_{I}(\bm{y}) - \phi)^2}{\phi^2} = \frac{\left(\hat{\phi}_{I}(g_{a, b, c}(\bm{y})) - a\phi\right)^2}{a^2\phi^2}
\end{eqnarray}
%and
%\begin{eqnarray}
%\frac{(\hat{\phi}_{I}(\bm{y}) - \phi)^2}{\phi^2} & = & \frac{\left(\hat{\phi}_{I}(g_{a, b, c}(\bm{y})) - a\phi\right)^2}{a^2\phi^2}
%\end{eqnarray}
for all $g_{a, b, c} \in \mathcal{G}_{SModel}$. We now present the minimax optimum estimators of $\delta$ and $\phi$ under the \emph{S-model}.

\begin{proposition}\label{Minimax_optimum_Smodel_estimator_phase_freq}
	The optimum (or minimum conditional risk) invariant estimators of $\delta$ and $\phi$, denoted by $\hat{\delta}_{MinRisk}$ and $\hat{\phi}_{MinRisk}$, respectively, under $\mathcal{G}_{SModel}$ defined in (\ref{S_model_group}), for the scale invariant squared error loss functions defined in (\ref{loss_fn1}) and (\ref{loss_fn2}), respectively, are given by 
	\begin{eqnarray}
	\hat{\delta}_{MinRisk}(\bm{y}) & = & \frac{\int_{\mathbb{R}^+}\int_{\mathbb{R}^2} \frac{\delta}{\phi^{2}}  f(\bm{y}|\bm{\theta}) d(d) d\delta d\phi }{\int_{\mathbb{R}^+}\int_{\mathbb{R}^2} \frac{1}{\phi^{2}} f(\bm{y}|\bm{\theta}) d(d) d\delta d\phi}, \label{Minimax_Smodel_offset} \\
	\hat{\phi}_{MinRisk}(\bm{y}) & = & \frac{\int_{\mathbb{R}^+}\int_{\mathbb{R}^2} \frac{1}{\phi} f(\bm{y}|\bm{\theta}) d(d) d\delta d\phi }{\int_{\mathbb{R}^+}\int_{\mathbb{R}^2} \frac{1}{\phi^{2}} f(\bm{y}|\bm{\theta}) d(d) d\delta d\phi}, \label{Minimax_Smodel_skew}
	\end{eqnarray}
	where $f(\bm{y}|\bm{\theta}) = f_{\bm{w}_1}\left(\frac{\bm{t}_{2} - \delta\bm{1}_P^T}{\phi} - d\bm{1}_P^T - \bm{t}_{1}\right) f_{\bm{w}_2}\left(\frac{\delta\bm{1}_P^T - \bm{t}_{3}}{\phi} - d\bm{1}_P^T + \bm{t}_{4}\right)$.
	Further, the derived optimum invariant estimators are minimax for the skew-normalized squared error loss in the restricted parameter space $\bm{\Theta}^*$ (see Appendix \ref{App_sec3} for proof).
\end{proposition}

\section{Simulation Results}\label{Sec6}
In this section, we compare the performance of the proposed minimax optimum estimators to the ML estimators discussed in \cite{Noh_2007, Leng_2010, Chaudhari_2008} via numerical simulations. We first briefly describe the approach used for generating the random queuing delays along with the generation of the packet timestamps. Then, we describe the various considered CSOE schemes, and finally, we present numerical results. For simplicity, we assume symmetric network conditions in the forward and reverse paths, i.e., $f_{\bm{w}_1}(.) = f_{\bm{w}_2}(.) = f_{\bm{w}}(.)$. Further, we assume the queuing delay samples $\{w_{kj}\}$ for $k = 1, 2$ and $j = 1, 2, \cdots, P$ are independent and identically distributed (i.i.d). Symmetric network conditions and i.i.d queuing delays are assumed in the signal model considered in \cite{Noh_2007, Leng_2010, Chaudhari_2008}.

We consider the two scenarios in our work, namely the backhaul network scenario discussed in Section \ref{Sec1}\footnote{PTP is used in conjunction with Synchronous Ethernet (SyncE) for synchronization in LTE backhaul networks. Although the SyncE standards are now mature, much of the deployed base of Ethernet equipment does not support it \cite{IEEE1588v2_whitepaper}. If a single Ethernet switch in the chain does not support SyncE, all nodes lower in the hierarchy do not receive the timing service \cite{IEEE1588v2_whitepaper}. PTP is the primary option for synchronization to operators with packet backhaul networks that do not support SyncE \cite{IEEE1588v2_whitepaper, IEEE1588v2_Huawei}. }, and the electrical grid scenario where PTP-based timing is used as an alternative to GPS-based timing in the electrical grid for scenarios when GPS signals are not available at the substation \cite{GPSattacks, GPSattacks1, NASPI}.

\subsection{Generating the random queuing delays and packet timestamps}
We briefly describe the generation of the random queuing delays in the considered packet switched networks.

\subsubsection{LTE backhaul networks}
We follow the approach given in \cite{Anand_bounds, Anand_2015} for generating the random queuing delays in the backhaul networks. We assume a Gigabit Ethernet network consisting of a cascade of $10$ switches between the master and slave node. A two-class non-preemptive priority queue is used to model the traffic at each switch. The network traffic at the switch comprises of the lower priority background traffic and the higher priority synchronization messages. We assume cross-traffic flows, where new background traffic is injected at each switch and this traffic exits at the subsequent switch. The arrival times and size of background traffic packets injected at each switch are assumed to be statistically independent. We use Traffic Model 1 (TM-1) and Traffic Model 2 (TM-2) from the ITU-T recommendation G.8261 \cite{ITU}, described in Table \ref{sec6_table1}, for generating the background traffic at each switch. The interarrival times between packets in background traffic are assumed to follow an exponential distribution, and we set the rate parameter of each exponential distribution accordingly to obtain the desired load factor, i.e., the percentage of the total capacity consumed by background traffic\cite{Anand_2015}. The empirical pdf of the queuing delays, shown in Fig. \ref{Empirical_pdf} were obtained using a custom MATLAB-based network simulator. The timestamps $t_{1i}$ and $t_{3i}$ are set to $40i$ $\mu s$ and $40i$ $\mu s + 20 \mu s$, respectively, for $i = 0, 1, \cdots, P-1$. For a given value of parameters $\{\phi, d, \delta\}$, the timestamps $t_{2i}$ and $t_{4i}$ are then generated using (\ref{FwdData}) and (\ref{RevData}), respectively, assuming $d_{ms} = d_{sm} = d$.

\begin{table}[t]
	\begin{center}
		\begin{tabular}{|c|c|c|}
			\hline
			{\bf Traffic Model} & {\bf Packet Sizes (in Bytes)} & {\bf \% of total load} \\
			\hline 
			TM-1 & \{64, 576, 1518\} & \{80\%, 5\%, 15\%\} \\
			\hline
			TM-2 & \{64, 576, 1518\} & \{30\%, 10\%, 60\%\} \\
			\hline			
		\end{tabular}
		\caption{Composition of background packets in the considered traffic models.}\label{sec6_table1}
	\end{center}
\end{table}

\subsubsection{Electrical grid networks}
We consider the scenario where the master clock in an Electrical Grid (EG) substation uses the available LTE-based packet switched network along with PTP to obtain the timing information from other sources \cite{NASPI}. We use the traffic model proposed in \cite{Khatib} for generating the random queuing delays in this scenario. A three-class non-preemptive priority queue is used for modeling the traffic at an access point in the EG network\footnote{We assume the access point is connected to a switch or router in the wired Ethernet network.}. The network traffic at the switch comprises of the Public Users (PU) traffic or background traffic, Fixed-Scheduling (FS) traffic and Event-Driven (ED) traffic. The FS traffic is the operational traffic between the utility's control center and the devices that contain meter readings (MR) data and is transmitted periodically. The ED traffic consists of the demand response traffic and other high priority traffic including timing synchronization packets. The arrival processes of the ED and PU traffic are assumed to be Poisson, while the FS traffic is assumed to be a deterministic batch arrival process \cite{Khatib}. The transmission priority, in descending order, is ED, PU, and FS. 

We consider a Gigabit Ethernet network and a cascade of $10$ switches between the master and slave node. The arrival times and sizes of traffic packets injected at each switch are assumed to be statistically independent of traffic at other access points. We use TM-1 for generating the PU traffic and assume cross-traffic flows. The rate parameter of the exponentially distributed inter-arrival times between packets of the PU traffic is set accordingly to obtain the desired load factor (the percentage of the total capacity consumed by the PU traffic.). The period of the FS traffic is assumed to be $1$ second with the packet size fixed to $512$ Bytes. The batch size of the FS traffic is a discrete random variable following a uniform distribution of maximum size $100$. This network scenario is abbreviated as EG-TM1. The empirical pdf of the queuing delays for the PTP synchronization packets in the considered networks are shown in Figure \ref{Empirical_pdf_sg}.

\subsection{Considered CSOE schemes}
We now briefly describe the considered CSOE schemes. 

\subsubsection{Gaussian Maximum Likelihood Estimate (GMLE)}
We assume the \emph{K-model} for this CSOE scheme. Leng and Wu \cite{Leng_2010} proposed an ML-based CSOE scheme under the assumption that the queuing delay follows a Gaussian distribution. As shown in \cite{Leng_2010}, the approach assumes the PDV pdf is a zero-mean Gaussian and the variance cancels out in the derivation of the ML estimate. The estimation is equivalent to the least squares fit (see \cite{Leng_2010}), which is very popular in statistics\footnote{For the considered scenarios in this paper, we compensate for the mean of $f_{\bm{w}}(.)$ before using this estimator.}. It can be shown that this CSOE scheme is invariant under $\mathcal{G}_{KModel}$ defined in (\ref{K_model_group}).

\subsubsection{Local Maximum Likelihood Estimate (LMLE)}
We assume the \emph{K-model} for this CSOE scheme. As discussed in Proposition \ref{MLE_comp}, the ML estimate under the \emph{K-model} is obtaining by finding the value of parameters that maximize the likelihood function (see (\ref{ML_estimate})). However, for small values of $P$, the likelihood function need not always be convex. The likelihood function is shown in Figure \ref{LMLE_LLF} for TM-1 network scenario under $40\%$ load for $\phi = 1$ and $\delta = 0$ for different values of $P$. We see that for small values of $P$, the likelihood function is not necessarily convex and sometimes it has many local maxima. In our simulations, we use the solution obtained from GMLE as the initial point in the search for the ML estimate. The obtained solution is called the \emph{Local Maximum Likelihood Estimate} since we cannot guarantee a global maximum. To date, there is no known way to assure a global maximum (or minimum) has been found.

\subsubsection{Minimax Optimum Estimate under K-model (Minimax-K)}
We assume the \emph{K-model} for this CSOE scheme. The unknown parameters $\delta$ and $\phi$ are estimated using (\ref{Minimax_Kmodel_offset}) and (\ref{Minimax_Kmodel_skew}), respectively.

\subsubsection{Minimax Optimum Estimate under S-model (Minimax-S)}
We assume the \emph{S-model} for this CSOE scheme. The unknown parameters $\delta$ and $\phi$ are estimated using (\ref{Minimax_Smodel_offset}) and (\ref{Minimax_Smodel_skew}), respectively.

It should be mentioned that we have used the \emph{K-model} for all considered ML-based CSOE schemes. (The fixed delay $d$ is assumed to be known.) We conjecture that this provides a lower bound on the performance of an ML-based CSOE scheme in the \emph{S-model}, as the presence of additional unknown nuisance parameters would generally degrade the performance of a CSOE scheme.

\subsection{Performance Metric used for comparing CSOE schemes}
Let $\hat{\delta}$ and $\hat{\phi}$ denote estimators of $\delta$ and $\phi$, respectively. The Root Mean Square estimation Error (RMSE) of $\hat{\delta}$ and $\hat{\phi}$, denoted by $\mbox{RMSE}(\hat{\delta})$ and $\mbox{RMSE}(\hat{\phi})$, respectively, is defined by 
\begin{eqnarray}
\mbox{RMSE}(\hat{\delta}) =  {\sqrt{\mbox{MSE}(\hat{\delta}) }}  & \mbox{ and } & 
\mbox{RMSE}(\hat{\phi})  =  {\sqrt{\mbox{MSE}(\hat{\phi}) }},
\end{eqnarray}
where $\mbox{MSE}(\hat{\delta})$ and $\mbox{MSE}(\hat{\phi})$ are defined in (\ref{MSE_delta}). In this paper, we use $\mbox{RMSE}(\hat{\delta})$ and $\mbox{RMSE}(\hat{\phi})$ to evaluate the performance of a CSOE scheme.

\subsection{Numerical results}
We carried out numerical simulations for the considered CSOE schemes under TM-1 and TM-2 LTE backhaul network scenarios for different values of load factors. Figures \ref{rmse_TM1_offset_results}--\ref{rmse_TM2_skew_results} show the RMSE performance for the considered CSOE schemes for $\{\phi, d, \delta\} = \{1, 2\mu s, 2\mu s\}$. From Figures \ref{rmse_TM1_offset_results}--\ref{rmse_TM2_skew_results}, we see that the performance of all the considered CSOE schemes improves with an increase in the number of two-way message exchanges. We also observe that the proposed minimax optimum CSOE schemes exhibit better performance compared to LMLE and GMLE for the considered network scenarios. Further, following Proposition \ref{MLE_comp}, the minimax optimum estimator under the \emph{K-model} exhibits the lowest mean square error among the considered CSOE schemes that are invariant under $\mathcal{G}_{KModel}$. Also, we observe no significant loss in performance of the minimax optimum estimator under the \emph{S-model} due to the unknown nuisance parameter $d$ for all the considered network scenarios. As expected, the GMLE does not exhibit good performance under low load network scenarios. However, as the load factor increases, the performance of GMLE improves as the PDV pdf approximates a Gaussian distribution (see TM-2 for load factors 60\%, 80\% in Figure \ref{Empirical_pdf}). For all the considered scenarios, the LMLE exhibits an improvement in performance compared to the GMLE with noticeable improvement for high loads in TM-1. Further for TM-2 under high loads, the LMLE clock skew estimator exhibits performance close to minimax optimum CSOE schemes. Figures \ref{rmse_TM1sg_offset_results}-\ref{rmse_TM1sg_skew_results} show the performance of the proposed CSOE schemes for smart grid networks. Similar performance gains are observed for the smart grid network scenarios.

\section{Conclusion and Future Work}\label{Sec7}
In this paper, we have developed minimax optimum estimators for clock skew and offset estimation in PTP. The minimax optimum estimators exhibit lower mean square estimation error performance than the ML-based estimation schemes for a variety of network scenarios. Further, the proposed estimators can be easily extended to other timing protocols based on the two-way message exchange including TPSN\cite{TPSN}, tiny-sync \cite{MiniSync}, and LTS \cite{LTS}. Throughout this paper, we assumed a known affine relationship between the fixed path delays. The presence of an unknown asymmetry could degrade the performance of a CSOE scheme. Future work can look into developing robust clock skew and offset estimation schemes when there is an unknown asymmetry between the fixed path delays.

% if have a single appendix:
%\appendix[Proof of the Zonklar Equations]
% or
%\appendix  % for no appendix heading
% do not use \section anymore after \appendix, only \section*
% is possibly needed

% use appendices with more than one appendix
% then use \section to start each appendix
% you must declare a \section before using any
% \subsection or using \label (\appendices by itself
% starts a section numbered zero.)
%

%\appendices
%\section{Invariant Decision Theory}\label{App_sec1}
%\input{appendix_sec1}

%\section{Proof of Corollary 3.1}\label{App_sec2}
%\input{appendix_sec2}

\appendix[Proof of Proposition \ref{Minimax_optimum_Smodel_estimator_phase_freq}]\label{App_sec3}
%\section
\noindent
\begin{proof}
%Since $\bar{\mathcal{G}}_{SModel}$ defined in (\ref{S_model_param_group}) is transitive, the optimum invariant estimators of $\delta$ under $\mathcal{G}_{SModel}$ from (\ref{S_model_group}), denoted by $\hat{\delta}_{MinRisk}$, can be obtained by solving
%\begin{eqnarray}\label{Bayes_risk_Smodel}
%\hat{\delta}_{MinRisk}(\bm{y})  =  \underbrace{\arg \min}_{\hat{\delta}} \int_{\bm{\Theta}} \frac{(\hat{\delta}(\bm{y}) - \delta)^2}{\phi^2}  \pi^r(\bm{\theta}|\bm{y}) d\bm{\theta},
%\end{eqnarray}
%where $\pi^r(\bm{\theta}|\bm{y})$ was defined in (\ref{posterior}) and $\pi^r$ is the right invariant prior. We now derive the right invariant prior for $\bar{\mathcal{G}}_{SModel}$. 
We first calculate the right invariant prior for $\bar{\mathcal{G}}_{SModel}$, defined in (\ref{S_model_param_group}). This is necessary for deriving the optimum invariant estimator under $\mathcal{G}_{SModel}$ defined in (\ref{S_model_group}). We follow the steps given in Example 17, Section 6.6 of \cite{Berger}\footnote{Example 17 derives the right invariant prior for the location-scale transformation group.} to calculate the right invariant prior for $\bar{\mathcal{G}}_{SModel}$ in (\ref{S_model_param_group}). The transformation $\bar{g}_{a, b, c} \in \bar{\mathcal{G}}_{SModel}$ from (\ref{S_model_param_group}) can be considered as a point $(a, b, c) \in \mathbb{R}^3$, so we can represent $\bar{\mathcal{G}}_{SModel}$ equivalently as 
\begin{eqnarray}\label{S_model_param_group2}
\bar{\mathcal{G}}_{SModel} & = & \{(a, b, c) : a \in \mathbb{R}^+, b \in \mathbb{R}, c \in \mathbb{R}  \}.
\end{eqnarray}
Let $\bar{g} = \bar{g}_{a, b, c} \in \bar{\mathcal{G}}_{SModel}$ from (\ref{S_model_param_group2}) and $\bar{g}_0 = \bar{g}_{a_0, b_0, c_0} \in \bar{\mathcal{G}}_{SModel}$ from (\ref{S_model_param_group2}). In the new notation, the group transformation operation $\bar{g}\rightarrow \bar{g}\bar{g}_0$ can be written as
\begin{eqnarray}
(a, b, c) \rightarrow (a, b, c)(a_0, b_0, c_0) = (aa_0, (b_0 + b/a_0), (ac_0 + c)).
\end{eqnarray}
The function
\begin{eqnarray}
t((a, b, c)) & = & (t_1, t_2, t_3) = (aa_0, (b_0 + b/a_0), (ac_0 + c)),
\end{eqnarray}
has the differential given by
\begin{eqnarray}
\bm{H}^r_{\bar{g}_0}(\bar{g}) & = & \begin{bmatrix}
\frac{\partial t_1}{\partial a} & \frac{\partial t_1}{\partial b}  & \frac{\partial t_1}{\partial c}\\
\frac{\partial t_2}{\partial a} & \frac{\partial t_2}{\partial b}  & \frac{\partial t_2}{\partial c}\\
\frac{\partial t_3}{\partial a} & \frac{\partial t_3}{\partial b}  & \frac{\partial t_3}{\partial c}\\
\end{bmatrix} = \begin{bmatrix}
a_0 & 0 & 0 \\
0 & 1/a_0 & 0 \\
c_0 & 0 & 1
\end{bmatrix}.
\end{eqnarray}
The Jacobian of the transformation $\bar{g}\rightarrow \bar{g}\bar{g}_0$ is given by (see Definition 8, Section 6.6, \cite{Berger})
\begin{eqnarray}\label{Jacobian_Smodel}
J^r_{\bar{g}_0}(\bar{g}) & = & |\det\bm{H}^r_{\bar{g}_0}(\bar{g})| = 1.
\end{eqnarray}
%The right invariant Haar density on $\bar{\mathcal{G}}_{SModel}$ from (\ref{S_model_param_group2}), is given by (from {\bf Result 1}, Section 6.6, \cite{Berger}), $h^r(\bar{g}) = \frac{1}{J^r_{\bar{g}}(\bar{e})} = 1$, where $\bar{e}$ is the identity element of $\bar{\mathcal{G}}_{SModel}$. As $\bar{\mathcal{G}}_{SModel}$ from (\ref{S_model_param_group}) is transitive on $\bm{\Theta}$, the right invariant prior density on $\bm{\Theta}$ is given by
%(Section 6.6.1, \cite{Berger})
Using (\ref{Jacobian_Smodel}) and {\bf Result 1} from Section 6.6 in \cite{Berger}, the right invariant prior density on $\bm{\Theta}$ is given by
\begin{eqnarray}\label{right_invariant_prior_Smodel}
\pi^r(\bm{\theta}) & = &  \mathcal{I}_{\mathbb{R}^+}(\phi) \mathcal{I}_{\mathbb{R}}(d) \mathcal{I}_{\mathbb{R}}(\delta).
\end{eqnarray}

The optimum invariant estimators of $\delta$ under $\mathcal{G}_{SModel}$ from (\ref{S_model_group}), denoted by $\hat{\delta}_{MinRisk}$, can now be obtained by solving
\begin{eqnarray}\label{Bayes_risk_Smodel}
\hat{\delta}_{MinRisk}(\bm{y})  =  \underbrace{\arg \min}_{\hat{\delta}} \int_{\bm{\Theta}} \frac{(\hat{\delta}(\bm{y}) - \delta)^2}{\phi^2}  \pi^r(\bm{\theta}|\bm{y}) d\bm{\theta},
\end{eqnarray}
where $\pi^r(\bm{\theta}|\bm{y}) =  \frac{f(\bm{y}|\bm{\theta})\pi^r(\bm{\theta})}{\int_{\bm{\Theta}}f(\bm{y}|\bm{\theta})\pi^r(\bm{\theta}) d\bm{\theta}}$, $\pi^r$ is the right invariant prior defined in (\ref{right_invariant_prior_Smodel}) and $f(\bm{y}|\bm{\theta})$ is defined in (\ref{S_model_pdf}). To find $\hat{\delta}_{MinRisk}$, we differentiate the objective function in (\ref{Bayes_risk_Smodel}) with respect to $\hat{\delta}(\bm{y})$ and set the result equal to zero. We have
\begin{eqnarray}
\hat{\delta}_{MinRisk}(\bm{y}) & = & \frac{\int_{\mathbb{R}^+}\int_{\mathbb{R}^2} \frac{\delta}{\phi^2} \pi^r(\bm{\theta}|\bm{y}) d\bm{\theta} }{\int_{\mathbb{R}^+}\int_{\mathbb{R}^2} \frac{1}{\phi^2} \pi^r(\bm{\theta}|\bm{y}) d\bm{\theta}} = \frac{\int_{\mathbb{R}^+}\int_{\mathbb{R}^2} \frac{\delta}{\phi^2} f(\bm{y}|\bm{\theta}) d\bm{\theta} }{\int_{\mathbb{R}^+}\int_{\mathbb{R}^2} \frac{1}{\phi^2} f(\bm{y}|\bm{\theta}) d\bm{\theta}}.
\end{eqnarray}
Similarly, the optimum invariant estimator of $\phi$ under $\mathcal{G}_{SModel}$ from (\ref{S_model_group}), denoted by $\hat{\phi}_{MinRisk}$, can be obtained by solving
\begin{eqnarray}\label{Bayes_risk_Smodel_phi}
\hat{\phi}_{MinRisk}(\bm{y}) & = & \underbrace{\arg \min}_{\hat{\phi}} \int_{\bm{\Theta}} \frac{(\hat{\phi}(\bm{y}) - \phi)^2}{\phi^2}  \pi^r(\bm{\theta}|\bm{y}) d\bm{\theta}.
\end{eqnarray}
Solving, we obtain
\begin{eqnarray}
\hat{\phi}_{MinRisk}(\bm{y})  =  \frac{\int_{\mathbb{R}^+}\int_{\mathbb{R}^2} \frac{1}{\phi} f(\bm{y}|\bm{\theta}) d(d) d\delta d\phi }{\int_{\mathbb{R}^+}\int_{\mathbb{R}^2} \frac{1}{\phi^{2}} f(\bm{y}|\bm{\theta}) d(d) d\delta d\phi}.
\end{eqnarray}
\noindent
{\bf Minimaxity of optimum invariant estimators in $\bm{\Theta}$: \\}
It frequently turns out that the optimum invariant estimators are minimax \cite{Berger} (see Part III of Section 5.3.2, page 353). 
%We now show the derived optimum invariant estimators are minimax in $\bm{\Theta}$ for the considered loss function using Theorem 5.1.12 from \cite{Lehmann}\footnote{We follow the steps given in the proof of Theorem 5.1.12 from \cite{Lehmann}.}. 
Consider a sequence of prior distributions, $\pi_k$ for $\bm{\theta}$, defined on $\bm{\Theta}$ as follows
\begin{eqnarray}
\pi_k(\bm{\theta}) & = & \frac{\mathcal{I}_{(0, k)}(\phi) \mathcal{I}_{(-k, k)}(d) \mathcal{I}_{(-k, k)}(\delta)}{N_k},
\end{eqnarray}
for $k = 1, 2, \cdots, $ and $N_k = \int_{\bm{\Theta}} \mathcal{I}_{(0, k)}(\phi) \mathcal{I}_{(-k, k)}(d) \mathcal{I}_{(-k, k)}(\delta) d\bm{\theta}$. The support of $\pi_k$ is given by
\begin{eqnarray}
\bm{\Theta}_k & = & \{ (\phi, d, \delta) : \phi \in (0, k), d \in (-k, k), \delta \in (-k, k) \}.
\end{eqnarray}
The optimal Bayes estimator of $\delta$, denoted by $\hat{\delta}_{\pi_k}$, for $\pi_k$ and the loss function given in (\ref{loss_fn1}) is obtained by 
\begin{eqnarray}\label{Bayesian_finite}
\hat{\delta}_{\pi_k} & = & \underbrace{\arg \min}_{\hat{\delta}} \mathcal{B} (\hat{\delta}, \pi_k) = \underbrace{\arg \min}_{\hat{\delta}} \int_{\bm{\Theta}} \frac{(\hat{\delta}(\bm{y}) - \delta)^2}{\phi^2}  \frac{f(\bm{y}|\bm{\theta})\pi_k(\bm{\theta})d\bm{\theta}}{\int_{\bm{\Theta}}f(\bm{y}|\bm{\theta})\pi_k(\bm{\theta}) d\bm{\theta}} ,
\end{eqnarray}
Solving (\ref{Bayesian_finite}), we obtain
\begin{eqnarray}
\hat{\delta}_{\pi_k}(\bm{y}) & = & \frac{\int_{\bm{\Theta}_k} \frac{\delta}{\phi^2} f(\bm{y}|\bm{\theta}) d\bm{\theta} }{\int_{\bm{\Theta}_k} \frac{1}{\phi^2} f(\bm{y}|\bm{\theta}) d\bm{\theta}}.
\end{eqnarray}
As $k \rightarrow \infty$, we see that $\bm{\Theta}_k \rightarrow \bm{\Theta}$, $\hat{\delta}_{\pi_k} \rightarrow \hat{\delta}_{MinRisk}$, and
\begin{eqnarray}
\mathcal{B}(\hat{\delta}_{\pi_k}, \pi_k) & \rightarrow & \mathcal{B}(\hat{\delta}_{MinRisk}, \pi_k) = \mathcal{M}(\hat{\delta}_{MinRisk}),
\end{eqnarray}
%since $\mathcal{B}(\hat{\delta}_{MinRisk}, \pi_k)$ is constant for all $\bm{\theta} \in \bm{\Theta}$ 
since $\hat{\delta}_{MinRisk}$ is an invariant estimator of $\delta$ (see (\ref{Invariant_risk}) in Section \ref{Sec3}).

Let $\hat{\delta}_{r}$ denote an estimator of $\delta$. For the loss function given in (\ref{loss_fn1}), we have
\begin{eqnarray}
\mathcal{M}(\hat{\delta}_r) & \ge & \mathcal{B}(\hat{\delta}_r, \pi_k) \ge \mathcal{B}(\hat{\delta}_{\pi_k}, \pi_k),
\end{eqnarray}
since the optimal Bayes estimator for a prior $\pi_k$ achieves the lowest average risk. Let $k \rightarrow \infty$, we have
\begin{eqnarray}
\mathcal{M}(\hat{\delta}_r) & \ge & \lim_{k \rightarrow \infty} \mathcal{B}(\hat{\delta}_{\pi_k}, \pi_k) = \mathcal{M}(\hat{\delta}_{MinRisk}).
\end{eqnarray}
Hence, the maximum risk of any estimator of $\delta$ is greater than or equal to the maximum risk of $\hat{\delta}_{MinRisk}$. Hence, $\hat{\delta}_{MinRisk}$ is a minimax estimator of $\delta$ for the skew-normalized loss function defined in (\ref{loss_fn1}). Similarly, we can show that $\hat{\phi}_{MinRisk}$ is a minimax estimator of $\phi$ for the skew-normalized loss function defined in (\ref{loss_fn2}).

\noindent
{\bf Minimaxity of optimum invariant estimators in $\bm{\Theta}^*$: \\}
Marchand and Strawderman \cite{marchand2012} gave conditions on $\bar{\mathcal{G}}_{SModel}$ defined in (\ref{S_model_param_group2}), under which the optimum invariant estimator remains minimax in the restricted parameter space, $\bm{\Theta}^*$. If there exists a sequence $\bar{g}_{a_k, b_k, c_k} \in \bar{\mathcal{G}}_{SModel}$ from (\ref{S_model_param_group2}), such that 
\begin{eqnarray}
\bar{g}_{a_k, b_k, c_k}(\bm{\Theta}^*) & \subseteq & \bar{g}_{a_{k+1}, b_{k+1}, c_{k+1}}(\bm{\Theta}^*), \label{condition1} \\
\bigcup_{k} \bar{g}_{a_k, b_k, c_k}(\bm{\Theta}^*) & = & \bm{\Theta}, \label{condition2}
\end{eqnarray}
where $\bar{g}_{a_k, b_k, c_k}(\bm{\Theta}^*) = \{ \bar{g}_{a_k, b_k, c_k}(\bm{\theta}) : \bm{\theta} \in \bm{\Theta}^* \}$, then $\hat{\delta}_{MinRisk}$ and $\hat{\phi}_{MinRisk}$ remains minimax in $\bm{\Theta}^*$ for the considered loss functions (See Theorem 1 of \cite{marchand2012}). Consider the sequence of transformations from $\bar{\mathcal{G}}_{SModel}$, defined as $\bar{g}_{a_k, b_k, c_k} = \bar{g}_{1, -k, 0}$ for $k = 1, 2, \cdots$. We have
\begin{eqnarray}
\bar{g}_{a_k, b_k, c_k}(\bm{\Theta}^*) & = & \{(\phi, d, \delta): \phi \in \mathbb{R}^+, d \ge (-k/\phi), \delta \in \mathbb{R} \}, \\
\bar{g}_{a_{k+1}, b_{k+1}, c_{k+1}}(\bm{\Theta}^*) & = & \{(\phi, d, \delta): \phi \in \mathbb{R}^+, d \ge (-(k+1)/\phi), \delta \in \mathbb{R} \}.
\end{eqnarray}
For this sequence of transformations, (\ref{condition1}) and (\ref{condition2}) are satisfied. Hence, the optimum invariant estimators $\hat{\delta}_{MinRisk}$ and $\hat{\phi}_{MinRisk}$ remain minimax in $\bm{\Theta}^*$ for the skew-normalized squared error loss functions defined in (\ref{loss_fn1}) and (\ref{loss_fn2}), respectively.
\end{proof}

\begin{figure}[t]
	\centering
	\begin{subfigure}[b]{0.48\columnwidth}
		\centering
		\includegraphics[height = 2.5 in, width = 0.9\columnwidth]{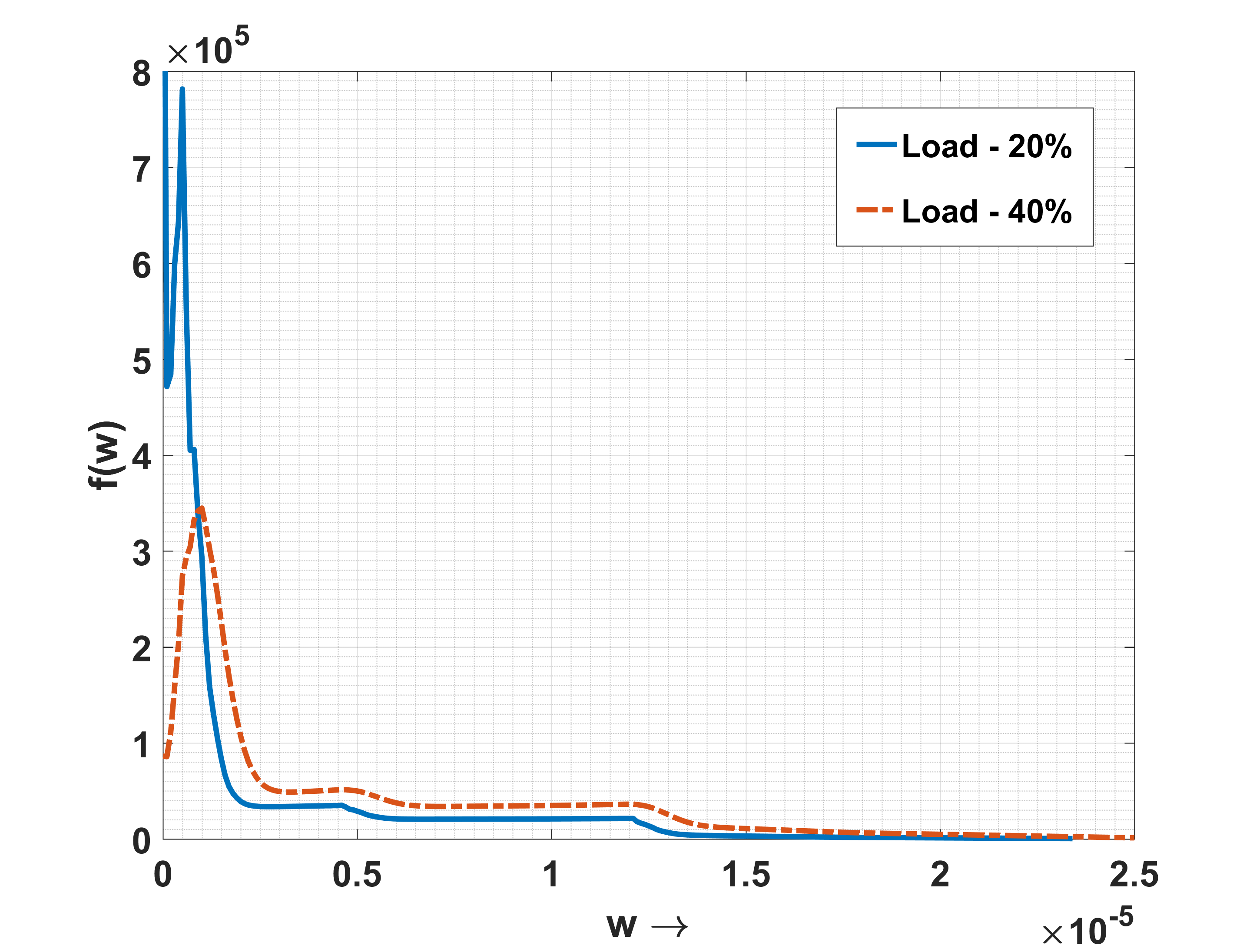}
		\caption{ }
	\end{subfigure}	
	~ %add desired spacing between images, e. g. ~, \quad, \qquad, \hfill etc. 
	%(or a blank line to force the subfigure onto a new line)
	\begin{subfigure}[b]{0.48\columnwidth}
		\centering
		\includegraphics[height = 2.5 in, width = 0.9\columnwidth]{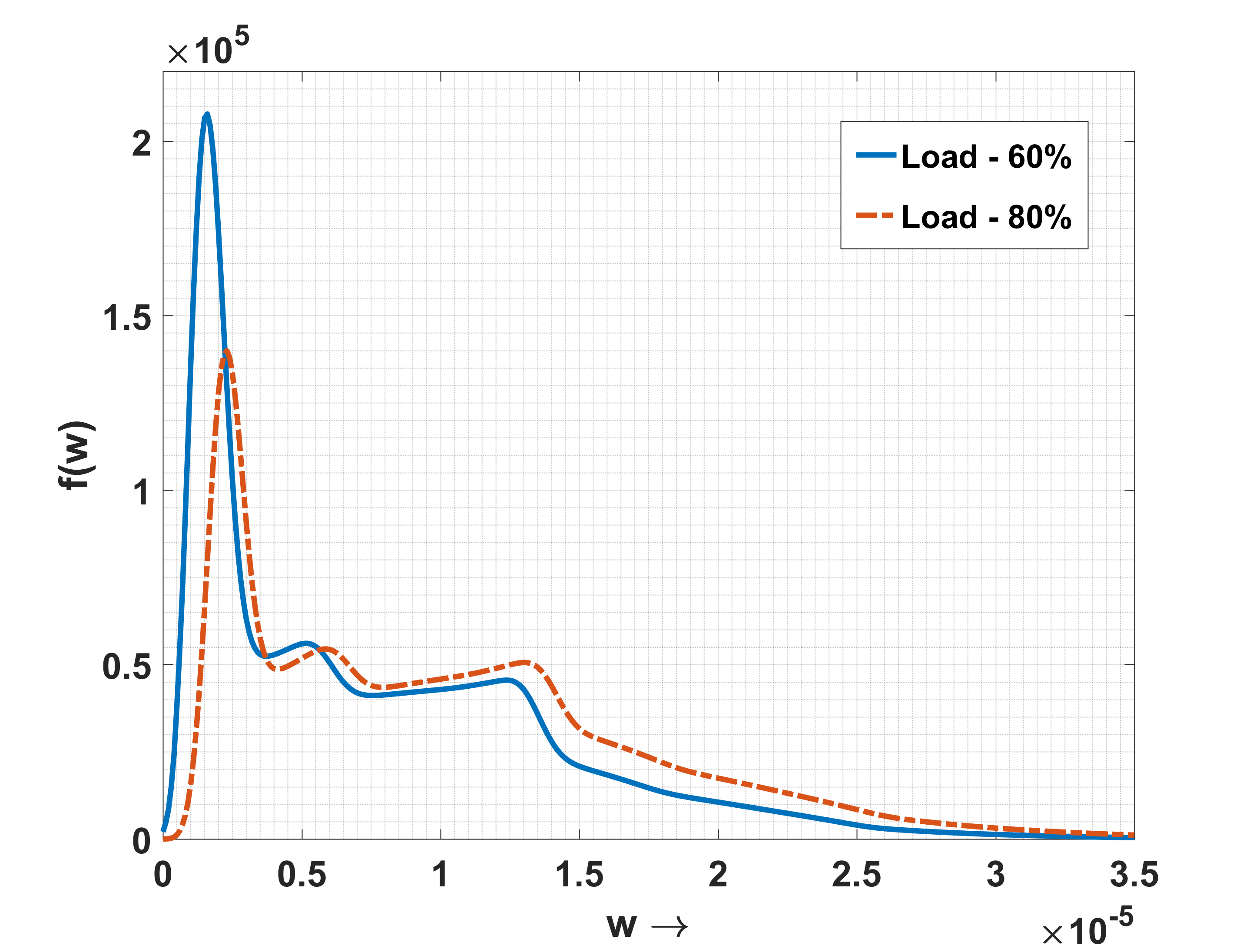}
		\caption{ }
	\end{subfigure}

	\begin{subfigure}[b]{0.48\columnwidth}
	\centering
	\includegraphics[height = 2.5 in, width = 0.9\columnwidth]{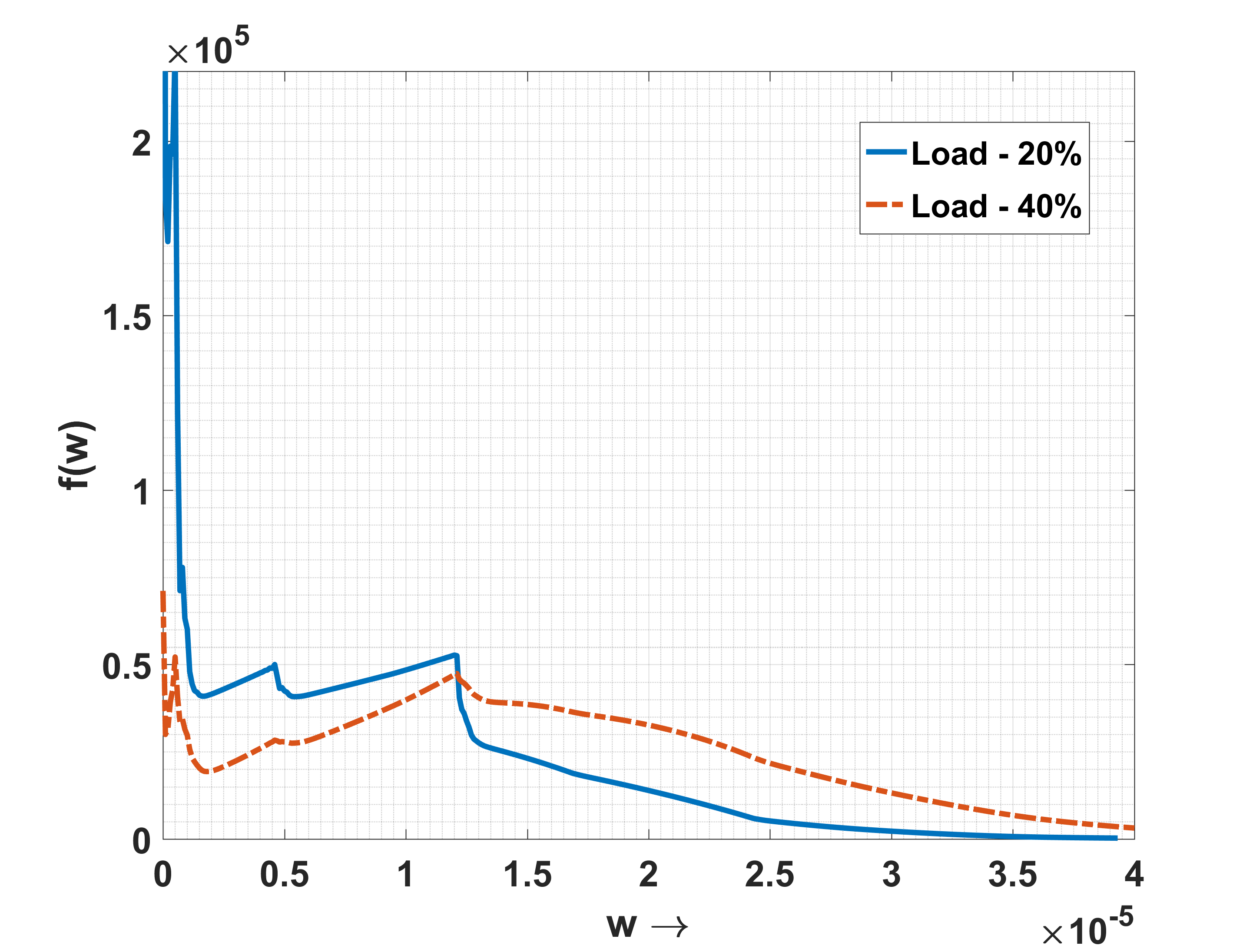}
	\caption{ }
	\end{subfigure}	
	~ %add desired spacing between images, e. g. ~, \quad, \qquad, \hfill etc. 
	%(or a blank line to force the subfigure onto a new line)
	\begin{subfigure}[b]{0.48\columnwidth}
		\centering
		\includegraphics[height = 2.5 in, width = 0.9\columnwidth]{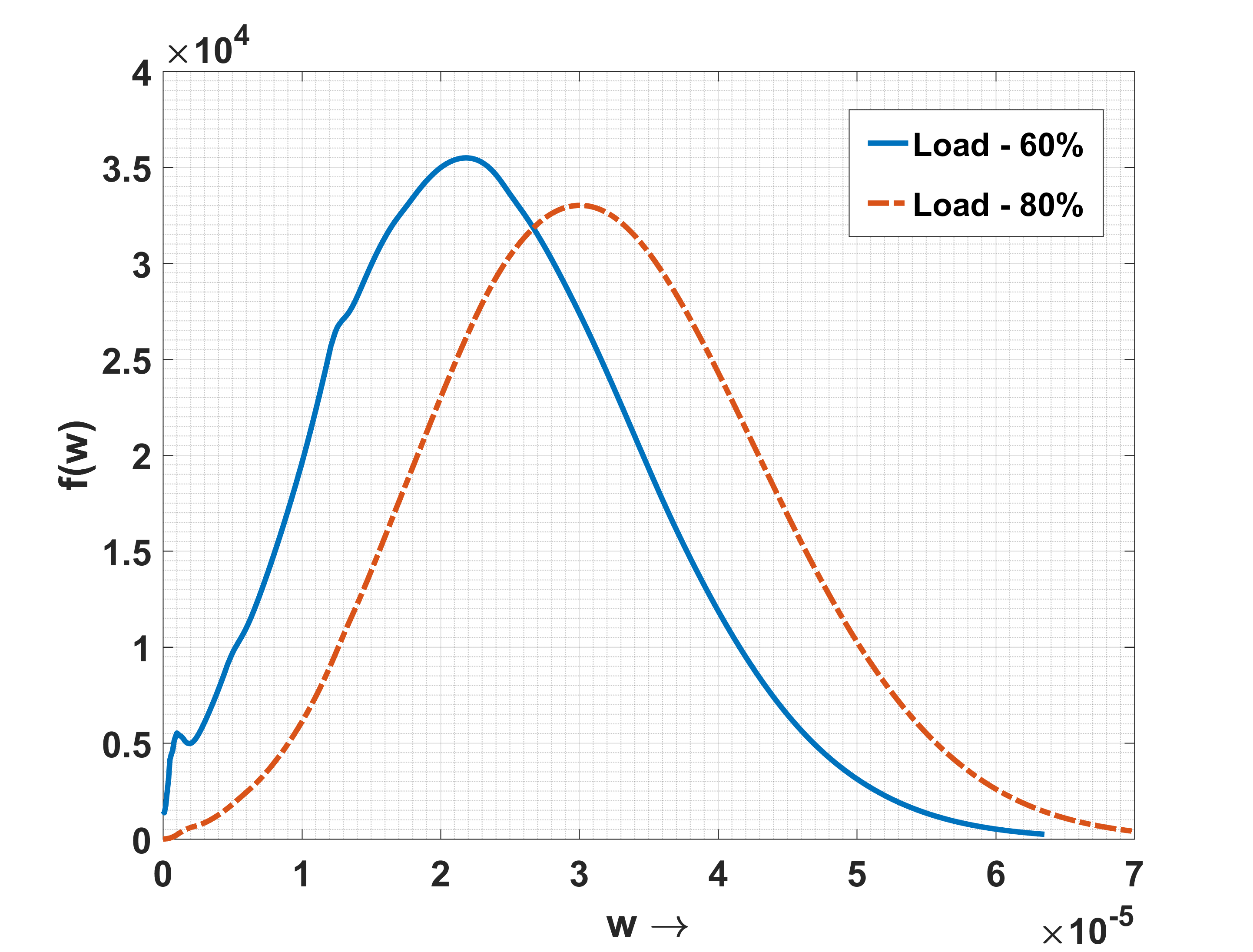}
		\caption{ }
	\end{subfigure}		
	\caption{Empirical pdf of queuing delays for $10$ switches between master and slave node for various loads, (a) TM-1 under 20\%, 40\% load, (b) TM-1 under 60\%, 80\% load,  (c) TM-2 under 20\%, 40\% load, (d) TM-2 under 60\%, 80\% load.}\label{Empirical_pdf}
\end{figure}

\begin{figure}[t]
	\centering
	\begin{subfigure}[b]{0.48\columnwidth}
		\centering
		\includegraphics[height = 2.5 in, width = 0.9\columnwidth]{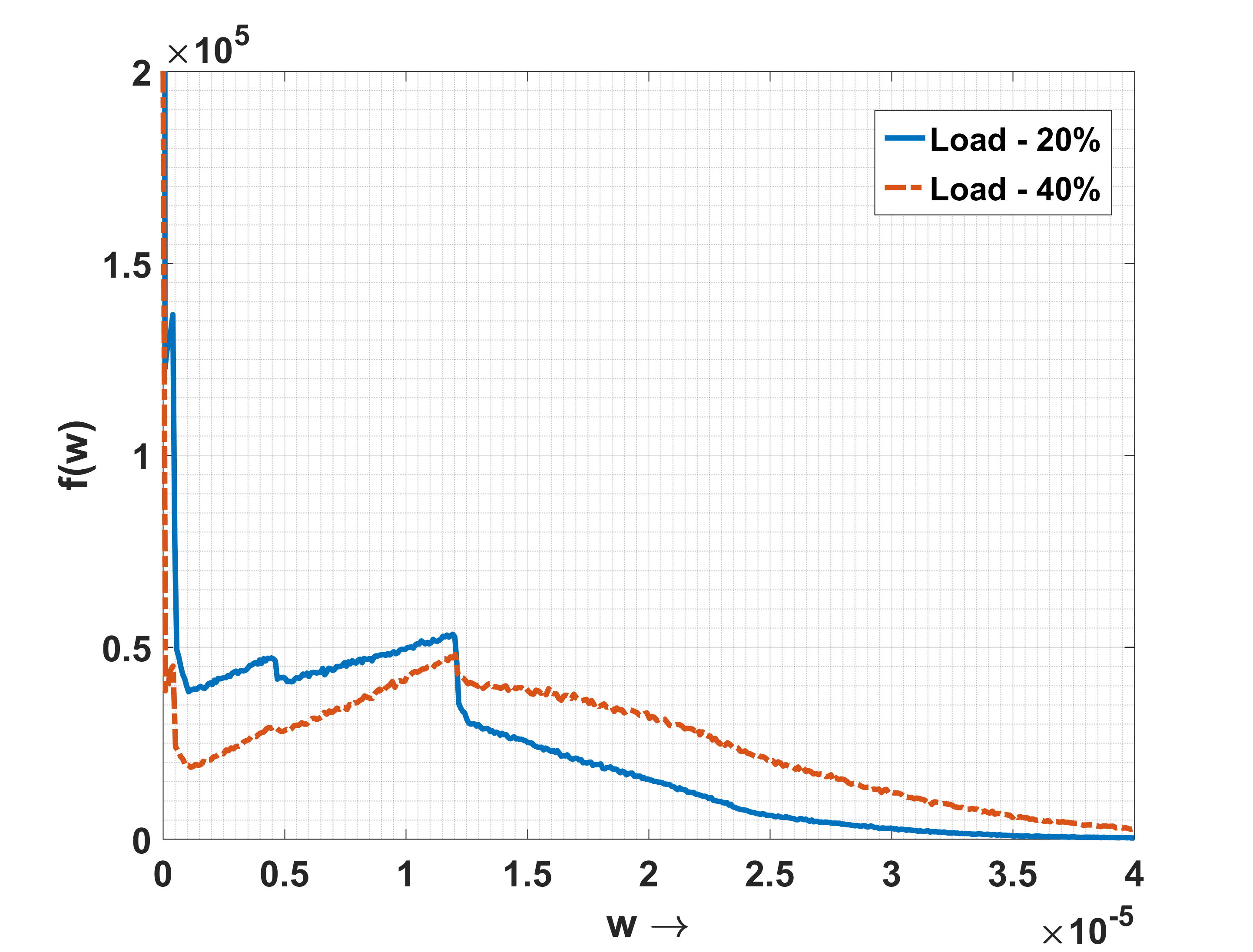}
		\caption{ }
	\end{subfigure}	
	~ %add desired spacing between images, e. g. ~, \quad, \qquad, \hfill etc. 
	%(or a blank line to force the subfigure onto a new line)
	\begin{subfigure}[b]{0.48\columnwidth}
		\centering
		\includegraphics[height = 2.5 in, width = 0.9\columnwidth]{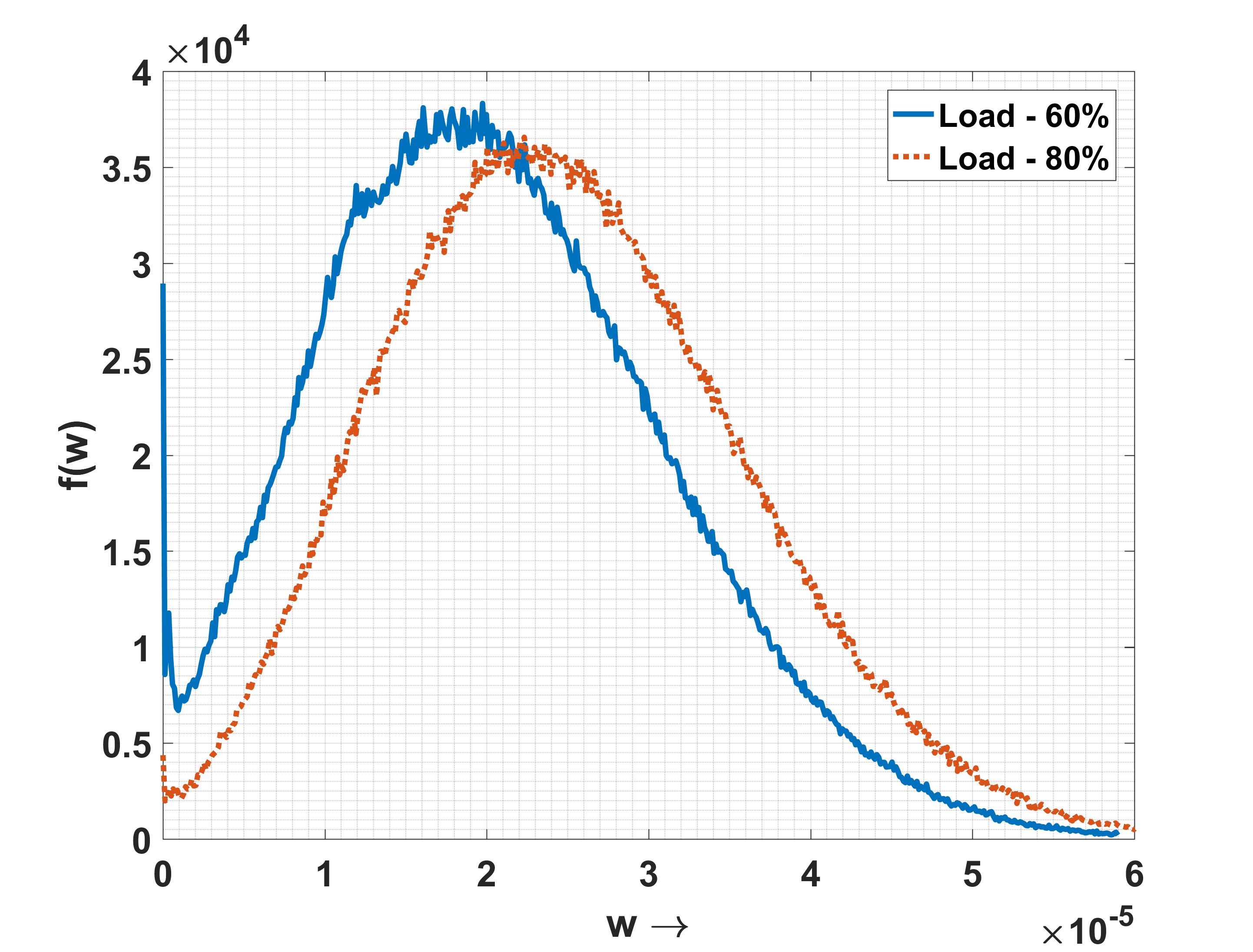}
		\caption{ }
	\end{subfigure}

%	\begin{subfigure}[b]{0.48\columnwidth}
%		\centering
%		\includegraphics[height = 2.5 in, width = 0.9\columnwidth]{TM2_load20_load40}
%		\caption{ }
%	\end{subfigure}	
%	~ %add desired spacing between images, e. g. ~, \quad, \qquad, \hfill etc. 
%	%(or a blank line to force the subfigure onto a new line)
%	\begin{subfigure}[b]{0.48\columnwidth}
%		\centering
%		\includegraphics[height = 2.5 in, width = 0.9\columnwidth]{TM2_load60_load80}
%		\caption{ }
%	\end{subfigure}		
	\caption{Empirical pdf of queuing delays for $10$ switches between master and slave node for various loads, (a) EG-TM1 under 20\%, 40\% load, (b) EG-TM1 under 60\%, 80\% load.}\label{Empirical_pdf_sg}
\end{figure}

\begin{figure}[t]
	\centering
	\begin{subfigure}[b]{0.48\columnwidth}
		\centering
		\includegraphics[height = 2.5 in, width = 0.9\columnwidth]{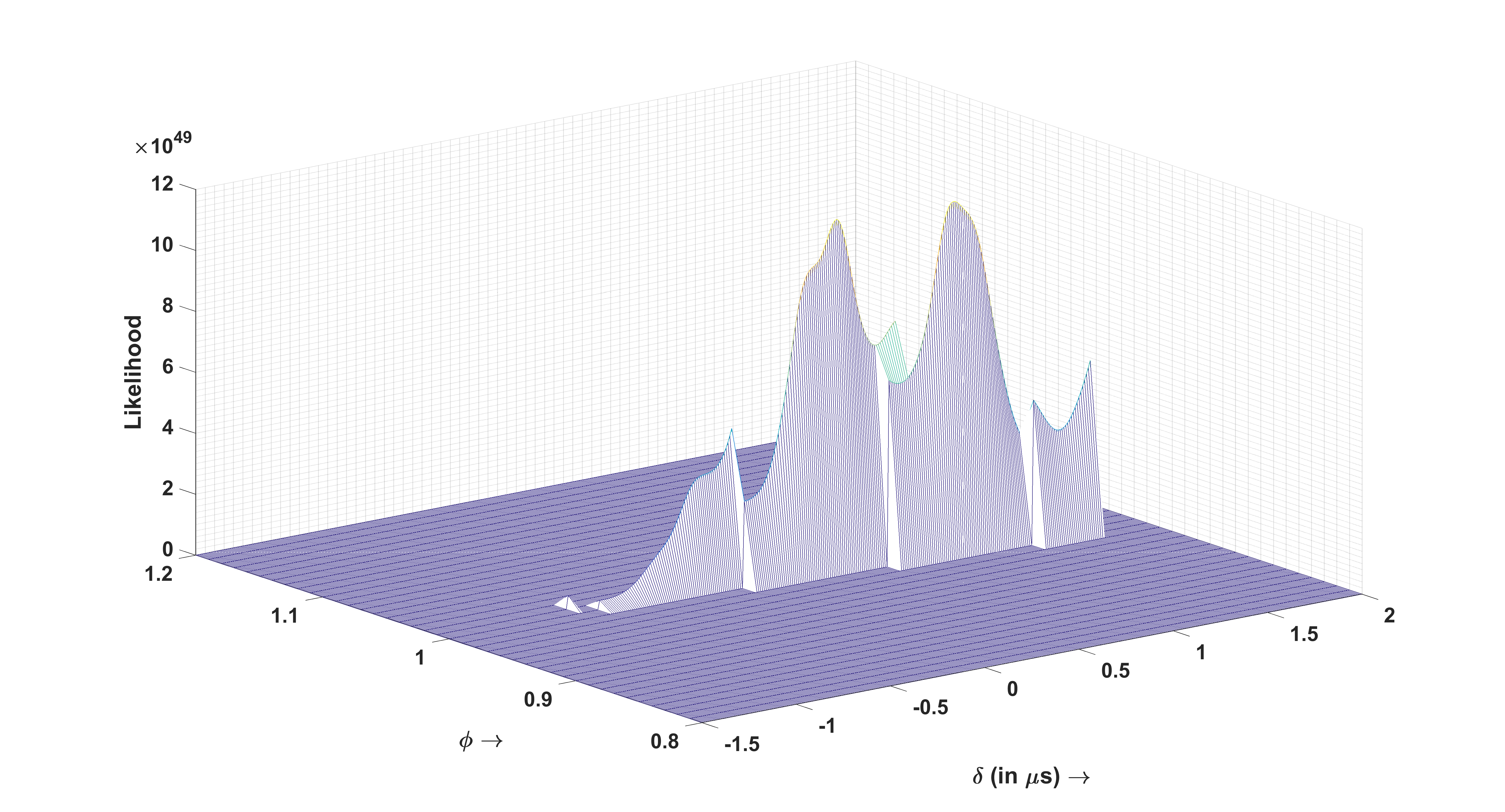}
		\caption{ }
	\end{subfigure}	
	~ %add desired spacing between images, e. g. ~, \quad, \qquad, \hfill etc. 
	%(or a blank line to force the subfigure onto a new line)
	\begin{subfigure}[b]{0.48\columnwidth}
		\centering
		\includegraphics[height = 2.5 in, width = 0.9\columnwidth]{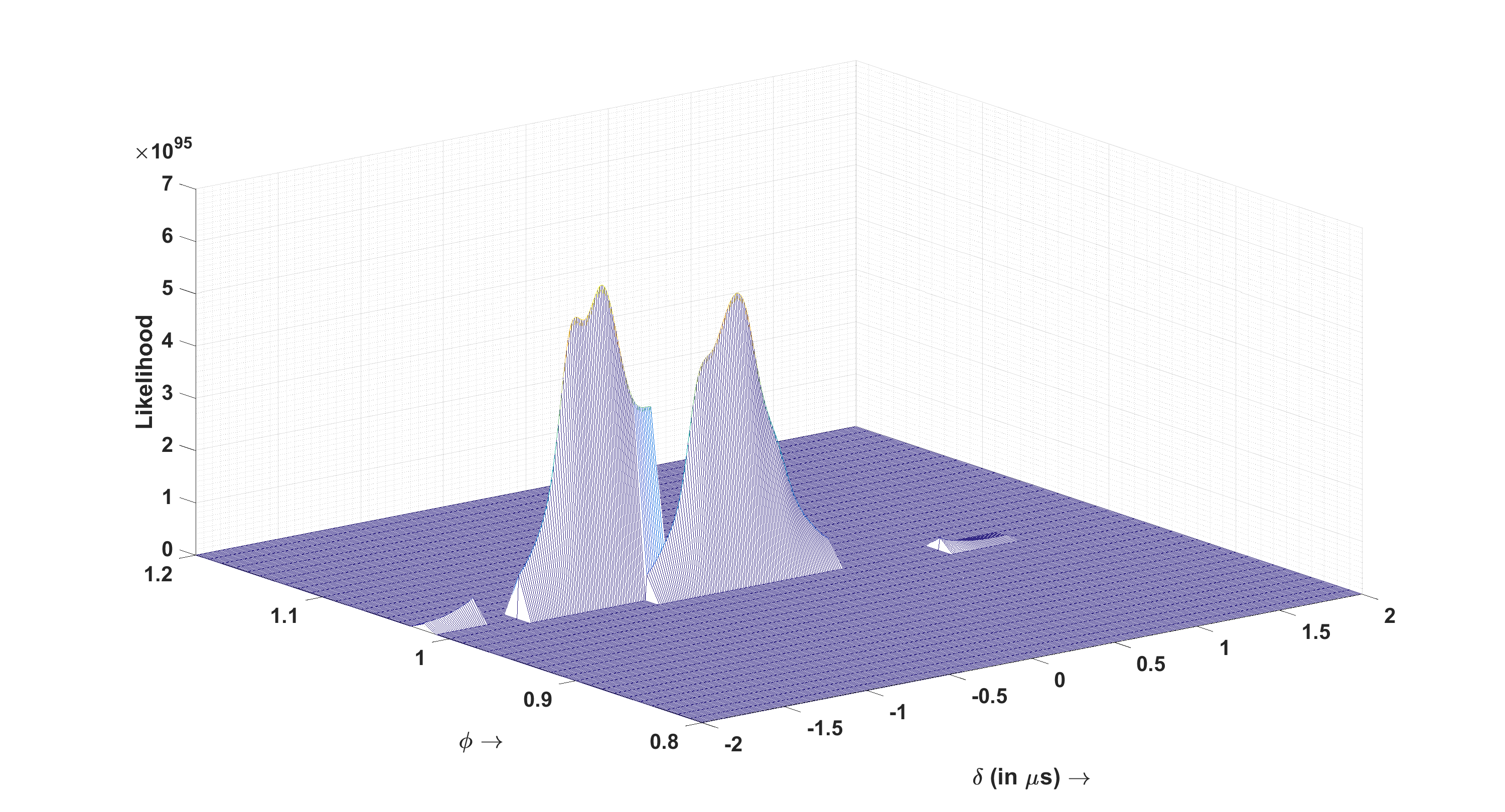}
		\caption{ }
	\end{subfigure}	
	
	\caption{Likelihood function for various values of the parameter for TM-1 under 40\% load for $\phi = 1$, $\delta = 0$ for (a) $P = 5$, (b) $P = 10$.}\label{LMLE_LLF}
\end{figure}

\begin{figure}[t]
	\centering
	\begin{subfigure}[b]{0.48\columnwidth}
		\centering
		\includegraphics[height = 2.5 in, width = \columnwidth]{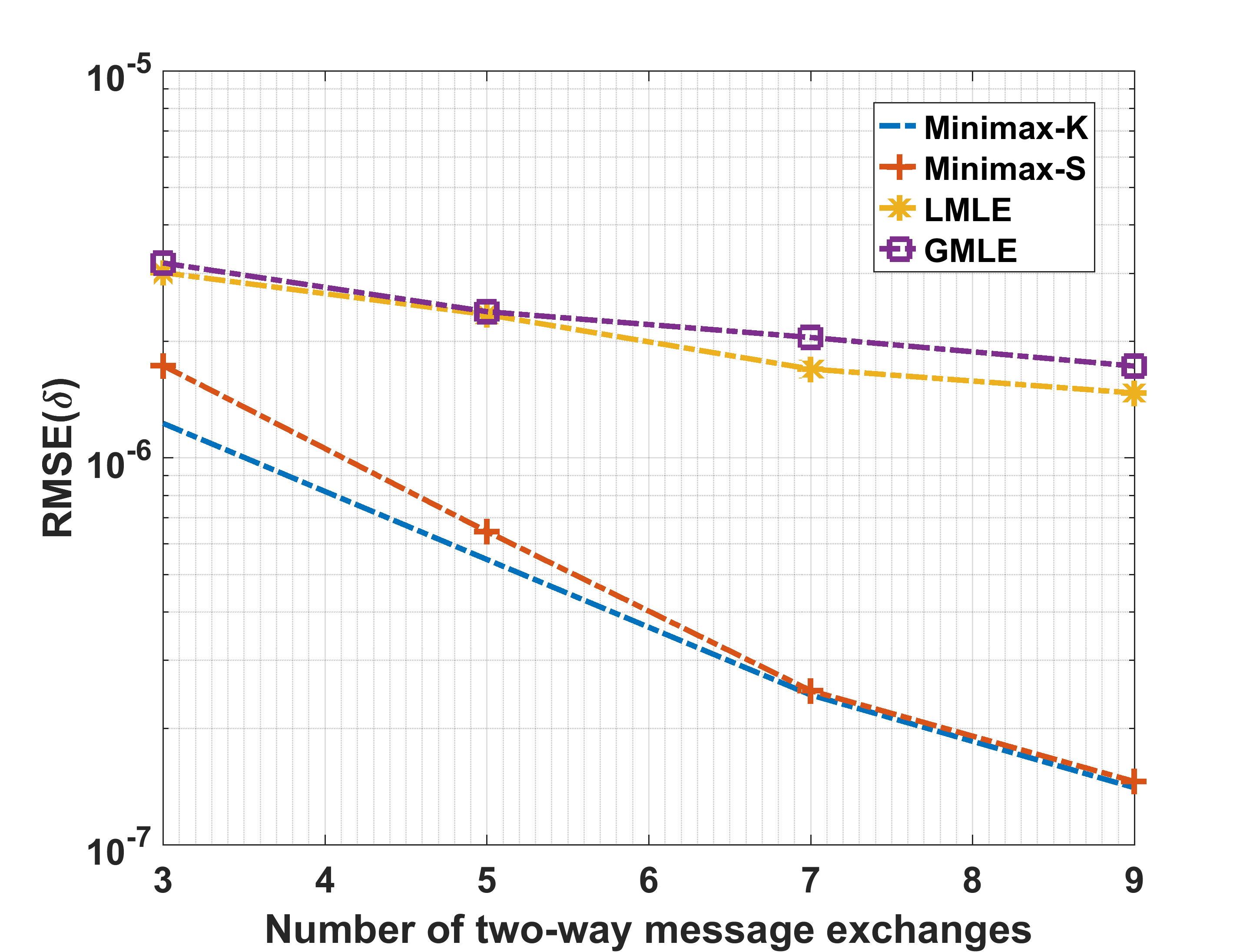}
		\caption{ }
	\end{subfigure}	
	~ %add desired spacing between images, e. g. ~, \quad, \qquad, \hfill etc. 
	%(or a blank line to force the subfigure onto a new line)
	\begin{subfigure}[b]{0.48\columnwidth}
		\centering
		\includegraphics[height = 2.5 in, width = \columnwidth]{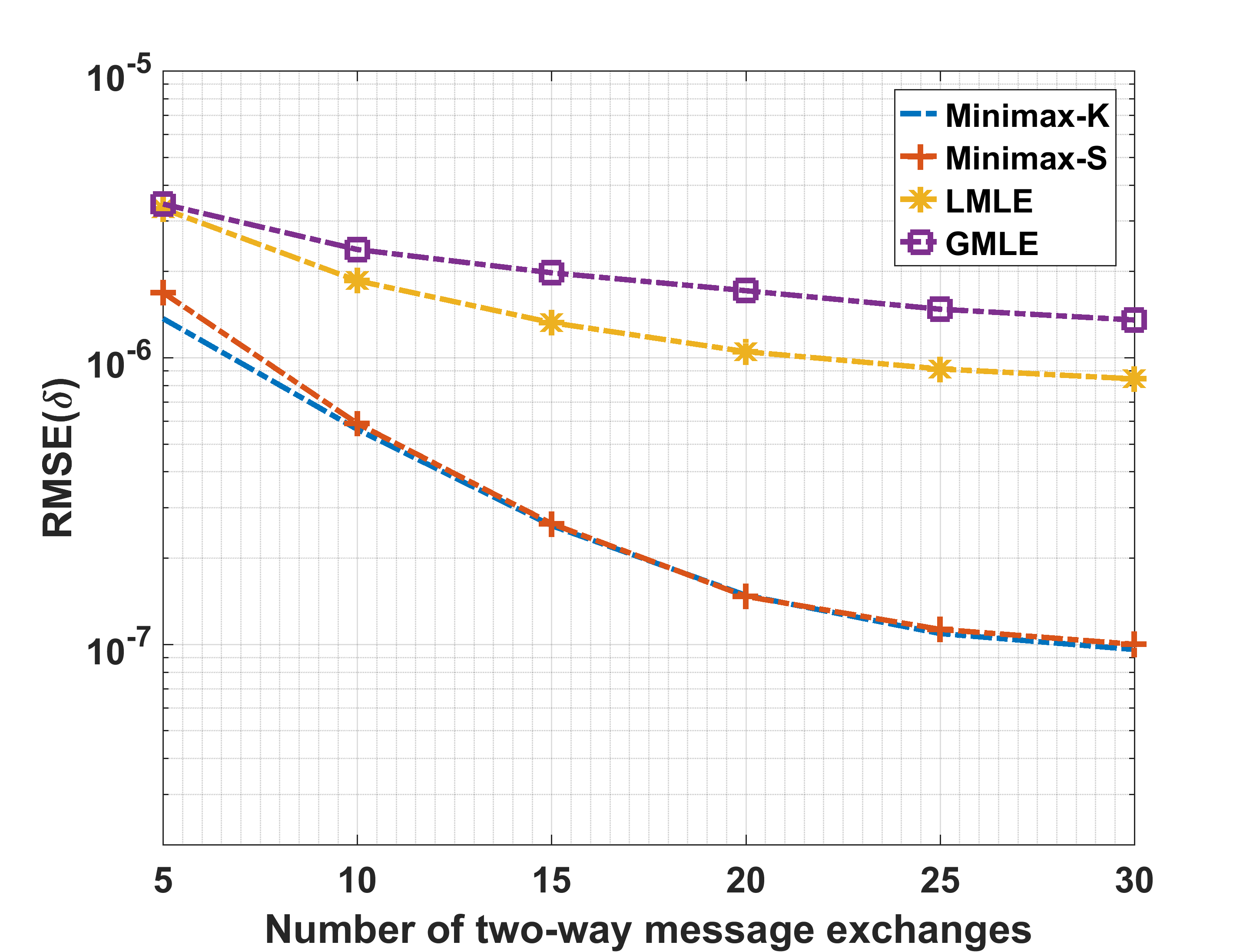}
		\caption{ }
	\end{subfigure}

	\begin{subfigure}[b]{0.48\columnwidth}
		\centering
		\includegraphics[height = 2.5 in, width = \columnwidth]{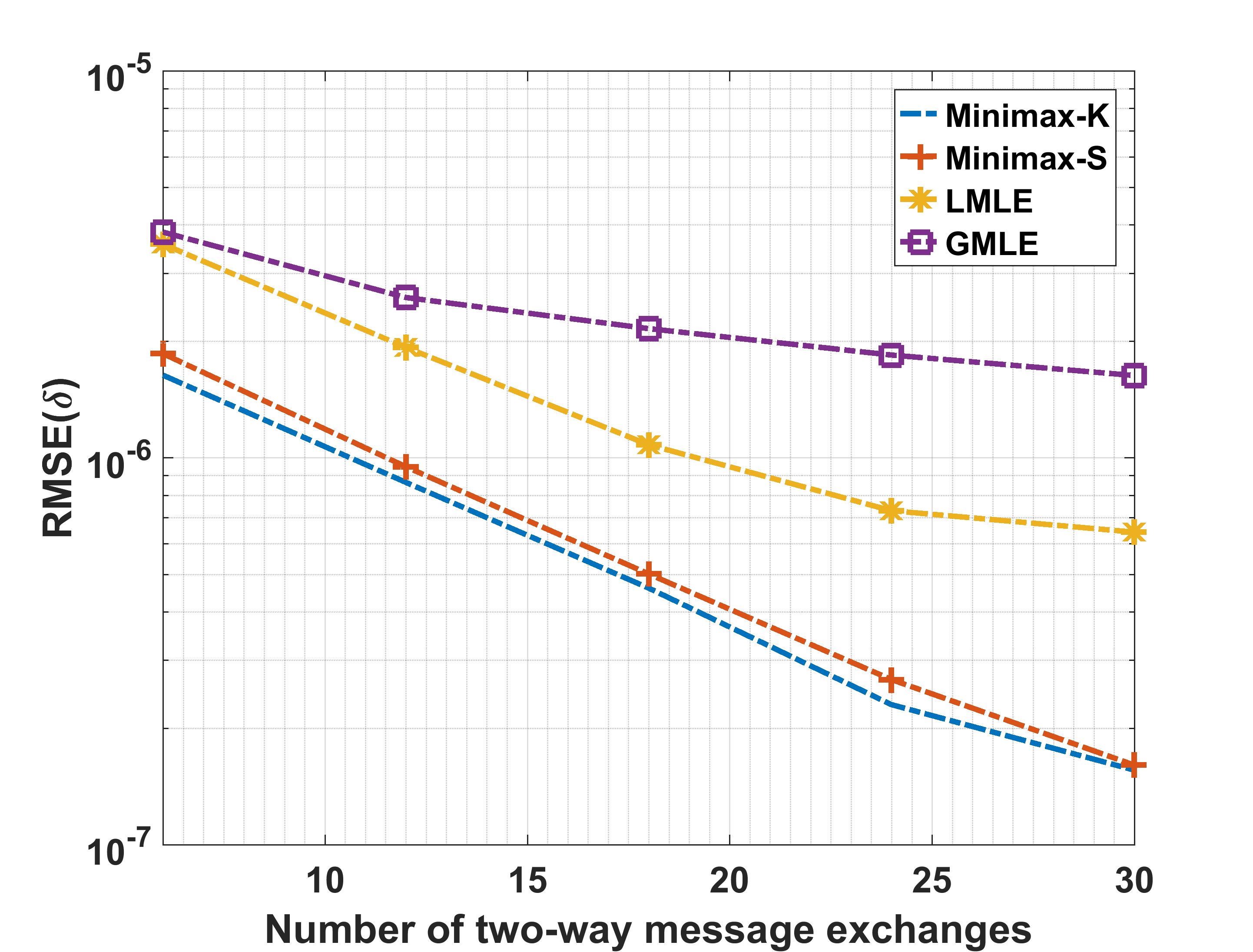}
		\caption{ }
	\end{subfigure}	
	~ %add desired spacing between images, e. g. ~, \quad, \qquad, \hfill etc. 
	%(or a blank line to force the subfigure onto a new line)
	\begin{subfigure}[b]{0.48\columnwidth}
		\centering
		\includegraphics[height = 2.5 in, width = \columnwidth]{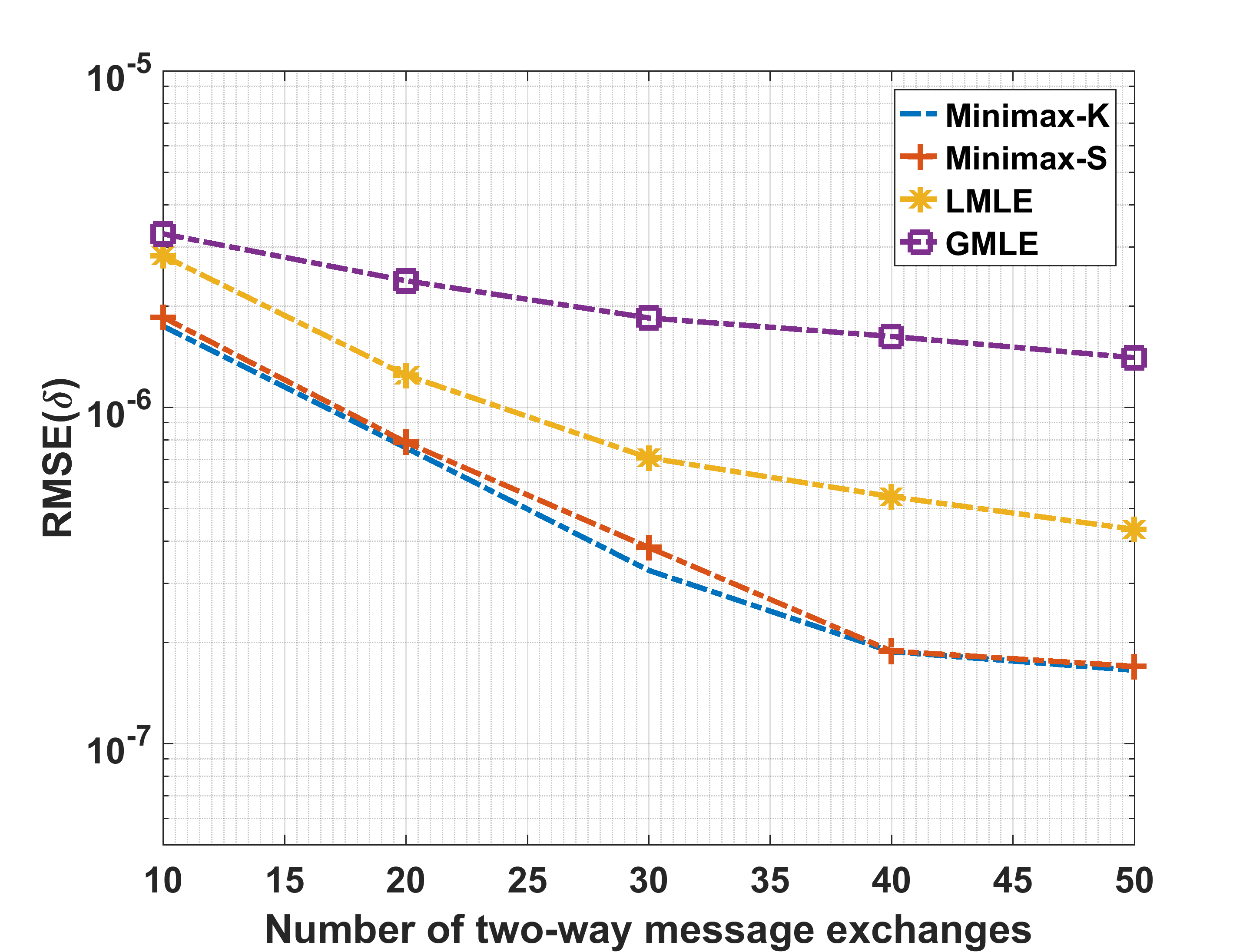}
		\caption{ }
	\end{subfigure}		
	\caption{RMSE of clock offset for various estimation schemes under TM-1 for various loads and $\{\phi, d, \delta \} = \{1, 2 \mu s, 2 \mu s\}$, (a) 20\% load, (b) 40\% load, (c) 60\% load, (d) 80\% load.}\label{rmse_TM1_offset_results}
\end{figure}

\begin{figure}[t]
	\centering
	\begin{subfigure}[b]{0.48\columnwidth}
		\centering
		\includegraphics[height = 2.5 in, width = \columnwidth]{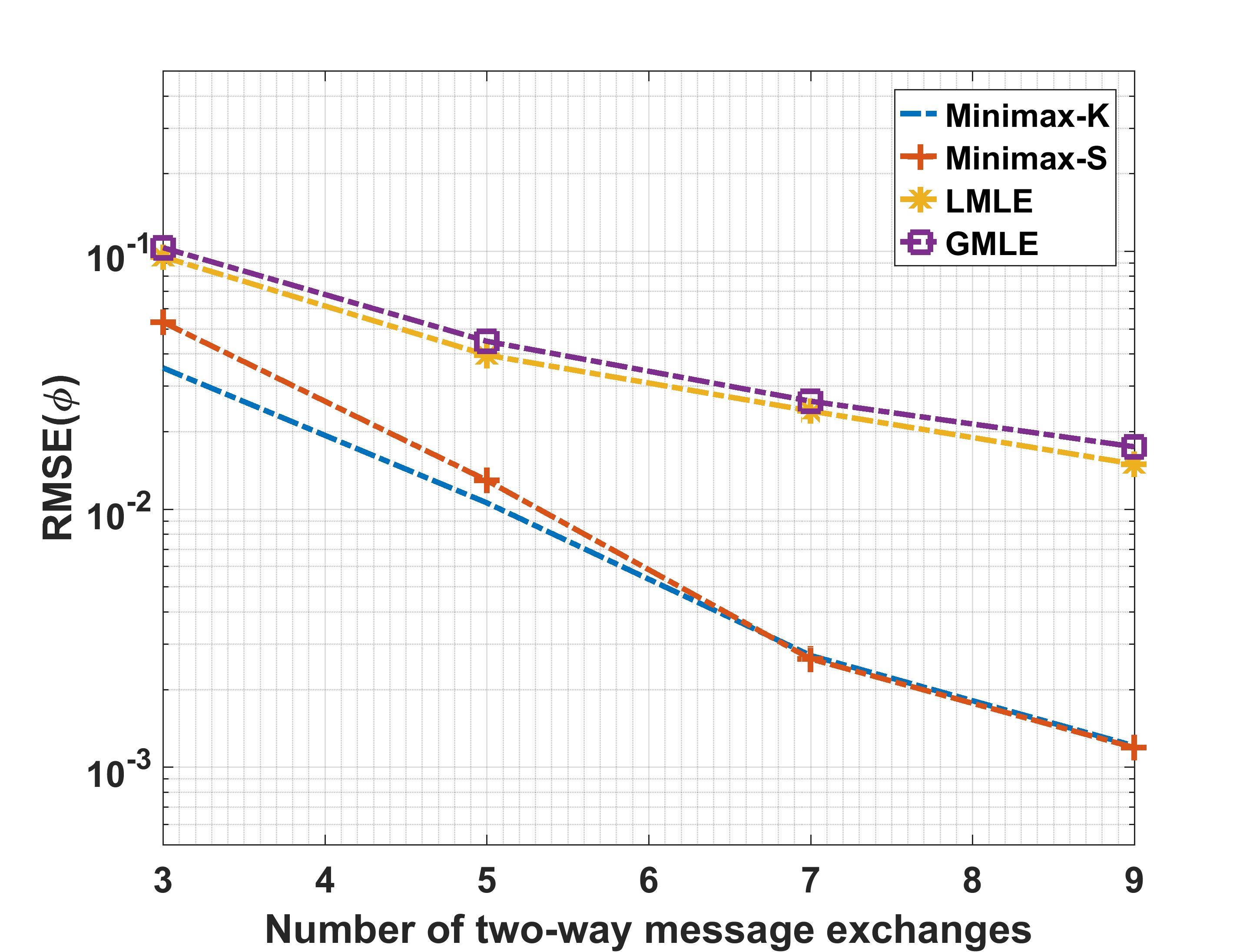}
		\caption{ }
	\end{subfigure}	
	~ %add desired spacing between images, e. g. ~, \quad, \qquad, \hfill etc. 
	%(or a blank line to force the subfigure onto a new line)
	\begin{subfigure}[b]{0.48\columnwidth}
		\centering
		\includegraphics[height = 2.5 in, width = \columnwidth]{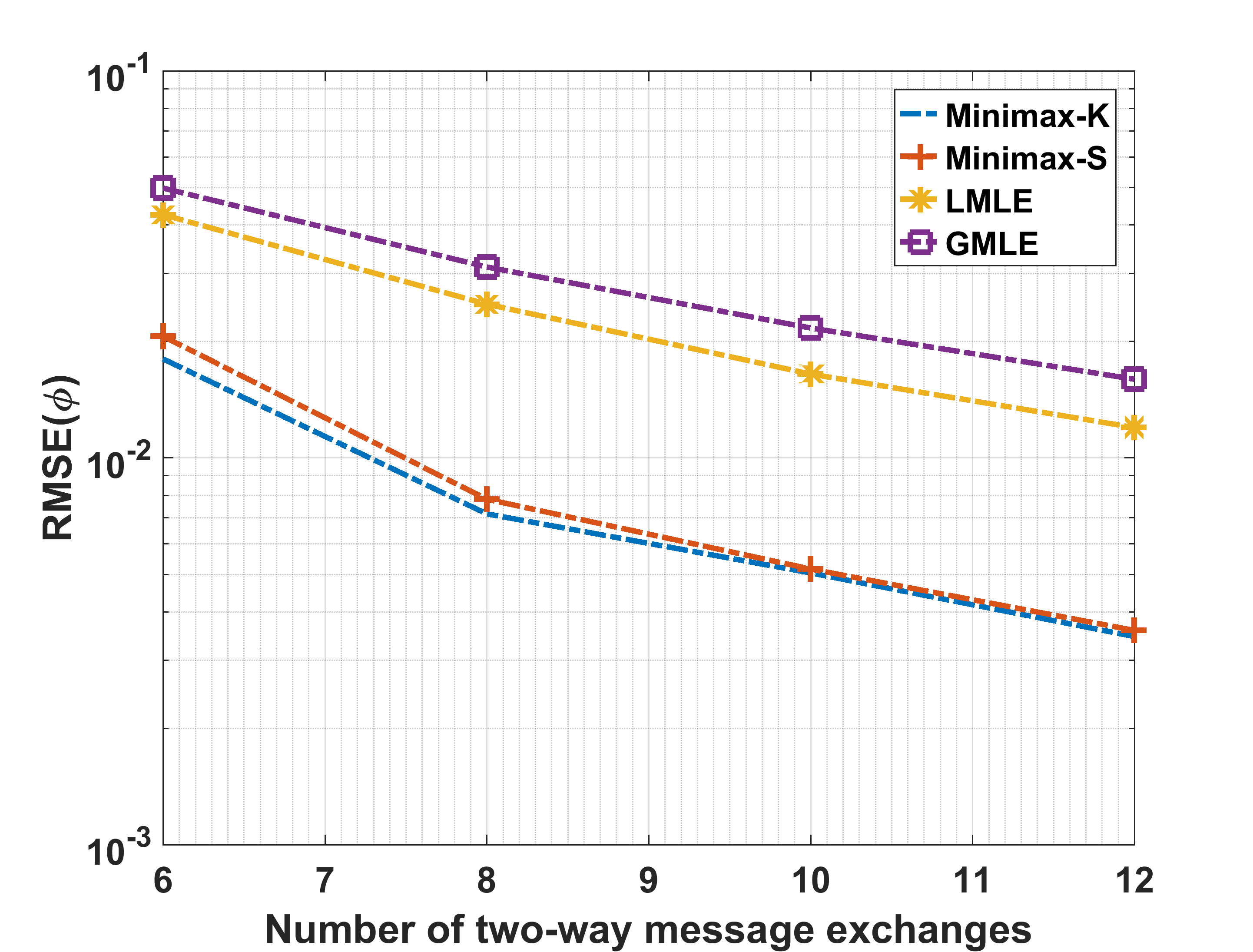}
		\caption{ }
	\end{subfigure}

	\begin{subfigure}[b]{0.48\columnwidth}
		\centering
		\includegraphics[height = 2.5 in, width = \columnwidth]{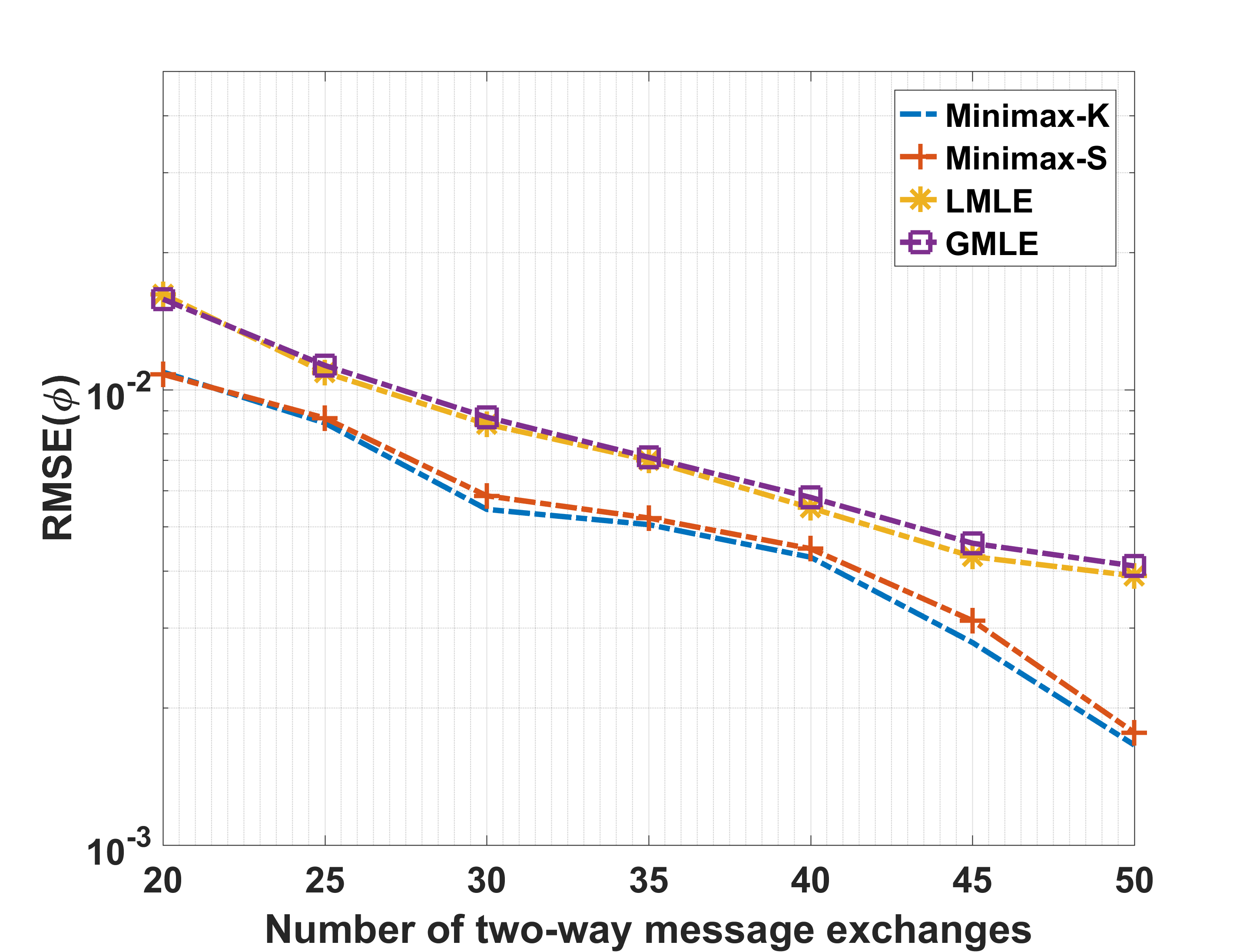}
		\caption{ }
	\end{subfigure}	
	~ %add desired spacing between images, e. g. ~, \quad, \qquad, \hfill etc. 
	%(or a blank line to force the subfigure onto a new line)
	\begin{subfigure}[b]{0.48\columnwidth}
		\centering
		\includegraphics[height = 2.5 in, width = \columnwidth]{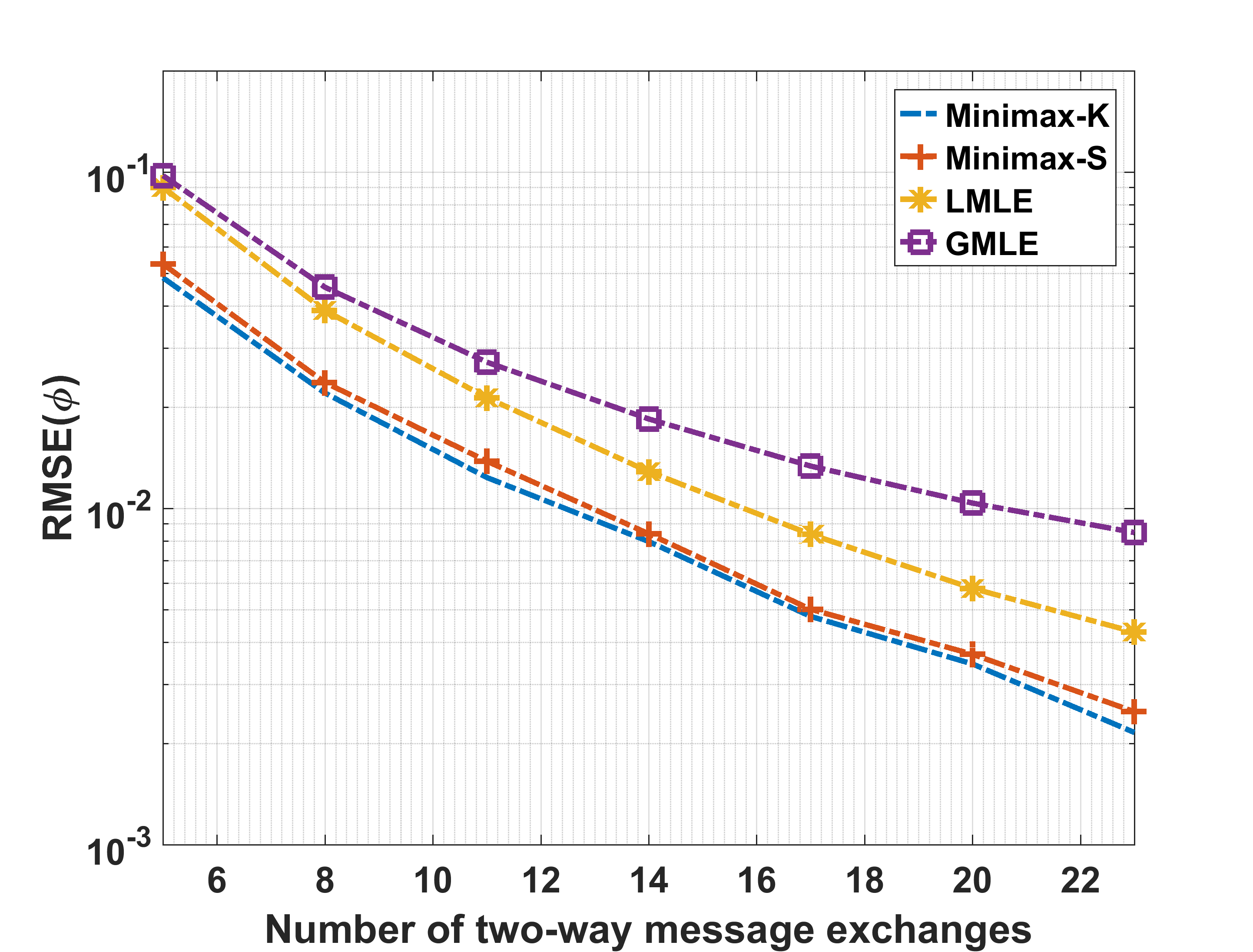}
		\caption{ }
	\end{subfigure}		
	\caption{RMSE of clock skew for various estimation schemes under TM-1 for various loads and $\{\phi, d, \delta \} = \{1, 2 \mu s, 2 \mu s\}$, (a) 20\% load, (b) 40\% load, (c) 60\% load, (d) 80\% load.}\label{rmse_TM1_skew_results}
\end{figure}

\begin{figure}[t]
	\centering
	\begin{subfigure}[b]{0.48\columnwidth}
		\centering
		\includegraphics[height = 2.5 in, width = \columnwidth]{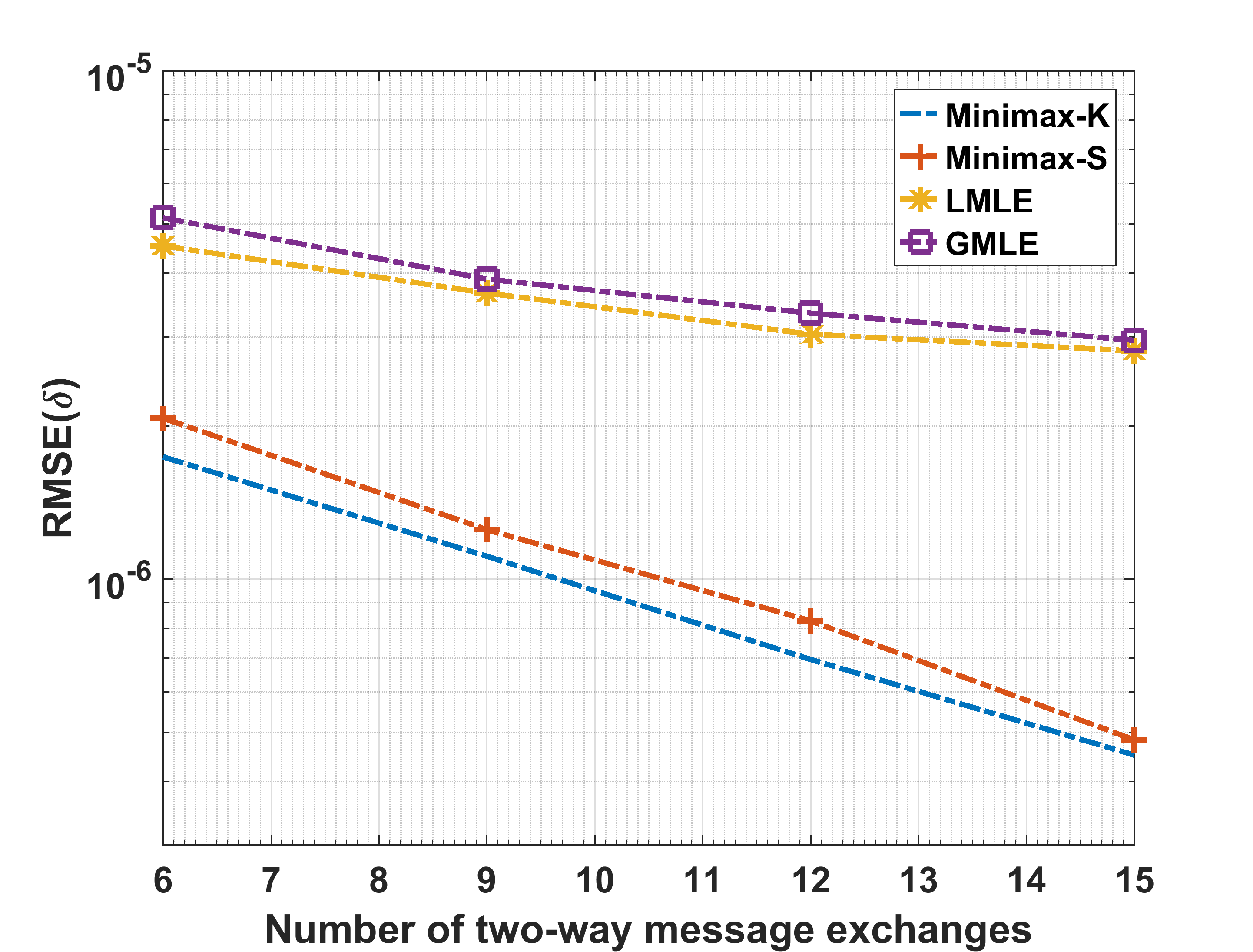}
		\caption{ }
	\end{subfigure}	
	~ %add desired spacing between images, e. g. ~, \quad, \qquad, \hfill etc. 
	%(or a blank line to force the subfigure onto a new line)
	\begin{subfigure}[b]{0.48\columnwidth}
		\centering
		\includegraphics[height = 2.5 in, width = \columnwidth]{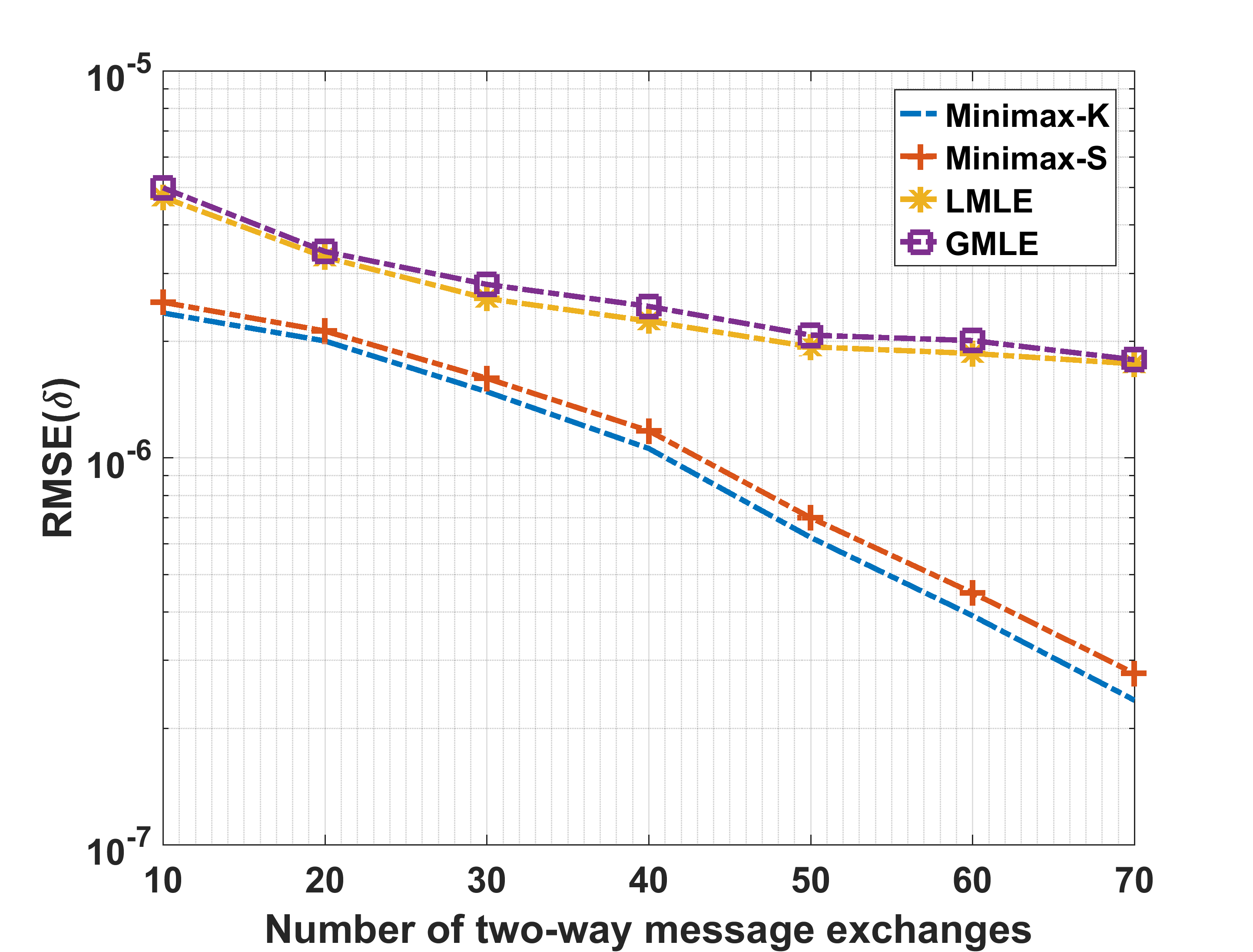}
		\caption{ }
	\end{subfigure}

	\begin{subfigure}[b]{0.48\columnwidth}
		\centering
		\includegraphics[height = 2.5 in, width = \columnwidth]{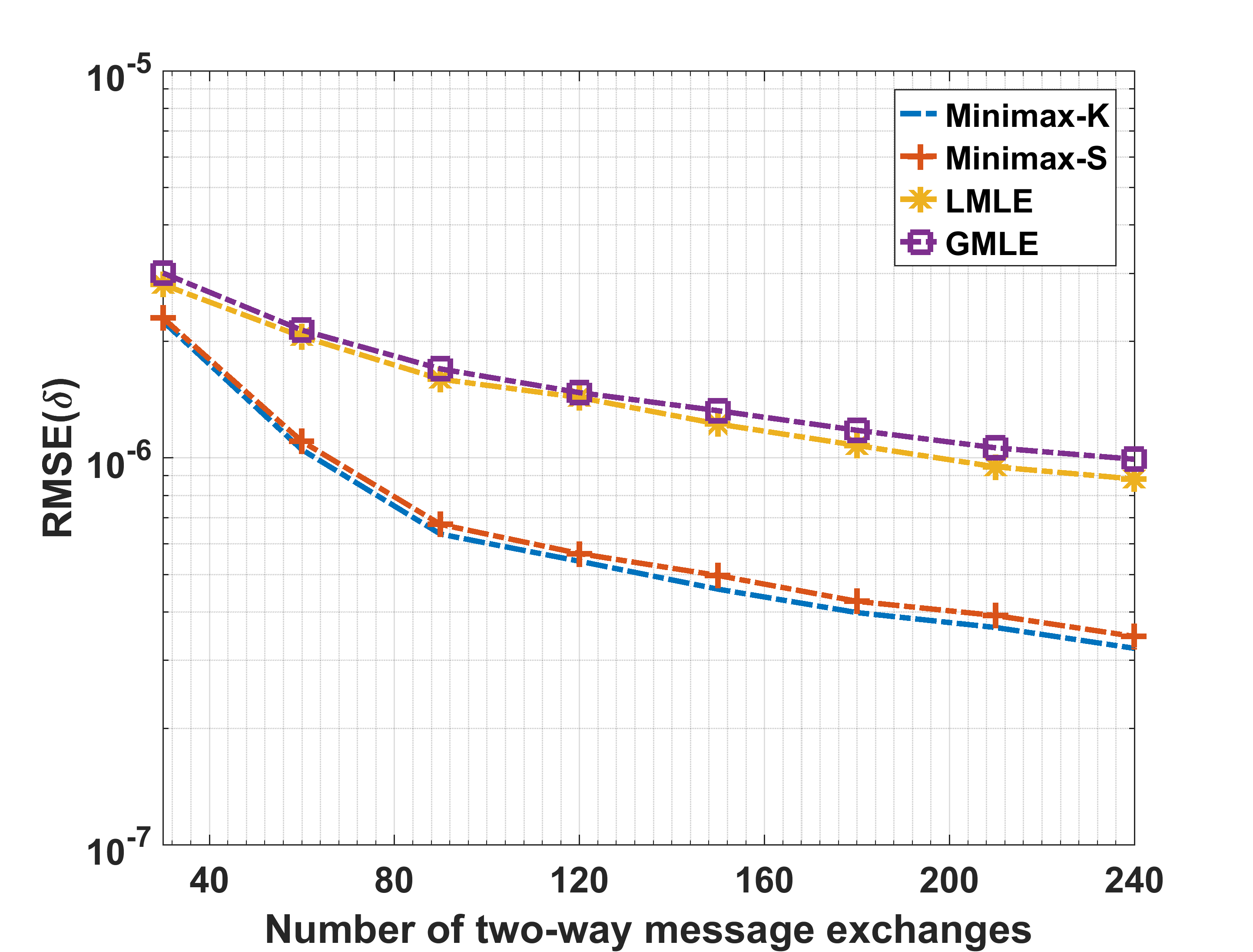}
		\caption{ }
	\end{subfigure}	
	~ %add desired spacing between images, e. g. ~, \quad, \qquad, \hfill etc. 
	%(or a blank line to force the subfigure onto a new line)
	\begin{subfigure}[b]{0.48\columnwidth}
		\centering
		\includegraphics[height = 2.5 in, width = \columnwidth]{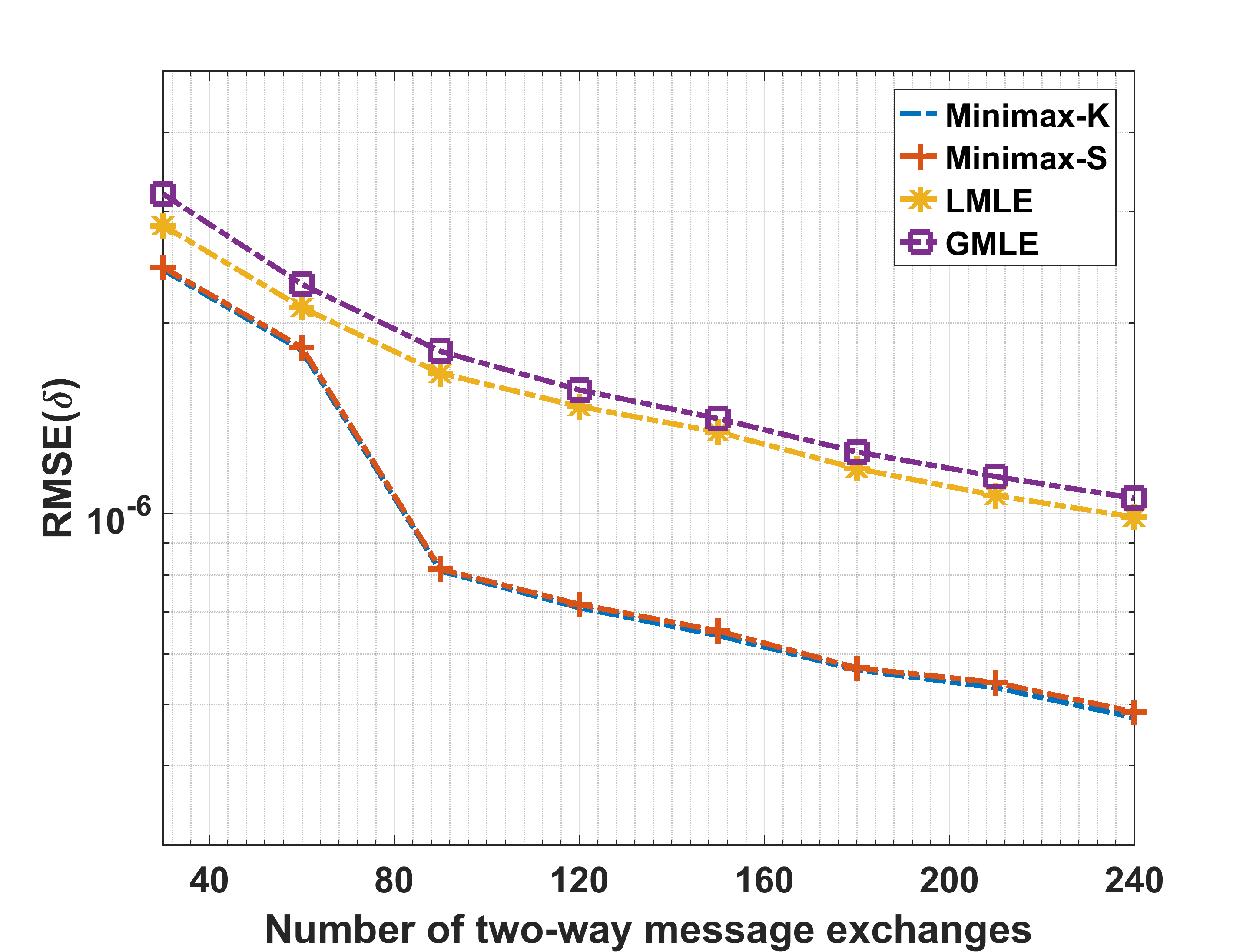}
		\caption{ }
	\end{subfigure}		
	\caption{RMSE of clock offset for various estimation schemes under TM-2 for various loads and $\{\phi, d, \delta \} = \{1, 2 \mu s, 2\mu s\}$, (a) 20\% load, (b) 40\% load, (c) 60\% load, (d) 80\% load.}\label{rmse_TM2_offset_results}
\end{figure}

\begin{figure}[t]
	\centering
	\begin{subfigure}[b]{0.48\columnwidth}
		\centering
		\includegraphics[height = 2.5 in, width = \columnwidth]{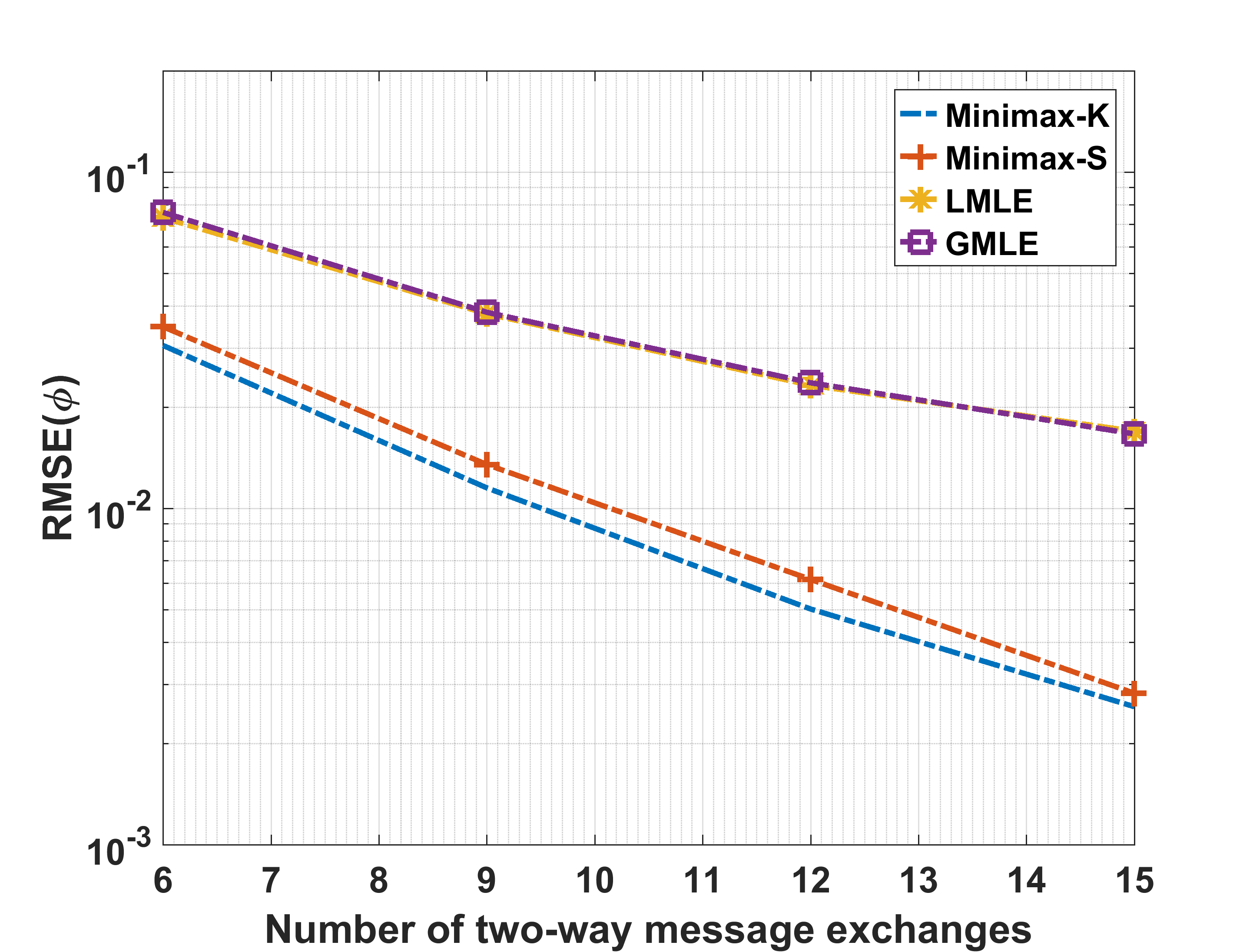}
		\caption{ }
	\end{subfigure}	
	~ %add desired spacing between images, e. g. ~, \quad, \qquad, \hfill etc. 
	%(or a blank line to force the subfigure onto a new line)
	\begin{subfigure}[b]{0.48\columnwidth}
		\centering
		\includegraphics[height = 2.5 in, width = \columnwidth]{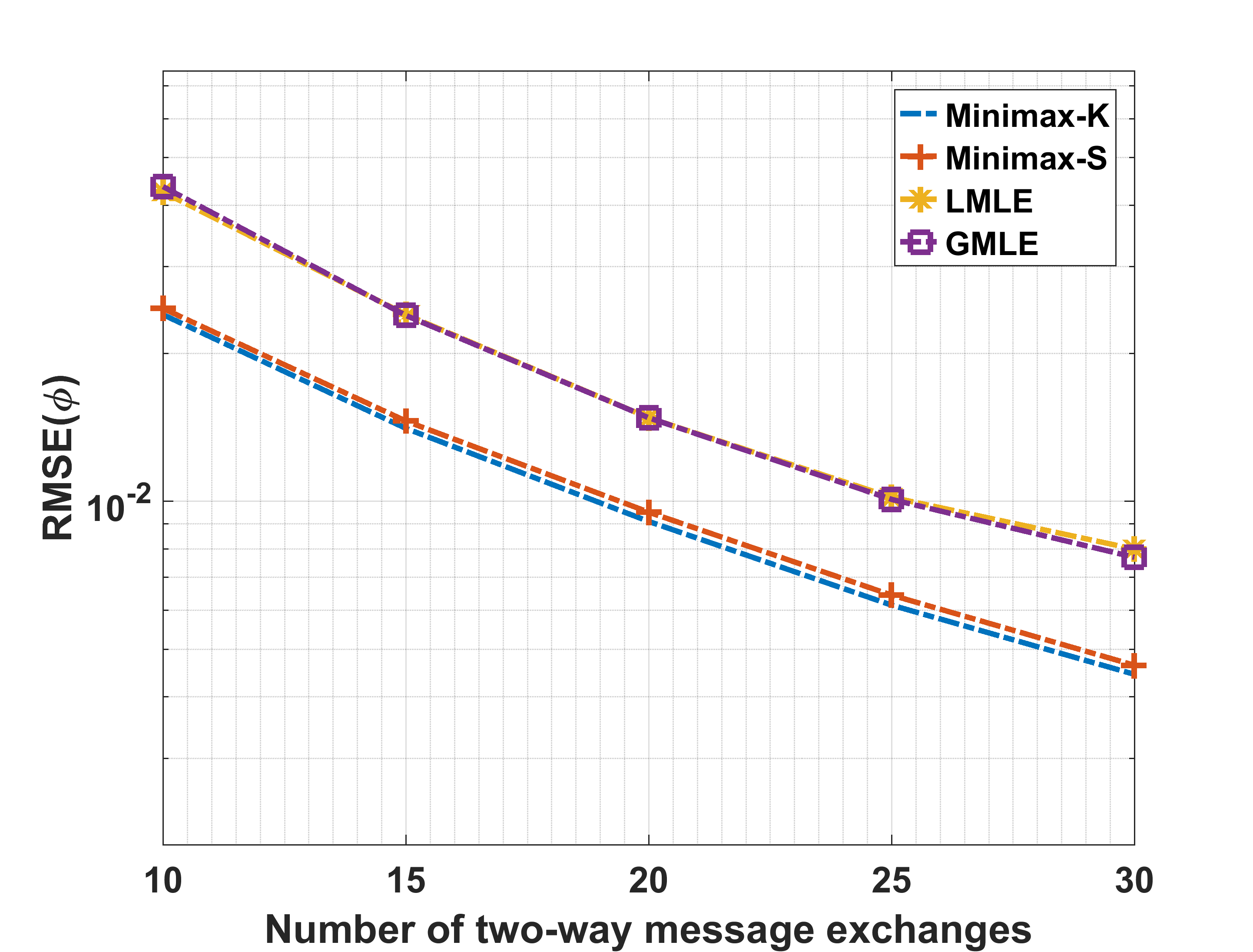}
		\caption{ }
	\end{subfigure}

	\begin{subfigure}[b]{0.48\columnwidth}
		\centering
		\includegraphics[height = 2.5 in, width = \columnwidth]{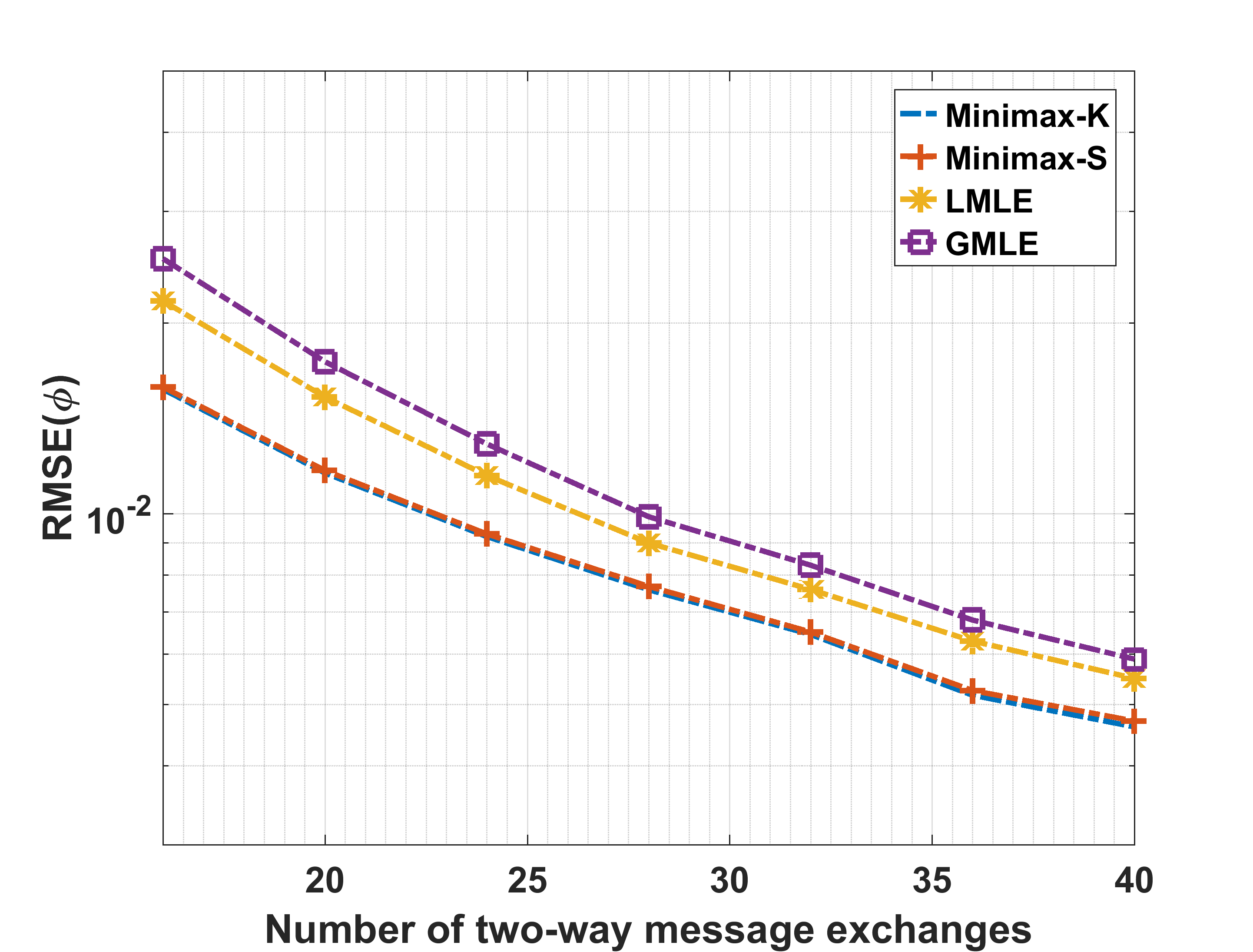}
		\caption{ }
	\end{subfigure}	
	~ %add desired spacing between images, e. g. ~, \quad, \qquad, \hfill etc. 
	%(or a blank line to force the subfigure onto a new line)
	\begin{subfigure}[b]{0.48\columnwidth}
		\centering
		\includegraphics[height = 2.5 in, width = \columnwidth]{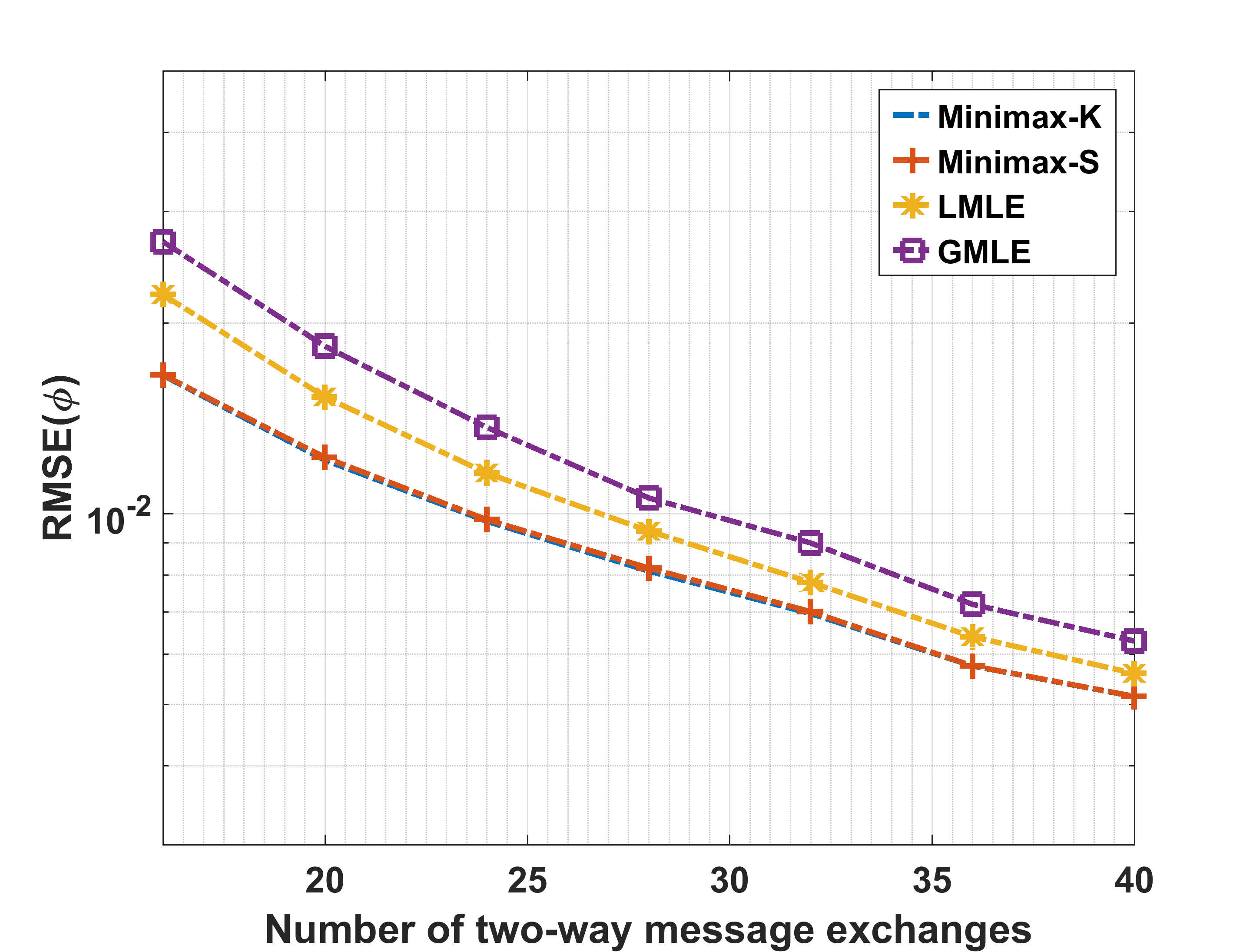}
		\caption{ }
	\end{subfigure}		
	\caption{RMSE of clock skew for various estimation schemes under TM-2 for various loads and $\{\phi, d, \delta \} = \{1, 2 \mu s, 2 \mu s\}$, (a) 20\% load, (b) 40\% load, (c) 60\% load, (d) 80\% load.}\label{rmse_TM2_skew_results}
\end{figure}

\begin{figure}[t]
	\centering
	\begin{subfigure}[b]{0.48\columnwidth}
		\centering
		\includegraphics[height = 2.5 in, width = \columnwidth]{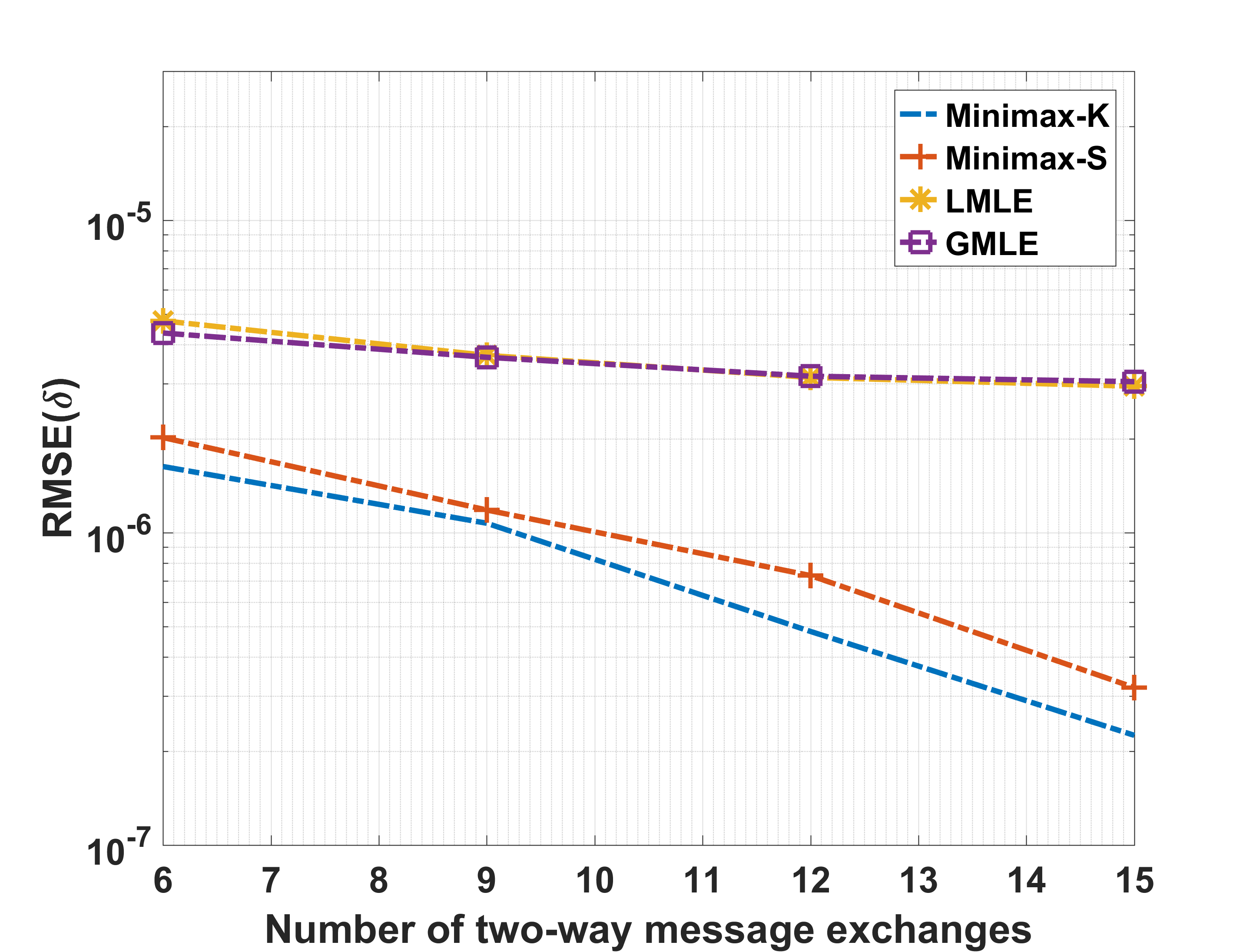}
		\caption{ }
	\end{subfigure}	
	~ %add desired spacing between images, e. g. ~, \quad, \qquad, \hfill etc. 
	%(or a blank line to force the subfigure onto a new line)
	\begin{subfigure}[b]{0.48\columnwidth}
		\centering
		\includegraphics[height = 2.5 in, width = \columnwidth]{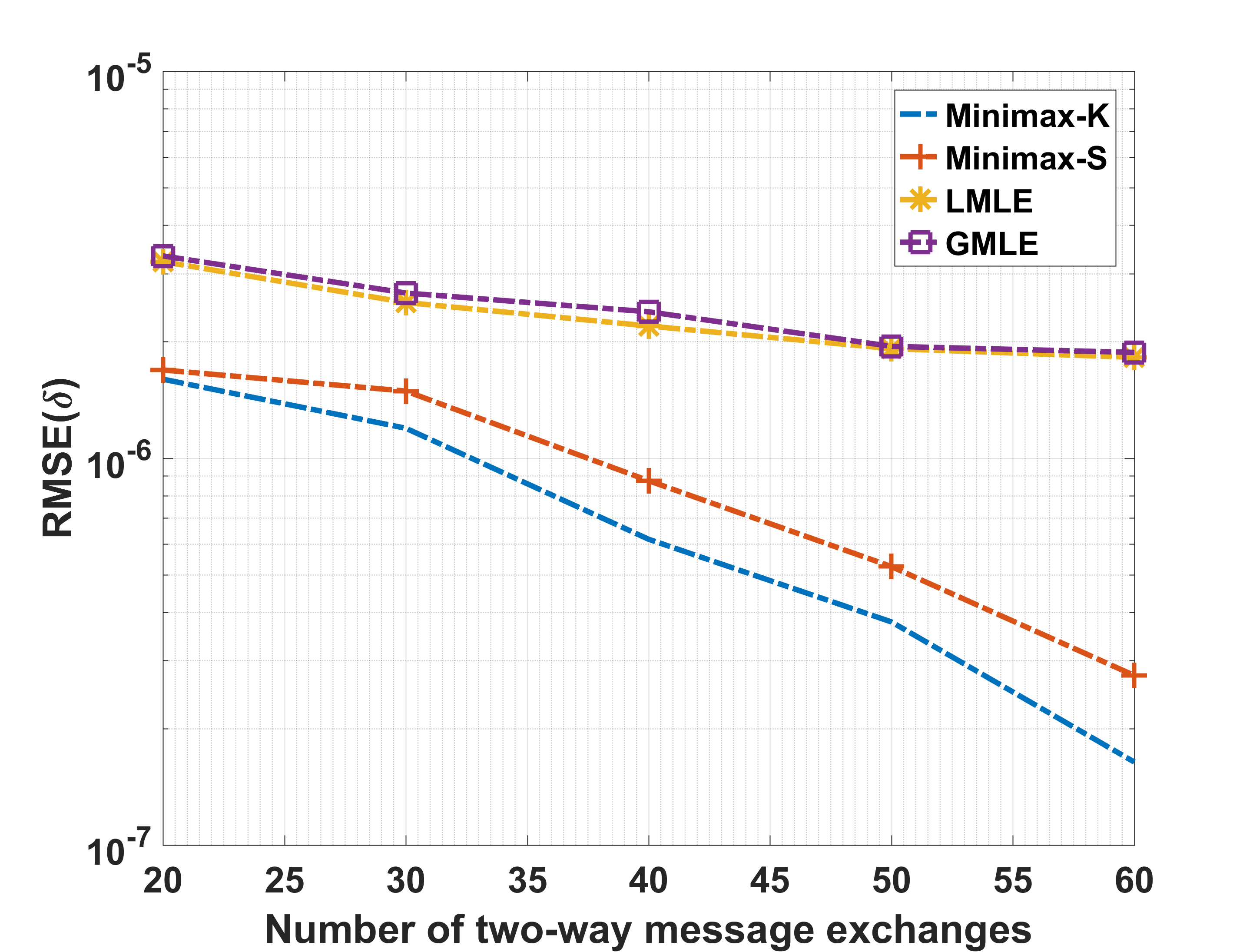}
		\caption{ }
	\end{subfigure}

	\begin{subfigure}[b]{0.48\columnwidth}
		\centering
		\includegraphics[height = 2.5 in, width = \columnwidth]{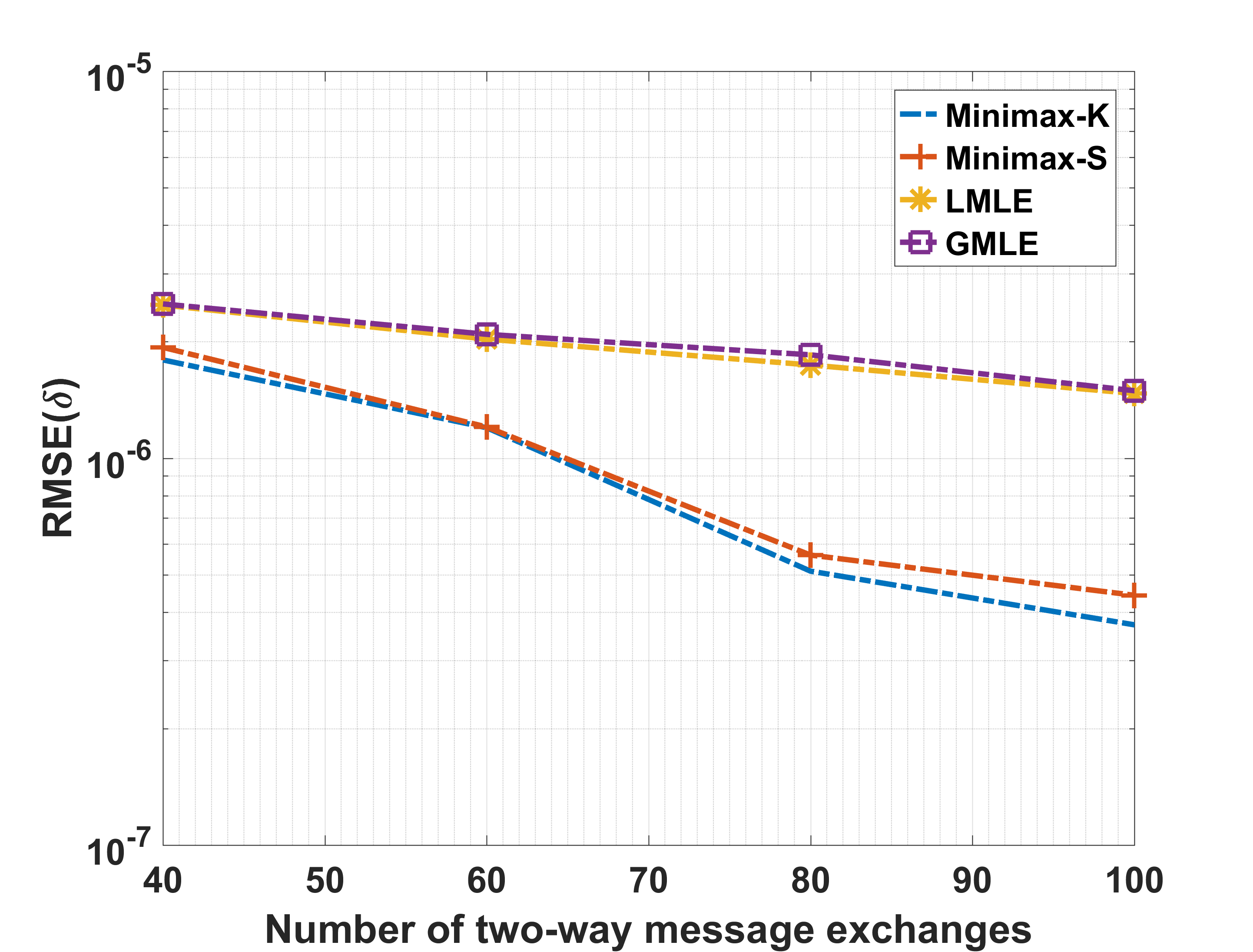}
		\caption{ }
	\end{subfigure}	
	~ %add desired spacing between images, e. g. ~, \quad, \qquad, \hfill etc. 
	%(or a blank line to force the subfigure onto a new line)
	\begin{subfigure}[b]{0.48\columnwidth}
		\centering
		\includegraphics[height = 2.5 in, width = \columnwidth]{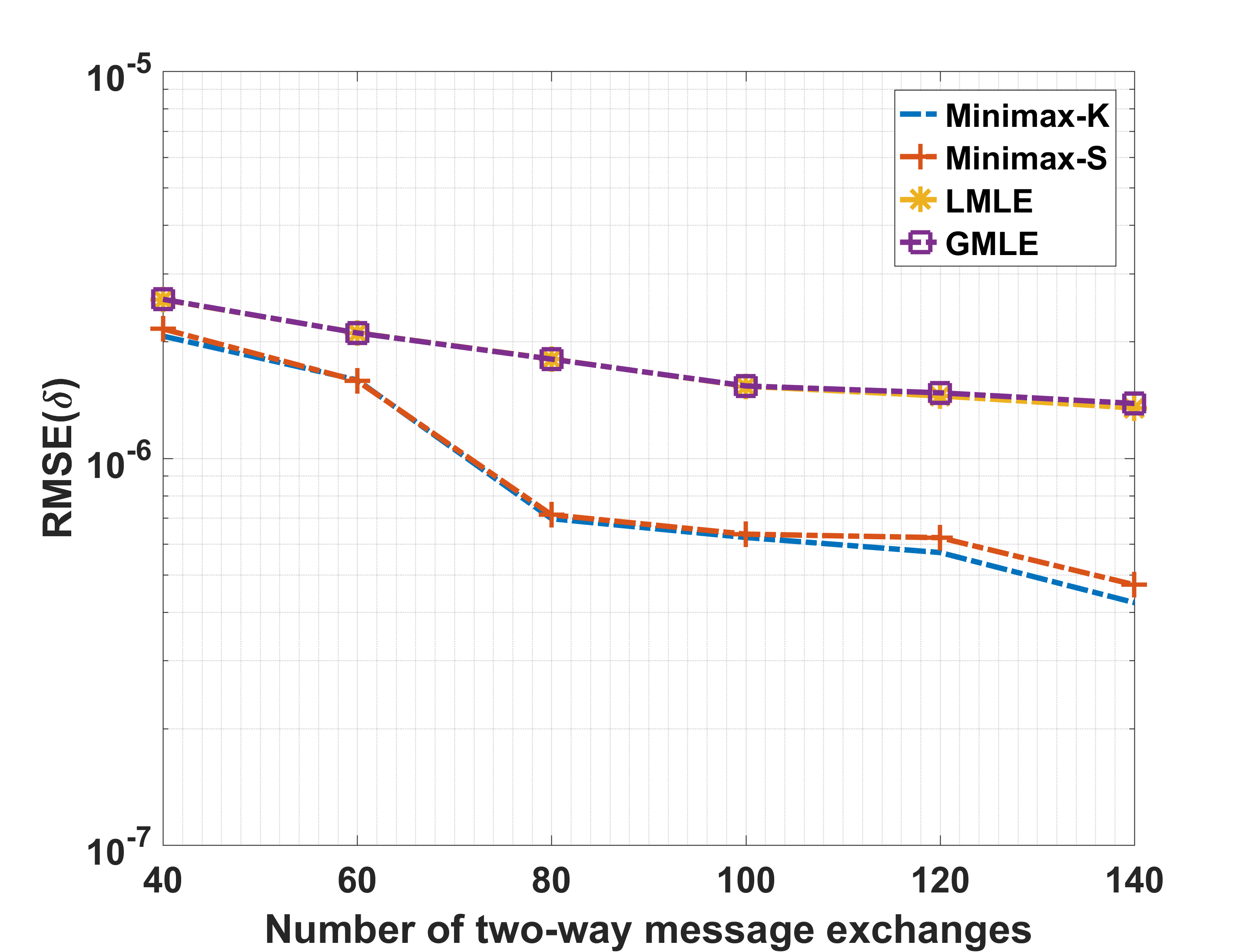}
		\caption{ }
	\end{subfigure}		
	\caption{RMSE of clock offset for various estimation schemes under EG-TM1 for various loads and $\{\phi, d, \delta \} = \{1, 2 \mu s, 2 \mu s\}$, (a) 20\% load, (b) 40\% load, (c) 60\% load, (d) 80\% load.}\label{rmse_TM1sg_offset_results}
\end{figure}

\begin{figure}[t]
	\centering
	\begin{subfigure}[b]{0.48\columnwidth}
		\centering
		\includegraphics[height = 2.5 in, width = \columnwidth]{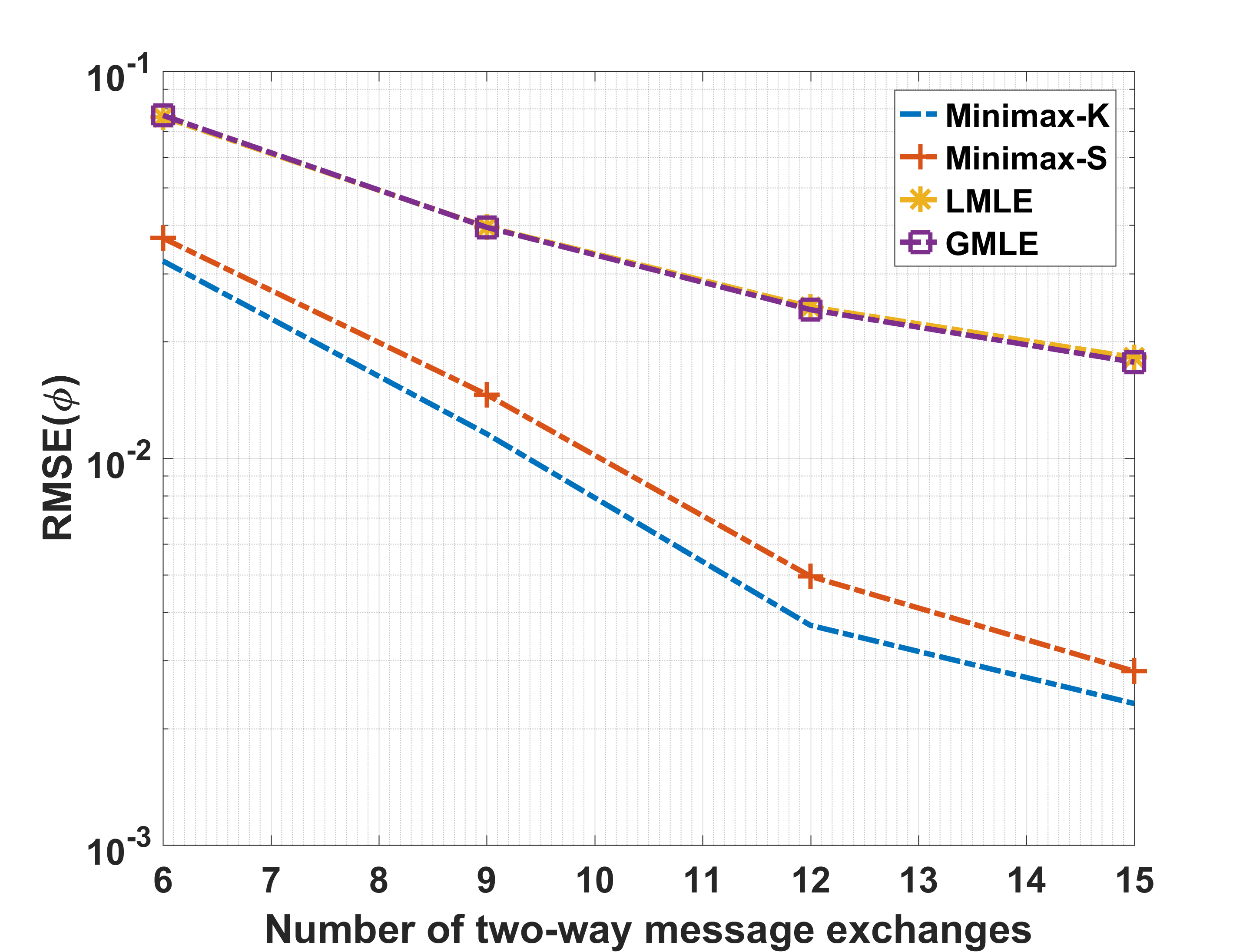}
		\caption{ }
	\end{subfigure}	
	~ %add desired spacing between images, e. g. ~, \quad, \qquad, \hfill etc. 
	%(or a blank line to force the subfigure onto a new line)
	\begin{subfigure}[b]{0.48\columnwidth}
		\centering
		\includegraphics[height = 2.5 in, width = \columnwidth]{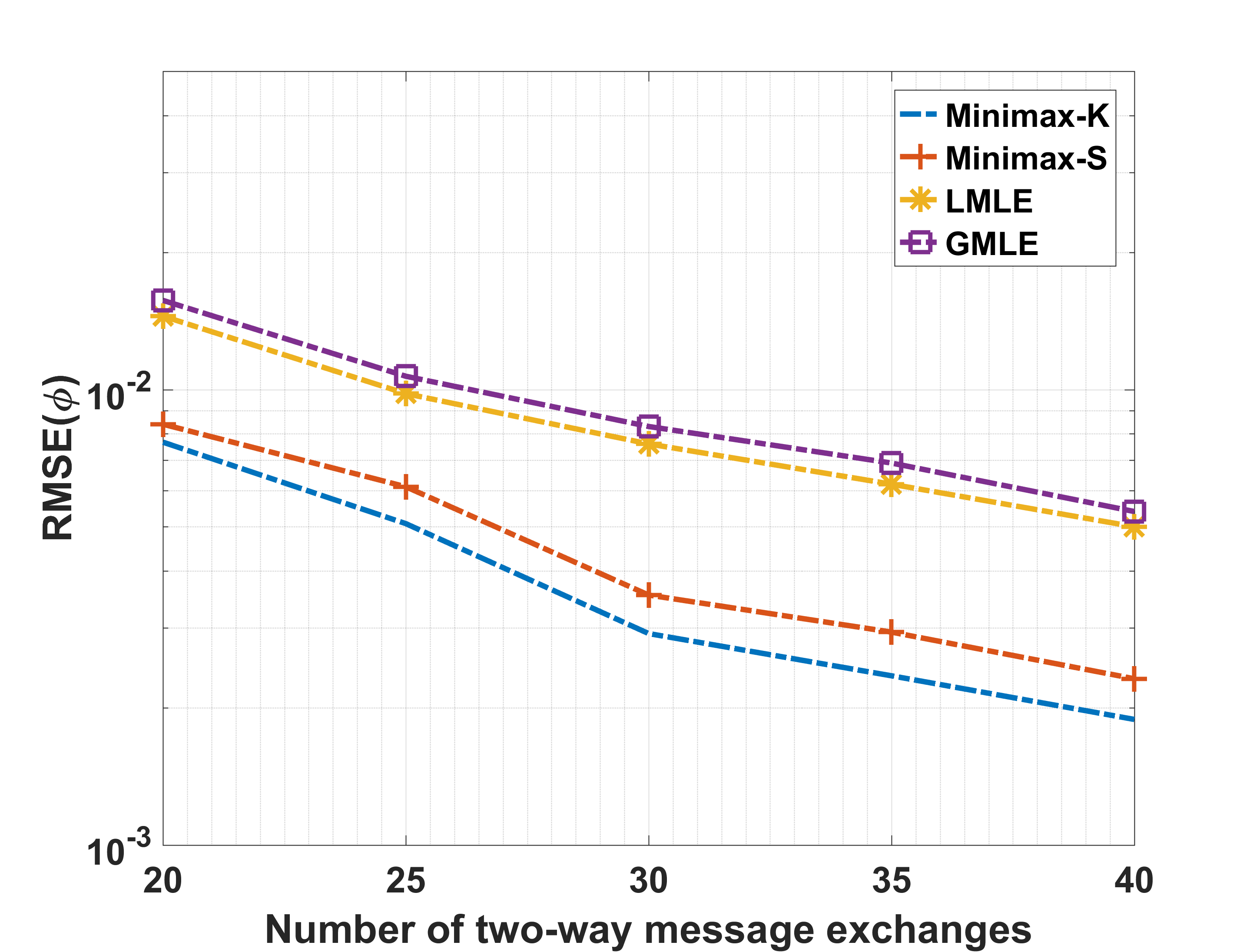}
		\caption{ }
	\end{subfigure}

	\begin{subfigure}[b]{0.48\columnwidth}
		\centering
		\includegraphics[height = 2.5 in, width = \columnwidth]{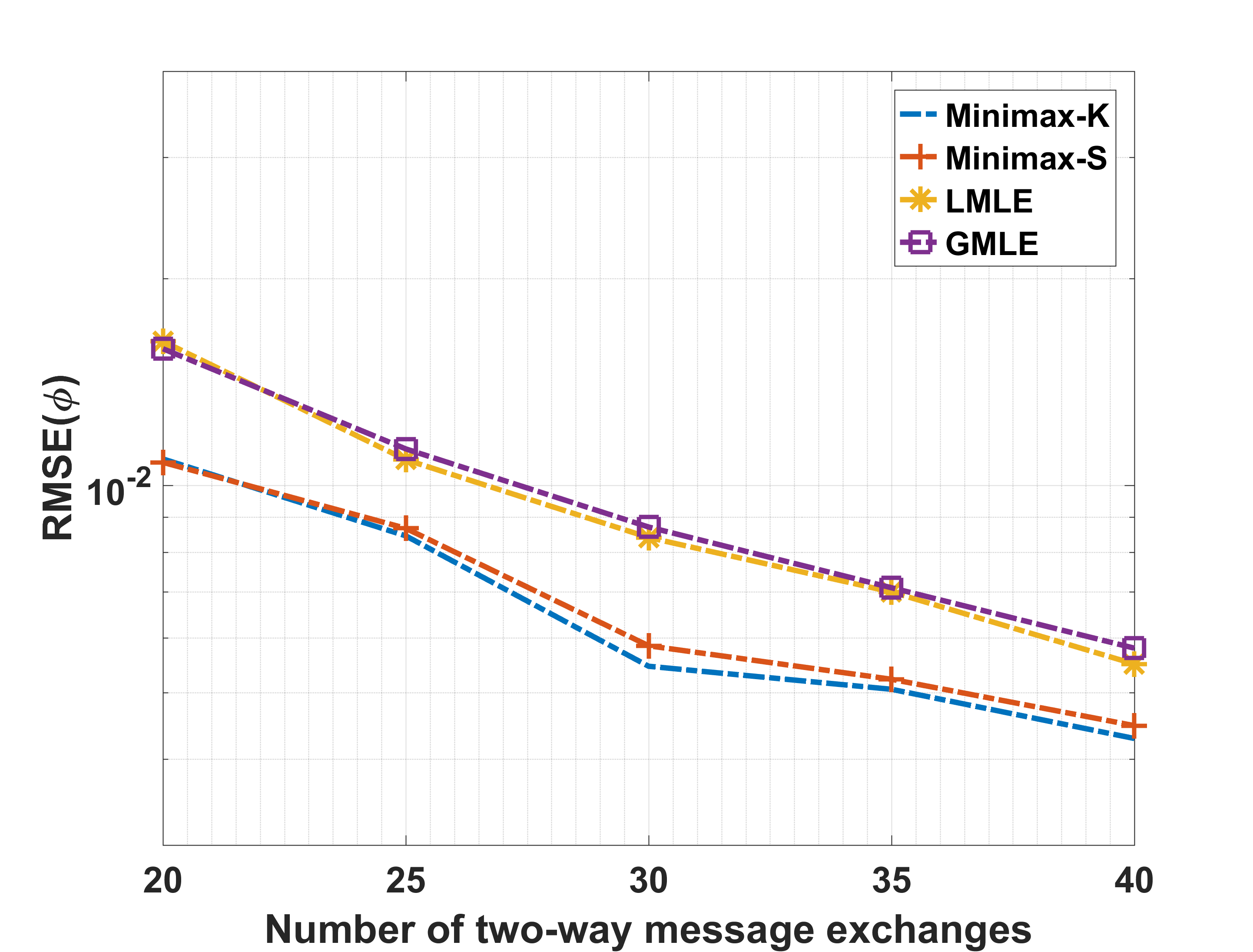}
		\caption{ }
	\end{subfigure}	
	~ %add desired spacing between images, e. g. ~, \quad, \qquad, \hfill etc. 
	%(or a blank line to force the subfigure onto a new line)
	\begin{subfigure}[b]{0.48\columnwidth}
		\centering
		\includegraphics[height = 2.5 in, width = \columnwidth]{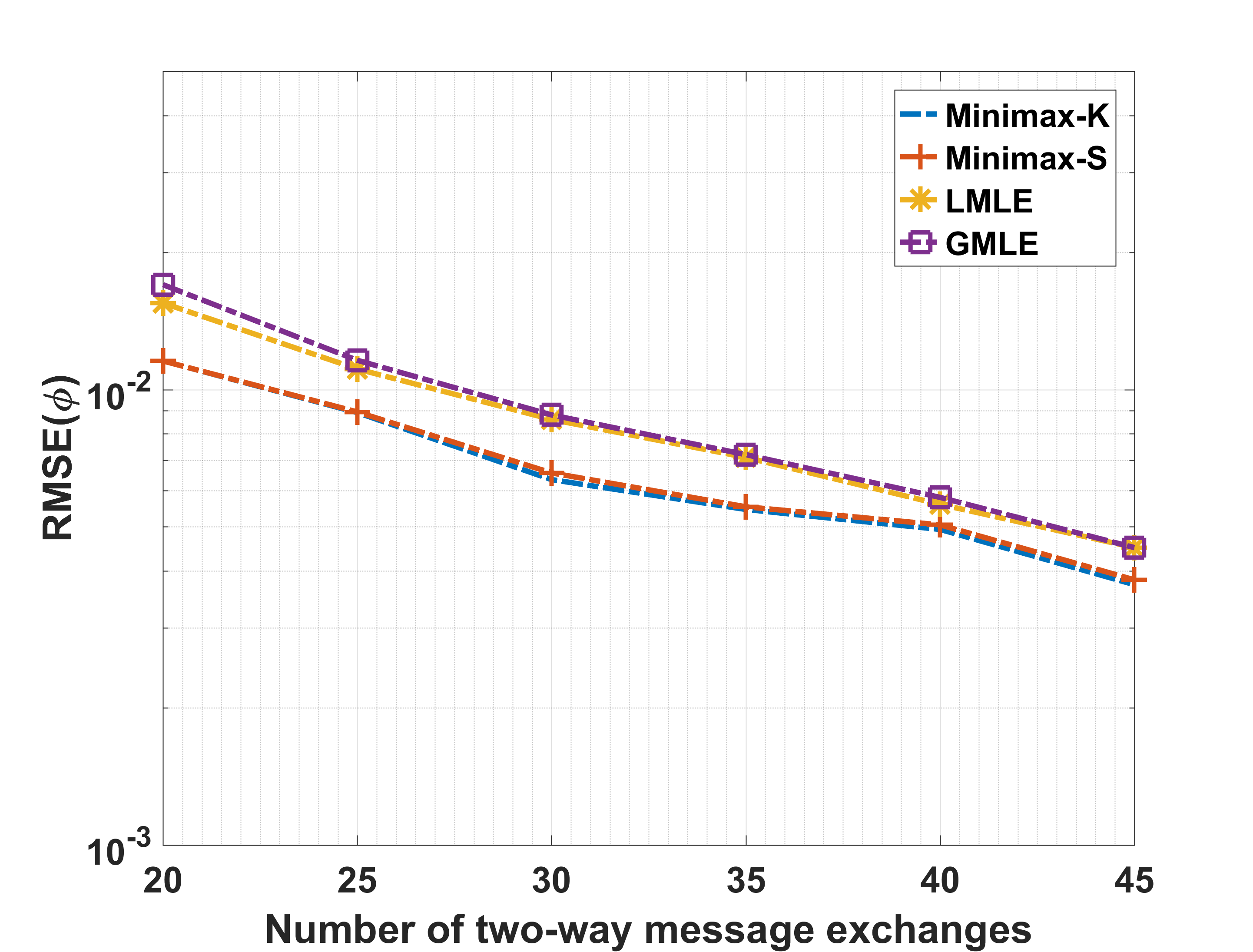}
		\caption{ }
	\end{subfigure}		
	\caption{RMSE of clock skew for various estimation schemes under EG-TM1 for various loads and $\{\phi, d, \delta \} = \{1, 2 \mu s, 2 \mu s\}$, (a) 20\% load, (b) 40\% load, (c) 60\% load, (d) 80\% load.}\label{rmse_TM1sg_skew_results}
\end{figure}

% you can choose not to have a title for an appendix
% if you want by leaving the argument blank
%\section{}
%% use section* for acknowledgement
%\section*{Acknowledgment}
%
%
%The authors would like to thank...

% Can use something like this to put references on a page
% by themselves when using endfloat and the captionsoff option.
\ifCLASSOPTIONcaptionsoff
  \newpage
\fi

% references section

% can use a bibliography generated by BibTeX as a .bbl file
% BibTeX documentation can be easily obtained at:
% http://www.ctan.org/tex-archive/biblio/bibtex/contrib/doc/
% The IEEEtran BibTeX style support page is at:
% http://www.michaelshell.org/tex/ieeetran/bibtex/

\bibliographystyle{IEEEtranTCOM}
% argument is your BibTeX string definitions and bibliography database(s)
\bibliography{IEEEabrv,refs}

% that's all folks
\end{document}